\documentclass[11pt]{article}
\usepackage[utf8]{inputenc}
\usepackage[left=1in, top=1in, right=1in, bottom=1in]{geometry}
\usepackage[authoryear, round]{natbib}
\usepackage{hyperref}
\usepackage{amsmath, amsthm, amsfonts}
\usepackage{mathtools}
\usepackage{thm-restate}
\usepackage{bm}
\usepackage{algorithm}
\usepackage{algpseudocode}
\usepackage{xcolor}
\usepackage{subcaption}
\usepackage{scrextend}

\newtheorem{theorem}{Theorem}
\newtheorem{proposition}[theorem]{Proposition}
\newtheorem{lemma}[theorem]{Lemma}
\newtheorem{definition}[theorem]{Definition}
\newtheorem{claim}[theorem]{Claim}

\newtheorem{corollary}[theorem]{Corollary}

\newtheorem{example}[theorem]{Example}

\DeclareMathOperator*{\argmax}{argmax}

\newcommand{\AutoAdjust}[3]{{ \mathchoice{ \left #1 #2  \right #3}{#1 #2 #3}{#1 #2 #3}{#1 #2 #3} }}
\newcommand{\Xcomment}[1]{{}}

\newcommand{\InBrackets}[1]{\AutoAdjust{[}{#1}{]}}% {\left[{#1}\right]}

\newcommand{\Ex}[2][]{{\mathbb E}_{#1}\InBrackets{#2}}
\newcommand{\Prx}[2][]{\Pr_{#1}\InBrackets{#2}}

\newcommand{\dd}{\mathrm{d}}  % for integrals

\newcommand{\eps}{\varepsilon}

\newcommand{\alg}{\textnormal{Alg}}

\newcommand{\given}{\;\mid\;}
\newcommand{\bigggiven}{\;\bigg|\;}
\newcommand{\Biggiven}{\;\Big|\;}

\title{Nash Convergence of Mean-Based Learning Algorithms in First-Price Auctions\footnote{A preliminary version of this paper was published at WWW 2022.}}
\author{
Xiaotie Deng\thanks{Center on Frontiers of Computing Studies, Peking University, email: {\tt xiaotie@pku.edu.cn}.
}
\and 
Xinyan Hu\thanks{Department of Electrical Engineering and Computer Sciences, University of California, Berkeley, email: {\tt xinyanhu@berkeley.edu}. Most of the work done at Center on Frontiers of Computing Studies, Peking University.}
\and
Tao Lin\thanks{School of Engineering and Applied Sciences, Harvard University, email: {\tt tlin@g.harvard.edu}}
\and
Weiqiang Zheng\thanks{Department of Computer Science, Yale University, email: {\tt weiqiang.zheng@yale.edu}}
}
\date{August, 2025}

\begin{document}

\maketitle

\begin{abstract}
The convergence properties of learning dynamics in repeated auctions is a timely and important question, with numerous applications in, e.g., online advertising markets. This work focuses on repeated first-price auctions where bidders with fixed values learn to bid using mean-based algorithms -- a large class of online learning algorithms that include popular no-regret algorithms such as Multiplicative Weights Update and Follow the Perturbed Leader. We completely characterize the learning dynamics of mean-based algorithms, under two notions of convergence: (1) \emph{time-average}: the fraction of rounds where bidders play a Nash equilibrium converges to 1; (2) \emph{last-iterate}: the mixed strategy profile of bidders converges to a Nash equilibrium. Specifically, the results depend on the number of bidders with the highest value:
\begin{itemize}
    \item If the number is at least three, the dynamics almost surely converges to a Nash equilibrium of the auction, in both time-average and last-iterate.  
    \item If the number is two, the dynamics almost surely converges to a Nash equilibrium in time-average but not necessarily last-iterate.
    \item If the number is one, the dynamics may not converge to a Nash equilibrium in time-average or last-iterate. 
\end{itemize}
Our discovery opens up new possibilities in the study of the convergence of learning dynamics.
\end{abstract}

\section{Introduction}
First-price auctions are the current trend in online advertising markets.
%Tell a story about how advertising system runs in reality.
A major example is the switch from second-price auctions to first-price auctions on Google's Ad Exchange platform in 2019~\citep{paes_leme_why_2020, goke_bidders_2022}. 

In contrast to second-price auctions, first-price auctions are non-truthful: bidders need to reason about other bidders' strategies and choose their own bidding strategies accordingly. % to maximize their utilities.
% Finding a good bidding strategy used to be a difficult task due to each bidder's lack of information of other bidders. 
% But given the repeated nature of online advertising auctions and with the advance of computing technology, nowadays' bidders are able to learn to bid using automated bidding algorithms \citep{aggarwal2024auto}. 
In today's online advertising auctions, a standard practice for bidders is to use data-driven learning algorithms to make bidding decisions automatically \citep{aggarwal2024auto}. 
% From the traditional perspective of auction theory \blue{finding the best bidding strategy is a difficult task} because of each bidder's lack of information of other bidders. \red{it is difficult to ?????}
% But nowadays, with the advance of computing technology and given the repeated nature of online advertising auctions, bidders are able to learn to bid using automated bidding algorithms. 
% As observed by \cite{nekipelov_econometrics_2015}, bidders' behavior on Bing is consistent with no-regret learning.
As one bidder adjusts bidding strategies using a learning algorithm, other bidders' payoffs are affected and thus they will adjust their strategies as well. 
%, which forms a learning dynamics. Then, a natural question follows: \emph{What is the outcome of such learning dynamics?  In particular, will it converge to a Nash equilibrium of the auction?} 
Then, a natural question follows: \emph{
%if all bidders in a repeated first-price auction use some learning algorithms to adjust bidding strategies at the same time, will they converge to a Nash equilibrium of the auction?
Do multi-agent learning dynamics in repeated first-price auctions converge to a Nash equilibrium? 
} 
This question is of both theoretical and practical interest.
%From the theoretical perspective, understanding the convergence of learning dynamics is fundamental in the area of learning in games.
Understanding the Nash convergence properties of multi-agent learning dynamics is fundamental in the theory of learning in games,
and it provides useful predictions on the bidders' welfare and platforms' revenue in online auction systems with learning agents. 

% In real-world advertising markets, the convergence to Nash equilibria in repeated first-price auctions provides guarantees on the social welfare and the platform's revenue. 

An early result on the convergence of multi-agent learning in repeated first-price auctions is given by \citet{hon-snir_learning_1998}. Considering repeated first-price auctions where bidders have fixed values for the item, \citep{hon-snir_learning_1998} shows that a Nash equilibrium may or may not be learned by the \emph{Fictitious Play} algorithm, where in each round of auctions every bidder best responds to the empirical distributions of other bidders' bids in history. 
Fictitious Play, however, is a deterministic algorithm that does not have the \emph{no-regret} property ---
%--- a desideratum for learning algorithms in adversarial environments --- which can be obtained only by randomized algorithms \citep{roughgarden_lecture_2016}.
a desideratum for learning algorithms in adversarial environments.
The no-regret property can only be obtained by randomized algorithms \citep{roughgarden_lecture_2016}.
% Given evidences that agents may use no-regret algorithms to bid in Internet ad auctions~\citep{nekipelov_econometrics_2015},
As observed by \cite{nekipelov_econometrics_2015} that bidders' behavior on Bing's advertising system is consistent with no-regret learning, 
it is hence important, from both theoretical and practical points of view, to understand the convergence property of randomized no-regret learning algorithms in repeated first-price auctions.  %This is the focus of our work.  %This is the subject of our paper, which is the highlight/bottom line/goal/emphasis/motivation/underline of our work,   which is the algorithms we focus on.
This motivates our work.

\paragraph{\bf Our contributions.}
Focusing on repeated first-price auctions where bidders have fixed values, we completely characterize the Nash convergence property of a wide class of randomized online learning algorithms called ``mean-based algorithms'' \citep{braverman_selling_2018}.
This class contains many popular no-regret algorithms, including Multiplicative Weights Update (MWU), Follow the Perturbed Leader (FTPL), EXP-3, etc.

We systematically analyze two notions of Nash convergence:
(1) \emph{time-average}: the fraction of rounds where bidders play a Nash equilibrium approaches 1 in the limit; (2) \emph{last-iterate}: the mixed strategy profile of bidders approaches a Nash equilibrium in the limit. 
%\red{Note that much work has been about the time-average convergence while people pay more and more attention to the last-iterate convergence but little is studied. We analyze from both perspectives.}
Specifically, the results turn out to depend on the number of bidders with the highest value:
\begin{itemize}
\item If the number is at least three, the learning dynamics of mean-based algorithms almost surely converges to Nash equilibrium, both in time-average and in last-iterate.  
\item If the number is two, the learning dynamics almost surely converges to Nash equilibrium in time-average but not necessarily in last-iterate. 
\item 
If the number is one, the learning dynamics may not converge to Nash equilibrium in time-average nor in last-iterate. 
\end{itemize} 

%\red{Add a few sentences to describe our experimental results.}
For the last case, the non-convergence result is proved for the Follow the Leader algorithm, which is a mean-based algorithm that is not necessarily no-regret.  We also show by experiments that no-regret mean-based algorithms such as MWU and $\eps_t$-Greedy may not last-iterate converge to a Nash equilibrium.

\paragraph{\bf Intuitions and techniques.}
The intuition behind our convergence results (the first two cases above) relates to the notion of ``iterative elimination of dominated strategies'' in game theory.  Suppose there are three bidders all having the same integer value $v$ for the item and choosing bids from the set $\{0, 1, \ldots, v-1\}$.  The unique Nash equilibrium is all bidders bidding $v-1$.  The elimination of dominated bids is as follows: firstly, bidding $0$ is dominated by bidding $1$ for each of the three bidders no matter what other bidders bid, so bidders will learn to bid $1$ or higher instead of bidding $0$ at the beginning; then, given that no bidders bid $0$, bidding $1$ is dominated by bidding $2$, so all bidders learn to bid at least $2$; ...; in this way all bidders learn to bid $v-1$.
This logic is implicit in \cite{hon-snir_learning_1998}, but their formal argument only works for deterministic algorithms such as Fictitious Play, not for randomized algorithms.  

Formally proving that randomized mean-based learning algorithms are able to iteratively eliminate dominated strategies is non-trivial because those algorithms can pick a dominated strategy with a small but positive probability. 
To overcome this obstacle, we use and generalize a technique from \cite{feng_convergence_2021}, who show that bidders in a second-price auction with multiple Nash equilibria converge to the truthful equilibrium if they use mean-based algorithms with an initial uniform exploration stage.
Their argument relies on the fact that, in a second-price auction, all bidders learn to eliminate the non-truthful bids with high probability during the uniform exploration stage.
Then, by using a ``time partitioning'' technique combined with martingale concentration inequalities, \citet{feng_convergence_2021} prove that the bidders will continue to play non-truthful bids with small probability after the exploration stage. 
We significantly generalize their technique in two aspects: (1) We allow arbitrary mean-based algorithms without an initial uniform exploration stage. (2) More importantly, while \citet{feng_convergence_2021}'s technique only works for one phase of elimination of dominated strategies, we further develop their technique to prove that mean-based learning algorithms are able to iteratively eliminate dominated strategies for any number of phases. 
This allows us to prove that mean-based learning algorithms converge to Nash equilibrium in repeated first-price auctions.

\subsection{Discussion} \label{sec:discussion}

\paragraph{\bf The fixed value assumption.}
Our work assumes that each bidder has a fixed value for the sold item throughout the repeated auction. Under this assumption, the first-price auction is known to be equivalent to Bertrand competition~\citep{bertrand1883review} where firms have fixed (but possibly heterogeneous) production costs~\citep{bichler_online_2024}. Therefore, our results also imply the (non-)convergence to the Nash equilibrium when firms set prices using mean-based learning algorithms in Bertrand competition games.
%Seemingly restrictive, this assumption captures the classic Bertrand competition~\citep{bertrand1883review} and

The fixed-value assumption is also common in the literature on repeated auctions, in various contexts including value inference~\citep{nekipelov_econometrics_2015}, dynamic pricing~\citep{amin_learning_2013, devanur_perfect_2015, immorlica_repeated_2017}, and bidding equilibrium~\citep{hon-snir_learning_1998, iyer_mean_2014, kolumbus_auctions_2022, banchio_artificial_2022}. An exception is the work by \cite{feng_convergence_2021}, who study repeated first-price auctions under the Bayesian assumption that bidders' values are i.i.d.~samples from a distribution.  However, their result is restricted to a $2$-symmetric-bidder setting with the Uniform$[0, 1]$ distribution where the Bayesian Nash equilibrium (BNE) is simply every bidder bidding half of their values.
%\red{, which can be learned by uniform exploration stage similarly}.
For general asymmetric distributions, there is no explicit characterization of the BNE \citep{lebrun_existence_1996, lebrun_first_1999, maskin_equilibrium_2000} despite the existence of (inefficient) numerical approximations~\citep{fibich_asymmetric_2003, escamocher_existence_2009, wang_bayesian_2020}. Moreover, since computing BNE in first-price auctions is known to be computationally intractable~\citep{filos2021complexity, chen_complexity_2023, filos2024computation}, the existence of simple, convergent learning dynamics appears unlikely in the Bayesian setting with general value distributions.
%No algorithms are known to be able to compute BNE efficiently for all asymmetric distributions, let alone a simple, generic learning algorithm.

% Moreover, in real-life auctions, fixed values do occur if a same item is sold repeatedly, bidders have stable values for that item, and the set of bidders is fixed.
% An example is a few large online travel agencies (Agoda, Airbnb, and Booking.com) competing for an ad slot about ``hotel booking''.
% In such Internet advertising auction scenarios, a large number of auctions can happen in a short amount of time, which allows learning bidders to converge to the Nash equilibrium quickly. 

Moreover, as we will show, even with the seemingly innocuous assumption of fixed values, the learning dynamics of mean-based algorithms already exhibit complicated behaviors: it may converge to different equilibria in different runs or not converge at all.  One can envision more unpredictable behaviors when values are not fixed.

\paragraph{\bf Learning in general games.} 
Our work is related to a fundamental question in the field of Learning in Games \citep{fudenberg_theory_1998, cesa-bianchi_prediction_2006, nisan_algorithmic_2007}: if players in a repeated game employ online learning algorithms to adjust strategies, will they converge to an equilibrium? And what kinds of equilibrium? 
% Although the answer to this question is ``no'' in general (see Related Works for details), positive results do exist for some special cases of online learning algorithms and games: for example, no-regret learning algorithms provably converge to Nash equilibria in zero-sum games, $2\times 2$ games, and routing games (see e.g., \citealp{fudenberg_theory_1998, cesa-bianchi_prediction_2006, nisan_algorithmic_2007}). %}
Classical results include the convergence of no-regret learning algorithms to a \emph{Coarse Correlated Equilibrium (CCE)} and no-internal-regret algorithms to a \emph{Correlated Equilibrium} in \emph{any} game \citep{foster_calibrated_1997, hart_simple_2000}.
But given that (coarse) correlated equilibria are much weaker than the archetypal solution concept of a Nash equilibrium, a more appealing and challenging question is the convergence towards a Nash equilibrium.
Positive answers to this question are only known for some special cases of algorithms and games: e.g., no-regret algorithms converge to Nash equilibria in zero-sum games, $2\times 2$ games, and routing games \citep{fudenberg_theory_1998, cesa-bianchi_prediction_2006, nisan_algorithmic_2007}.
In contrast, several works give non-convergence examples: e.g., the non-convergence of MWU in a $3\times 3$ game~\citep{daskalakis_learning_2010} and Regularized Learning Dynamics in zero-sum games~\citep{mertikopoulos_cycles_2018}.
%, and Follow the Regularized Leader to unstrict Nash equilibria \citep{flokas2020noregret}. 
In this work we study the Nash equilibrium convergence property in first-price auctions for a large class of learning algorithms, namely the mean-based algorithms, and provide both positive and negative results.

\paragraph{\bf Last v.s.~average iterate convergence.}  
% We emphasize that previous results on convergence of learning dynamics in games are attained mostly in an average sense, i.e., the empirical frequency of player's actions converges. 
We emphasize that previous results on convergence of learning dynamics to Nash equilibria in games are mostly attained in an average sense, i.e., the empirical distributions of players' actions converge. 
Our notion of time-average convergence, which requires players to play a Nash equilibrium in almost every round, is different from the convergence of empirical distributions; in fact, ours is stronger if the Nash equilibrium is unique.
% This notion of convergence is implied by our definition of time-average convergence when there is a unique Nash equilibrium.
% Nevertheless, time-average convergence fails to capture the full picture of the dynamics since it has no guarantee on whether players' last-iterate (mixed) strategy profile converges to a Nash equilibrium.
Nevertheless, time-average convergence fails to capture the full picture of the dynamics since players' last-iterate (mixed) strategy profile may not converge.
%Guarantee of convergence in an average sense does not capture the full picture of the dynamics. 
% Many results about last-iterate convergence show that dynamics diverges or enters a limit cycle even for a simple $3\times 3$ game~\citep{daskalakis_learning_2010} or zero-sum games~\citep{mertikopoulos_cycles_2018}.
% Little is known about last-iterate convergence of learning dynamics beyond gradient-based algorithms in two-player zero-sum games~\cite{daskalakis_last-iterate_2018,wei_linear_2021}.
Existing results about last-iterate convergence show that most of the learning dynamics actually diverge or enter a limit cycle even in a simple $3\times 3$ game~\citep{daskalakis_learning_2010} or zero-sum games~\citep{mertikopoulos_cycles_2018}, except for a few convergence examples like optimistic gradient descent/ascent in two-player zero-sum games or monotone games~\citep{daskalakis_last-iterate_2018, wei_linear_2021, cai2022finitetime}. 
% Our convergence results, built on systematical techniques, shed light on further analysis of convergence and learning in games. 
Our results and techniques, regarding the convergence of any mean-based algorithm in first-price auctions, shed light on further study of last-iterate convergence in more general settings. 

%Recently, inspired by the training of Generative Adversarial Networks,
%Last-iterate convergence results are only attained in special cases such as two-player zero-sum games with gradient-based algorithms [].

\subsection{Additional Related Works}
\paragraph{\bf Online learning in auctions.}
% We first review additional related works about online learning in repeated auctions.\footnote{We do not review works about the \emph{batch learning} setting, e.g., sample complexity.}
A large fraction of existing works on online learning in repeated auctions are from the \emph{seller}'s perspective, i.e., studying how a seller can maximize revenue by adaptively changing the rules of the auction (e.g., reservation price) over time (e.g., \cite{blum_near-optimal_2005, amin_learning_2013, mohri_optimal_2014, cesa-bianchi_regret_2015, braverman_selling_2018, huang_learning_2018, abernethy_learning_2019, kanoria_incentive-compatible_2019, deng_game-theoretic_2020, golrezaei_dynamic_2021}).
We focus on the \emph{bidders'} learning problem. 
% The other line of works study the seller's online learning problem: how a seller can maximize revenue by adaptively changing the rules of the auction (e.g., reservation price) over time, under different assumptions on bidders' behavior (e.g., \cite{blum_near-optimal_2005, amin_learning_2013, mohri_optimal_2014, cesa-bianchi_regret_2015, braverman_selling_2018, huang_learning_2018, abernethy_learning_2019, kanoria_incentive-compatible_2019, deng_game-theoretic_2020, golrezaei_dynamic_2021}). 
% In particular, \cite{braverman_selling_2018} show that, by first giving the item for free and then charging a high price, a seller in a single-bidder auction can extract full surplus from the bidder if the bidder uses mean-based algorithms.  
% We study the same class of algorithms, but focus on bidders' learning problem in a fixed auction with more than one bidder.  
% 

Existing works from bidders' perspective are mostly about ``learning to bid'', studying how to design no-regret algorithms for a bidder to bid in various formats of repeated auctions, including first-price auctions \citep{balseiro_contextual_2019, han2020learning, zhang2022leveraging, badanidiyuru_learning_2023, wang2023learning, han_optimal_2025, galgana2025learning}, second-price auctions \citep{iyer_mean_2014, weed_online_2016}, and more general auctions \citep{feng_learning_2018, karaca2020no}.
%{\color{red} If the format of auction can be changed by the seller over time, \cite{braverman_selling_2018} show that the seller can extract full surplus from a buyer who follows a certain class of no-regret learning algorithms. }
Those works take the perspective of a \emph{single} bidder, without considering the interaction among \emph{multiple} bidders, all learning to bid at the same time.  % We study the consequence of such intera
We instead study the consequence of such interaction, showing that the learning dynamics of multiple bidders may or may not converge to the Nash equilibrium of the auction. 

\paragraph{\bf Multi-agent learning in first-price auctions.}
In addition to the aforementioned works by \cite{feng_convergence_2021} and \cite{hon-snir_learning_1998}, we review some other works on multi-agent learning in first-price auctions.
%include, e.g., several empirical works by
Experimentally, \citep{bichler_learning_2021, goke_bidders_2022, banchio_artificial_2022} observe the convergence for some learning algorithms. 
Theoretically, 
\citep{kolumbus_auctions_2022} proves that in repeated first-price auctions with two mean-based learning bidders, \emph{if} the dynamics converge to some limit, then this limit must be a CCE in which the bidder with the higher value submits bids that are close to the lower value.  However, they do not provide conditions under which the dynamics converge.
We prove that the dynamics converge if the two bidders have the same value and in fact converge to the stronger notion of a Nash equilibrium.
% If the two bidders have different values, \cite{kolumbus_auctions_2022} suggest that the MWU algorithm may converge while we experimentally demonstrate that other mean-based algorithms like $\eps$-Greedy may not converge (at least in the last-iterate sense).
% remains unclear.
Our results complement \cite{kolumbus_auctions_2022} and support the aforementioned experimental findings.

\paragraph{\bf Learning to iteratively eliminate dominated strategies.}
To our knowledge, we are the first to formally prove that mean-based learning algorithms are able to iteratively eliminate dominated strategies in repeated games.
Although this result seems intuitive, the proof is involved due to the randomness of mean-based algorithms.  As mentioned in the Introduction, we generalize a ``time-partitioning'' technique in \cite{feng_convergence_2021} to overcome this technical difficulty. 
Some works on multi-agent learning in other games \citep{wu_multi-agent_2022, feng_peer_2022} have also observed but did not formally prove this result.
% \tao{I added a citation to the work that follows us. }
% Interestingly, though, \cite{wu_multi-agent_2022} note that mean-based algorithms need an \emph{exponential} time to iteratively eliminate all dominated strategies in some special games, while \cite{wang_learning_2023} develop a polynomial-time algorithm that is not mean-based. 
% The exact convergence rate of mean-based algorithms in first-price auctions remains unknown.

\paragraph{\bf Overview of subsequent works.}
% After the conference version of our paper \citep{deng_nash_2022}, there are several works on the convergence of learning dynamics in first-price auctions.
We review several works on learning dynamics in first-price auctions that appeared after our conference paper \citep{deng_nash_2022}. 
Paralleling our results for mean-based learning algorithms and using different techniques, \citet{ahunbay2025semicoarse} show the last-iterate convergence of the non-mean-based projected gradient descent algorithm~\citep{zinkevich2003online} in first-price auctions. 
Following our work, \citet{bichler_online_2024} generalize our techniques and prove that mean-based algorithms can iteratively eliminate dominated strategies in more general games.

\paragraph{\bf Organization of the paper.}
We discuss the model and preliminaries in Section~\ref{sec:model} and present our main results in Section~\ref{sec:main-result}.  
Section~\ref{sec:main-proof} presents the proof of Theorem~\ref{thm:main-M1-3}, which covers the main ideas and proof techniques of all our convergence results.
% Section~\ref{sec:conclusion} concludes and discusses future directions.
Section~\ref{sec:body-exp} includes experimental results.
We conclude and discuss future directions in Section \ref{sec:conclusion}. 
Missing proofs from Section~\ref{sec:main-result} and~\ref{sec:main-proof} are in Appendix~\ref{app:main-result-proof} and~\ref{app:main-proof} respectively.    
% The proofs of other main results are in Appendix~\ref{app:main-result-proof} and the missing proofs in this section are in Appendix~\ref{app:main-proof}.

\section{Model and Preliminaries}
\label{sec:model}
\paragraph{\bf Repeated first-price auctions.} 
%%% =========================  WWW version ========================= 
We consider repeated first-price sealed-bid auctions where a single seller sells a good to a set of $N\ge 2$ players (bidders) $\mathcal N=\{1, 2, ..., N\}$ for infinitely many rounds.  
Each player $i\in\mathcal N$ has a fixed private value $v^i$ for the good throughout.
%As discussed, the assumption of fixed value captures the scenario where players have stable values for a fixed good and the set of players is unchanged over time. 
See Section~\ref{sec:discussion} for a discussion on this assumption. 
We assume that $v^i$ is a positive integer in some range $\{1, \ldots, V\}$ where $V$ is an upper bound on $v^i$.  Assume $V\ge 3$. 
No player knows other players' values.
Without loss of generality, assume $v^{1}\ge v^{2}\ge \cdots \ge v^{N}$.
%===================================================================

% ========================= NeurIPS version =======================
% We consider repeated first-price sealed-bid auctions where a single seller sells a good to a set of $N\ge 2$ players (bidders) $\mathcal N=\{1, 2, ..., N\}$ for infinite rounds.  
% Each player $i$ has a fixed type (value) $v^i$ for the good throughout all rounds.
% We assume $v^i$ is a positive integer in some range $\{1, \ldots, V\}$ where $V$ is an upper bound on $v^i$.  Suppose $V\ge 3$. 
% No player knows other players' values.
% Without loss of generality, assume $v^{1}\ge v^{2}\ge \cdots \ge v^{N}$.

% While some works on repeated auctions consider cases where players' types vary over time (e.g., \citet{feng_convergence_2021} take the Bayesian assumption that types are i.i.d.~drawn
% from a distribution at every round), it is not uncommon in the literature to assume constant types as we do (see e.g., \citealp{hon-snir_learning_1998, devanur_perfect_2015, immorlica_repeated_2017}).  
% The assumption of constant types captures the scenario where players have relatively stable values and the set of players is unchanged over time.  
% We could assume that $v^i$ is generated by sampling from some distribution at the beginning of the repeated auctions, although how $v^i$ is actually generated is unimportant to our result.
% % {\color{red} where the prior is also unknown to bidders.} 

%=====================================================================

At each round $t\ge 1$ of the repeated auctions, each bidder $i$ submits a bid $b^i_t \in \{0, 1, \ldots, V\}$ to compete for the good.  A discrete set of bids captures the reality that the minimum unit of money is a cent.  The bidder with the highest bid wins the good.  If there are more than one highest bidders, the good is allocated to one of them uniformly at random.  The bidder who wins the good pays her bid $b^i_t$, obtaining utility $v^i - b^i_t$; other bidders obtain utility $0$.
Let $u^i(b^i_t, \bm b^{-i}_t)$ denote bidder $i$'s (expected) utility when $i$ bids $b^i_t$ while other bidders bid $\bm b^{-i}_t = (b^1_t, \ldots, b^{i-1}_t, b^{i+1}_t, \ldots, b^N_t)$, i.e., $u^i(b^i_t, \bm b^{-i}_t) = (v^i - b^i_t)\mathbb{I}[b^i_t = \max_{j\in \mathcal N} b^j_t] \frac{1}{|\argmax_{j\in \mathcal N} b^j_t|}$.  

We assume that bidders never bid above or equal to their values since that brings them negative or zero utility, which is clearly dominated by bidding $0$.  We denote the set of possible bids of each bidder $i$ by $\mathcal B^i=\{0, 1, \ldots, v^i-1\}$. % \footnote{We could allow a bidder to bid above $v^i-1$.  But a rational bidder will quickly learn to not place such bids. }  
% As noted by \citep{filos-ratsikas_complexity_2021}, it is standard in the literature to assume that $0$ is in the bidding space.  Bidding $0$ can be interpreted as not participating in the auction. 
% It is also standard to assume that bidders never overbid, i.e., never bid above their values.  We further assume that bidders never bid exactly their values, since this leads to zero utility for the bidder. 

\paragraph{\bf Online learning.} 
We assume that each bidder $i\in \mathcal N$ chooses her bids at every round using an online learning algorithm.  Specifically, we regard the set of possible bids $\mathcal{B}^i$ as a set of actions (or arms).  At each round $t$, the algorithm picks (possibly in a random way) an action $b^i_t \in \mathcal B^i$ to play, and then receives some feedback.  The feedback may include the rewards (i.e., utilities) of all possible actions in $\mathcal B^i$ (in the \emph{experts} setting) or only the reward  of the chosen action $b^i_t$ (in
the \emph{multi-arm bandit} setting). 
With feedback, the algorithm updates its choice of actions in future rounds.  We do not assume a specific feedback model in this work.  Our analysis will apply to all online learning algorithms that satisfy the following property, called ``mean-based'' \citep{braverman_selling_2018, feng_convergence_2021}, which roughly says that the algorithm picks actions with low average historical rewards with low probabilities.
%\tao{Does our analysis really apply to partial-feedback mean-based algorithms like EXP3?}
%\wq{EXP3 is mean-based so our results apply. This is also discussed in a recent follow-up paper: \url{https://arxiv.org/html/2412.15707v1} section B.1.3}
%\tao{The mean-based property we are using is actually stronger than the one in \cite{braverman_selling_2018}.  \cite{braverman_selling_2018} says $\Pr[b' \mid \exists b', \alpha_{t-1}(b') - \alpha_{t-1}(b) > V \gamma_t ] \le \gamma_t$. This is indeed satisfied by EXP-3.  However, what we need is that: \textbf{for any history} $H_{t-1}$ satisfying $\alpha_{t-1}(b') - \alpha_{t-1}(b) > V \gamma_t$, $\Pr[b' \mid H_{t-1}] \le \gamma_t$. For partial-feedback bandit algorithms like EXP-3, the history $H_{t-1}$ includes utility estimation.  If the utility estimation happens to be very bad (in particular, $\alpha_{t-1}(b')$ is small but $\hat \alpha_{t-1}(b')$ is large) in the history $H_{t-1}$, then we might play $b'$ at round $t$ with high probability.} 
\begin{definition}[mean-based algorithm]
Let $\alpha_t^i(b)$ be the average reward of action $b$ in the first $t$ rounds: $\alpha_t^i(b) = \frac{1}{t} \sum_{s=1}^t u^{i}(b, \bm b^{-i}_s)$.
% Let $(\gamma_t)_{t=1}^\infty$ be a decreasing sequence such that $\gamma_t\to 0$ as $t\to\infty$. 
An algorithm is \emph{$\gamma_t$-mean-based} if, for any $b\in \mathcal B^i$, whenever there exists $b'\in\mathcal B^i$ such that $\alpha_{t-1}^i(b') - \alpha_{t-1}^i(b) > V \gamma_t$, the probability that the algorithm picks $b$ at round $t$ is at most $\gamma_t$. 
An algorithm is \emph{mean-based} if it is $\gamma_t$-mean-based for some decreasing sequence $(\gamma_t)_{t=1}^\infty$ satisfying $\gamma_t \to 0$ as $t\to\infty$.
\end{definition}

%In this work we assume $\gamma_t=O(\frac{1}{t^r})$ for some constant $0<r<1$.  
In this work, we assume that the online learning algorithm runs for infinitely many rounds. This captures the scenario where bidders do not know how long they will be in the auction and hence use infinite-horizon learning algorithms.   
%Infinite-round mean-based algorithms 
%with $\gamma_t=O(\frac{1}{t^r})$
%can be obtained by modifying classical finite-round no-regret learning algorithms (e.g., MWU and FTPL) to have a learning rate parameter that decreases with time at an appropriate rate (see e.g., \cite{mohri_foundations_2012}). 
Infinite-horizon mean-based algorithms can be obtained by modifying classical finite-horizon mean-based algorithms (e.g., MWU) with constant learning rates to have decreasing learning rates. Some examples of infinite-horizon mean-based algorithms are given below:  
\begin{example}\label{ex:mean-based-algorithms}
Let $(\eps_t)_{t=1}^\infty$ be a decreasing sequence approaching $0$.  The following algorithms are mean-based: 
\begin{itemize}
    \item \emph{Follow the Leader} (also called \emph{Greedy}): at each round $t\ge 1$, each player $i\in \mathcal N$ chooses an action $b\in \argmax_{b\in\mathcal B^i}\{ \alpha_{t-1}^i(b) \}$ (with a tie-breaking rule specified by the algorithm). 
    \item \emph{$\eps_t$-Greedy}: at each round $t\ge 1$, each player $i\in \mathcal N$ with probability $1-\eps_t$ chooses $b\in \argmax_{b\in\mathcal B^i}\{ \alpha_{t-1}^i(b) \}$, with probability $\eps_t$ chooses an action in $\mathcal B^i$ uniformly at random.
    \item \emph{Multiplicative Weights Update} (MWU, also called \emph{Hedge}): at each round $t\ge 1$, each player $i\in \mathcal N$ chooses each action $b\in\mathcal B^i$ with probability $\frac{w_{t-1}(b)}{\sum_{b'\in\mathcal B^i} w_{t-1}(b')}$, where $w_t(b) = \exp(\eps_t\sum_{s=1}^t u^i(b, \bm b_s^{-i}))$.\footnote{This MWU algorithm is different from another MWU algorithm where the weight is defined by $w_t(b) = w_{t-1}(b) \cdot \exp(\eps_t u^i(b, \bm b_t^{-i})) = \exp(\sum_{s=1}^t \eps_s u^i(b, \bm b_s^{-i}))$.
    The latter algorithm is not mean-based because the rewards $u^i(b, \bm b_s^{-i})$ in earlier rounds matter more than rewards in later rounds given decreasing $\eps_s$. 
    The algorithm we define here treat rewards at different rounds equally and is hence mean-based.}
\end{itemize}
\end{example}

Clearly, Follow the Leader is $(\gamma_t=0)$-mean-based and $\eps_t$-Greedy is $\eps_t$-mean-based.  One can see \cite{braverman_selling_2018} for why MWU is mean-based.  Additionally,  MWU is no-regret when the sequence $(\eps_t)_{t=1}^\infty$ is set to $\eps_t = O(1/\sqrt t)$ %\citep{mohri_foundations_2012, sahai_2018_lecture}.
(see, e.g., Theorem 2.3 in \cite{cesa-bianchi_prediction_2006}). 

%===========================
%
% Bidders can choose bids $b^i_t$ at each round based on the information they observed from previous rounds (((((to maximize their aggregate utility from all rounds))))).
%
%Each player runs a no-regret learning algorithm in the repeated first-price auction. Each round $t$, player $i$ bids $b_t^i$, and gets reward $r_t^i=u(b_t^i, b_t^{-i})$. When we say the reward of player $i$ bidding $b$ in round $t$, the bids of other players are fixed as the history is, which is $u(b, b_t^{-i})$.
%
%
%===========================
%
\paragraph{\bf Equilibria in first-price auctions.} Before presenting our main results, we characterize the set of all Nash equilibria in the single-round first-price auction where bidders have fixed values $v^1\ge v^2 \ge \cdots \ge v^N$.  We focus on pure-strategy Nash equilibria. %\footnote{Whether and how our results extend to mixed-strategy Nash equilibria is open.}
Recall that $u^i(b^i, \bm b^{-i})$ denotes the utility of bidder $i$ when she bids $b^i$ while others bid $\bm b^{-i} = (b^1, \ldots, b^{i-1}, b^{i+1}, \ldots, b^N)$.  A bidding profile $\bm b = (b^1, \ldots, b^N) = (b^i, \bm b^{-i})$ is called a \emph{Nash equilibrium} if $u^i(\bm b) \ge u^i(b', \bm b^{-i})$ for any $b'\in \mathcal B^i$ and any $i\in \mathcal N$.  
Let $M^i$ be the set of bidders who have the same value as bidder $i$, $M^i = \{j\in \mathcal N: v^j=v^i\}$. $M^1$ is the set of bidders with the highest value. 
% \begin{proposition}\label{prop:Nash}
% The following bidding profile is a Nash equilibrium in the first-price auction with fixed values: 
% \begin{itemize}
%     \item The case of $|M^1| \ge 3$: $b^i=v^1-1$ for $i \in M^1$ and $b^j \le v^1-2$ for $j\notin M^1$.
%     \item The case of $|M^1| = 2$: 
%     \begin{itemize}
%         \item If $v^1 = v^2 > v^3 + 1$ (or $N=2$), two Nash equilibria are possible: $b^1=b^2=v^1-1$ or $b^1=b^2=v^1-2$, with $b^j\le v^1-3$ for $j\notin M^1$. 
%         \item If $v^1=v^2 = v^3 + 1$: $b^1=b^2=v^1-1$ and $b^j \le v^1-2$ for $j\notin M^1$.  
%     \end{itemize}
%     \item The case of $|M^1|=1$:
%     \begin{itemize}
%         \item One Nash equilibrium is: $b^1=v^2$, at least one bidder in $M^2$ bids $v^2-1$, all other bidders bid $b^j\le v^2-1$.  
%         \item If $v^1=v^2+1$ and $|M^2|=1$, another Nash equilibrium is: $b^1=b^2=v^2-1$, $b^j\le v^2-2$ for $j\notin\{1, 2\}$. 
%     \end{itemize}
% \end{itemize}
% \end{proposition}
\begin{proposition}\label{prop:Nash}
The set of (pure-strategy) Nash equilibria in the first-price auction with fixed values $v^1\ge v^2 \ge \cdots \ge v^N$ are bidding profiles $\bm b = (b^1, \ldots, b^N)$ that satisfy the following: 
\begin{itemize}
    \item The case of $|M^1| \ge 3$: $b^i=v^1-1$ for $i \in M^1$ and $b^j \le v^1-2$ for $j\notin M^1$.
    \item The case of $|M^1| = 2$: 
    \begin{itemize}
        \item If $N=2$ or $v^1 = v^2 > v^3 + 1$: there are two types of Nash equilibria: (1) $b^1=b^2=v^1-1$, with $b^j\le v^1-3$ for $j\notin M^1$; (2) $b^1=b^2=v^1-2$, with $b^j\le v^1-3$ for $j\notin M^1$. 
        \item If $N>2$ and $v^1=v^2 = v^3 + 1$: $b^1=b^2=v^1-1$ and $b^j \le v^1-2$ for $j\notin M^1$.  
    \end{itemize}
    \item The case of $|M^1|=1$:
    \begin{itemize}
        \item Bidding profiles that satisfy the following are Nash equilibria: $b^1=v^2$, at least one bidder in $M^2$ bids $v^2-1$, all other bidders bid $b^j\le v^2-1$.  
        \item If $v^1=v^2+1$ and $|M^2|=1$, then there is another type of Nash equilibria: $b^1=b^2=v^2-1$, $b^j\le v^2-2$ for $j\notin\{1, 2\}$. 
    \end{itemize}
\end{itemize}
There are no other (pure-strategy) Nash equilibria. 
\end{proposition}
\noindent The proof of this proposition is straightforward and omitted. 
Intuitively, whenever more than one bidder has the highest value ($|M^1| \ge 2$), they should compete with each other by bidding $v^1-1$ (or $v^1-2$ if $|M^1|=2$ and no other bidders are able to compete with them).  When $|M^1| = 1$, the unique highest-value bidder (bidder $1$) competes with the second-highest bidders ($M^2$).

\section{Main Results: Convergence of Mean-Based Algorithms}
\label{sec:main-result}
We introduce some additional notations.  Let $\bm x^i_t\in \Delta(\mathcal B^i)$ be the mixed strategy of player $i$ in round $t$, where the $b$-th component of $\bm x^i_t$ is the probability that player $i$ bids $b\in\mathcal B^i$ in round $t$. The sequence $(\bm x^i_t)_{t=1}^\infty$ is a stochastic process, where the randomness comes from the bidder's randomized learning algorithms. Let $\bm 1_{b}$ be the vector $(0, ..., 0, 1, 0, ..., 0)$ where $1$ is in the $b$-th position. 

Our main results about the convergence of mean-based algorithms in repeated first-price auctions depend on the number of bidders with the highest value, $|M^1|$.  
\subsection{The case of $|M^1| \ge 3$} 
% \begin{theorem}\label{thm:main-M1-3}
% If $|M^1|\ge 3$, 
% \begin{itemize}
% \item Time-average convergence of action sequence: 
% \begin{equation}
%     \Pr\left[ \forall i\in M^1, \lim_{t\to\infty} f^i_t(v^1-1) = 1 \right] = 1. 
% \end{equation}
% \item Last-iterate convergence of mixed strategy:
% \begin{equation}
%     \Pr\left[ \forall i\in M^1, \lim_{t\to\infty} \bm x^i_t = \bm 1_{v^1-1} \right] = 1. 
% \end{equation}
% \end{itemize}
% \end{theorem}
\begin{theorem}\label{thm:main-M1-3}
If $|M^1|\ge 3$ and every bidder follows a mean-based algorithm, then, with probability $1$, both of the following events happen:
\begin{itemize}
\item Time-average convergence of bid sequence: $\lim_{t\to\infty} \frac{1}{t} \sum_{s=1}^t \mathbb{I}\big[ \forall i\in M^1, b_s^i=v^1-1 \big] = 1$.
\item Last-iterate convergence of mixed strategy profile: $\forall i\in M^1,~ \lim_{t\to\infty} \bm x^i_t = \bm 1_{v^1-1}$.
\end{itemize}
% Note that 
% $\mathbb{I}[\bm b_s \textnormal{ is a Nash equilibrium}] =\mathbb{I}[\forall i\in M^1, b_s^i=v^1-1]$ by Proposition \ref{prop:Nash}.
\end{theorem}
Theorem~\ref{thm:main-M1-3} can be interpreted as follows.  According to Proposition~\ref{prop:Nash}, when $|M^1| \ge 3$, the bidding profile $\bm b_s$ at round $s$ is a Nash equilibrium if and only if $\forall i\in M^1, b_s^i=v^1-1$, with bidders not in $M^1$ bidding $\le v^1-2$ by assumption (note that the bidders not in $M^1$ can follow a mixed strategy and need not converge to a deterministic bid).  Hence, the first result of Theorem~\ref{thm:main-M1-3} implies that the fraction of rounds where bidders play a Nash equilibrium approaches $1$ in the limit.  The second result shows that all bidders in $M^1$ bid $v^1-1$ with certainty eventually, achieving a Nash equilibrium. We will prove Theorem \ref{thm:main-M1-3} in Section~\ref{sec:main-proof}. 
% In addition to the above limit result, we also have a finite-time convergence result in the proof the theorem (see Proposition ???). 

\subsection{The case of $|M^1|=2$}
% \begin{theorem}\label{thm:main-M1-2}
% Time-average convergence of action sequence: if $|M^1|=2$ and $v^3<v^1-1$ (or $n=2$), then with probability $1$, one of the following two holds: 
% \begin{equation*}
%      [ \forall i\in M^1, \lim_{t\to\infty}  f^i_t = \bm 1_{v^1-1}] \text{ or } [\forall i\in M^1, \lim_{t\to\infty}  f^i_t = \bm 1_{v^1-2}]. 
% \end{equation*}
% \end{theorem}

% \begin{theorem}\label{thm:main-M1-2-v1-1}
% If $|M^1|=2$, $v^3=v^1-1$, and every bidder follows an arbitrary mean-based algorithm, then with probability $1$, both of the following two events happens:
% \begin{itemize}
%     \item $\lim_{t\to\infty}\frac{1}{t}\sum_{s=1}^{t} \mathbb{I}[\forall i\in M^1, b_s^i=v^1-1]=1$.
%     \item $\forall i\in M^1$, $\lim_{t\to\infty}\bm x_t^i=\bm{1}_{v^1-1}$.
% \end{itemize}
% \end{theorem}

% \begin{theorem}\label{thm:main-M1-2}
% If $|M^1|=2$, $v^3<v^1-1$ (or $n=2$), and every bidder follows an arbitrary mean-based algorithm, then with probability $1$, one of the following two events happens:
% \begin{itemize}
%     \item $\lim_{t\to\infty}\frac{1}{t}\sum_{s=1}^{t} \mathbb{I}[\forall i\in M^1, \ b_s^i=v^1-2]=1$;
%     \item $\lim_{t\to\infty}\frac{1}{t}\sum_{s=1}^{t} \mathbb{I}[\forall i\in M^1, b_s^i=v^1-1]=1$ and $\forall i\in M^1$, $\lim_{t\to\infty}\bm x_t^i=\bm{1}_{v^1-1}$.
% \end{itemize}
% \end{theorem}

\begin{theorem}\label{thm:main-M1-2}
If $|M^1|=2$ and every bidder follows a mean-based algorithm, then, with probability $1$, one of the following two events happens:
\begin{itemize}
    \item Time-average convergence to $v^1-2$: $\lim_{t\to\infty}\frac{1}{t}\sum_{s=1}^{t} \mathbb{I}[\forall i\in M^1, \ b_s^i=v^1-2]=1$;
    \item Time-average and last-iterate convergence to $v^1-1$: $\lim_{t\to\infty}\frac{1}{t}\sum_{s=1}^{t} \mathbb{I}[\forall i\in M^1, \ b_s^i=v^1-1]=1~$ and $~\forall i\in M^1$, $\lim_{t\to\infty}\bm x_t^i=\bm{1}_{v^1-1}$.
\end{itemize}
Moreover, if $N > 2$ and $v^3=v^1-1$ then only the second event happens. 
\end{theorem}

For the case of $N=2$ or $v^3<v^1-1$, according to Proposition~\ref{prop:Nash}, $\bm b_s$ is a Nash equilibrium if and only if both bidders in $M^1$ play $v^1-1$ or $v^1-2$ at the same time, with other bidders bidding $\le v^1-3$. Hence, Theorem \ref{thm:main-M1-2} shows that the bidders eventually converge to one of the two possible types of equilibria. 
Interestingly, experiments show that some mean-based algorithms converge to the equilibrium of $v^1-1$ while some converge to $v^1-2$. Also, one randomized algorithm may converge to different equilibria in different runs. 
See Section~\ref{sec:body-exp} for details. 

% Last-iterate convergence result does not always holds when $|M^1| = 2$ and the dynamics time-average converges to the equilibrium of $v^1-2$.
In the case of time-average convergence to the equilibrium of $v^1-2$ (the first case of Theorem~\ref{thm:main-M1-2}), the last-iterate convergence result ($\forall i\in M^1$, $\lim_{t\to\infty}\bm x_t^i=\bm{1}_{v^1-2}$) does not always hold. 
Consider an example with 2 bidders, with $v^1=v^2 = 3$. We can construct a $\gamma_t$-mean-based algorithm with $\gamma_t = O(\frac{1}{t^{1/4}})$ such that, with constant probability, $\lim_{t\to\infty}\frac{1}{t}\sum_{s=1}^{t} \mathbb{I}[\forall i\in M^1, \ b_s^i=v^1-2]=1$ holds but in infinitely many rounds we have $\bm x_t^i = \bm 1_2 = \bm 1_{v^1-1}$, so the algorithm does not converge in last iterate. 
The key idea is that, when $\alpha_t^i(1) - \alpha_t^i(2)$ is positive but lower than $V\gamma_t$ in some round $t$ (which happens infinitely often), we can let the algorithm bid $2$ with certainty in round $t+1$.  This does not violate the $\gamma_t$-mean-based property.  The algorithm is presented in Algorithm \ref{alg:ex}, and the full analysis is in Appendix \ref{proof:M2counterexample}. 
% The key idea is that with high probability, in infinitely many rounds, $\alpha_t^i(1) - \alpha_t^i(2)$ is positive but lower than $\gamma_t$. And the algorithm bids $2$ with probability 1 in round $t+1$, which dose not violate the mean-based property. 
\begin{proposition}\label{prop:M2counterexample}
There exists an example with $|M^1| = 2$ and a mean-based algorithm (Algorithm \ref{alg:ex}) such that, with constant probability, % the bidders' mixed strategy profiles do not converge to a Nash equilibrium in last-iterate. 
the bidders time-average converge to the Nash equilibrium of $v^1-2$ but do not last-iterate converge to a Nash equilibrium. 
\end{proposition}

\begin{algorithm}[ht]
\caption{A mean-based bidding algorithm}\label{alg:ex}
\begin{algorithmic}[1]
\Require Value $v = 3$.
\State Let $T_0 = 10^{12}$. $T_k = 32^k T_0$ for $k \ge 0$. $\gamma_t = T_k^{-1/4}$ for $t \in [T_k + 1, T_{k+1}]$.
\State \textbf{In the first $T_0$ rounds:} Bid $b_t = 1$ for $t \le T_0 - T_0^{2/3}$ and bid $b_t = 0$ for $T_0 - T_0^{2/3} + 1 \le t \le T_0$.
\For{$t \ge T_0 + 1$}
% \If{$t \le T_0 - T_0^{2/3}$}
% \State Bid $b_t = 1$.
% \ElsIf{$T_0 - T_0^{2/3} +1 \le t \le T_0$}
% \State Bid $b_t = 0$.
% \Else
\State Find $k$ such that $ 32^kT_0 +1 \le t \le 32^{k+1}T_0$.
\If {$t = T_k +1$, $\argmax_b{\alpha_{t-1}}(b)=1$, and $\alpha_{t-1}^i(1) - \alpha_{t-1}^i(2) < V\gamma_{t}$}
\State Bid $b_t = 2$.
\Else
\State Bid $b_t \in \argmax_{b\in\{0, 1, 2\}}{\alpha_{t-1}}(b)$ with probability $1- T_{k+1}^{-1/3}$ and $0$ with probability $ T_{k+1}^{-1/3}$.
\EndIf
%\EndIf
\EndFor
\end{algorithmic}
\end{algorithm}

\subsection{The case of $|M^1|=1$} 
In the case of $|M^1| = 1$, mean-based learning dynamics may not converge to a Nash equilibrium of the first-price auction in time average or last iterate, as shown in the following example.

\begin{example}\label{ex:M1=1}
Let $v^1 = 10$, $v^2 = v^3 = 7$. Assume that players use the \emph{Follow the Leader} algorithm % i.e., 
%\red{best respond to the history in each round,}
%\blue{$\to$ 
%in each round $t$, bid $\argmax_{b \in \mathcal B^i}\{\alpha_{t-1}^i(b)\}$
%}
with a specific tie-breaking rule.  
% This is a $0$-mean-based algorithm by definition.
They may generate the following bid sequence $(b_t^1,b_t^2,b_t^3)_{t\ge 1}$:
$(7, 6, 1), (7, 1, 6), (7, 1, 1), (7, 6, 1), (7, 1, 6), (7, 1, 1), \ldots$, 
while satisfying $0$-mean-based. 
Note that $(7, 1, 1)$ is not a Nash equilibrium according to Proposition \ref{prop:Nash} but it appears in $\frac{1}{3}$ fraction of rounds, which means that the dynamics do not converge in the time-average sense or the last-iterate sense to a Nash equilibrium. 
\end{example}

\begin{proof}
To prove this example, we need to verify that the players' algorithms are indeed $0$-mean-based under the above bid sequence. Because players $2$ and $3$ always get zero utility no matter what they bid, we only need to verify the $0$-mean-based property for player $1$.  Let $q_t$ be the fraction of rounds in the first $t$ rounds where one of players $2$ and $3$ bids $6$ (in the other $1-q_t$ fraction of rounds both players $2$ and $3$ bid $1$); clearly, $q_t \ge \frac{2}{3}$ for any $t\ge 1$.  For player $1$, at each round $t$ her average utility by bidding $7$ is $\alpha_{t-1}^1(7) = 10 - 7 = 3$; by bidding $6$, $\alpha^1_{t-1}(6) = (10-6)(\frac{1}{2}q_{t-1} + (1-q_{t-1})) = 4 (1-\frac{q_{t-1}}{2}) \le \frac{8}{3} < 3$; by bidding $2$,  $\alpha^1_{t-1}(2) = (10-2)(1-q_{t-1}) \le \frac{8}{3} < 3$; and clearly $\alpha^1_{t-1}(b) < 3$ for any other bid.  Hence, $7 = \argmax_{b\in \mathcal B^1}\{\alpha_{t-1}^1(b)\}$, player $1$ satisfies $0$-mean-based. 
\end{proof}

Example~\ref{ex:M1=1} also shows that, in the case of $|M^1|=1$, the learning dynamics generated by a mean-based algorithm may not converge to Nash equilibrium in the classical sense of ``convergence of empirical distribution''.  Specifically, let $p_t^i = \frac{1}{t} \sum_{s=1}^t \bm{1}_{b_s^i} \in \Delta(\mathcal B^i)$ be the empirical distribution of player $i$'s bids up to round $t$.  ``Convergence of empirical distribution'' means that the players' empirical distributions $(p_t^1,p_t^2,p_t^3)_{t\ge1}$ converge to a mixed-strategy Nash equilibrium in the limit.
In Example \ref{ex:M1=1}, the players' empirical distributions converge to $(p^1,p^2,p^3)$ where $p^1(7) = 1$ and for $i=2, 3$, $p^i(6) = \frac{1}{3}$ and $p^i(1) = \frac{2}{3}$.  Given bidders 2 and 3's strategies $(p^2, p^3)$, bidder $1$ can obtain utility $(10-2)(\frac{2}{3})^2 = \frac{32}{9}$ by bidding $2$, which is larger than the utility of bidding $7$, which is $10-7 = 3$.  Thus, bidder $1$'s strategy $p^1$ is not a best response to $(p^2, p^3)$, hence $(p^1, p^2, p^3)$ is not a Nash equilibrium.

The mean-based algorithm in Example~\ref{ex:M1=1} is not no-regret.  In Section~\ref{sec:body-exp} we show by experiments that non-convergence phenomena also hold for no-regret mean-based algorithms, e.g., MWU. 

\section{Proof of Theorem~\ref{thm:main-M1-3}}
\label{sec:main-proof}
The proof of Theorem~\ref{thm:main-M1-3} covers the main ideas and proof techniques of our convergence results, so we present it here. 
%The proofs of other main results are in Appendix~\ref{app:main-result-proof} and the missing proofs in this section are in Appendix~\ref{app:main-proof}. 
We first provide a proof sketch. Then in Section~\ref{sec:property} we provide some properties of mean-based algorithms that will be used in the formal proof.  Section~\ref{sec:first-stage} and Section~\ref{sec:main-M1-3-proof} prove Theorem~\ref{thm:main-M1-3}. 

\paragraph{Proof sketch}
At a high level, the proof uses the idea of iterative elimination of dominated strategies in game theory. 
We first use an induction argument to show that bidders with the highest value (i.e., those in $M^1$) will gradually learn to eliminate bids $0, 1, \ldots, v^1-3$.  Then we further prove that: if $|M^1|=3$, they will eliminate $v^1-2$ and hence converge to $v^1-1$; if $|M^1|=2$, the two bidders may end up playing $v^1-1$ or $v^1-2$.

To see why bidders in $M^1$ will learn to eliminate $0$, suppose that there are two bidders in total and one of them (say, bidder $i$) bids $b$ with probability $P(b)$ in history.  For the other bidder (say, bidder $j$), if bidder $j$ bids $0$, she obtains utility $\alpha(0) = (v^1-0)\frac{P(0)}{2}$; if she bids $1$, she obtains utility $\alpha(1) = (v^1-1)(P(0)+\frac{P(1)}{2})$.  Since $\alpha(1) - \alpha(0) = \frac{v^1-2}{2}P(0) + (v^1-1)\frac{P(1)}{2} > 0$ (assuming $v^1\ge 3$), bidding $1$ is better than bidding $0$ for bidder $j$.  Given that bidder $j$ is using a mean-based algorithm, she will play $0$ with small probability (say, zero probability).  The same argument applies to bidder $i$.  Hence, both bidders will learn to not play $0$.  Then we take an inductive step: assuming that no bidders play $0, \ldots, k-1$, we have $\alpha(k+1) - \alpha(k) = \frac{v^1-k-2}{2}P(k) + \frac{v^1-k-1}{2}P(k+1) > 0$ for $k\le v^1-3$, therefore $k+1$ is a better response than $k$ and both players will avoid bidding $k$.  An induction shows that they will finally learn to avoid $0, 1, \ldots, v^1-3$. 
Then, for the case of $|M^1|\ge 3$, we will use an additional lemma (Lemma~\ref{lemma:mean-based-v-2}) to show that, if bidders bid $0, 1, \ldots, v^1-3$ rarely in history, they will also avoid bidding $v^1-2$ in the future.

However, mean-based learning bidders may choose dominated bids with a small but positive probability. To prove a high-probability convergence result, we use and generalize a time-partitioning technique from \cite{feng_convergence_2021}. We partition the time horizon into periods $1 < T_0 < T_1 < T_2 < \cdots$.  If bidders bid $0, 1, \ldots, k-1$ with low frequency from round $1$ to $T_{k-1}$, then using the mean-based properties in Lemma \ref{lemma:mean-based} and Lemma~\ref{lemma:mean-based-v-2}, we show that they will bid $k$ with probability at most $\gamma_t$ in each round from $T_{k-1}+1, T_{k-1}+2, \ldots, $ to $ T_k$. A use of Azuma's inequality shows that the frequency of bid $k$ in period $(T_{k-1}, T_k]$ is also low with high probability, which concludes the induction. 
By constructing an appropriate partition and taking a union bound over all periods, we are able to prove that the frequency of bids less than $v^1-1$ converges to $0$ with high probability. 

% The summation of the failure probability in the use of Azuma's inequality may be too large to converge to $0$ without the partition. 

% Notice that partitioning time is significant here.
% Notice that we need to partition the time and also choose period lengths appropriately so that the summation of the failure probability in the use of Azuma's inequality won't be too large. 

% This partitioning argument reduces the failure probability in the use of concentration inequality, which gives us the ``with probability $1$'' convergence result. 

\subsection{Iterative Elimination Properties of Mean-Based Algorithms} \label{sec:property}
We define some notations for the proofs.  Let $P_t^i(k)$ be the frequency that the highest bid submitted by bidders other than $i$ is $k$ during the first $t$ rounds:
\[ P_t^i(k) = \frac{1}{t} \sum_{s=1}^t \mathbb{I}[\max_{j\ne i} b_s^j = k ].\]
Let $P_t^i(0:k) = \sum_{\ell = 0}^k P_t^i(\ell)$, and $P_t^i(0:-1) = 0$.  Let $Q_t^i(k)$ be the probability that bidder $i$ wins the item with ties if she bids $k$ in history: 
\[ Q_t^i(k) = \frac{1}{t} \sum_{s=1}^t \mathbb{I}[\max_{j\ne i} b_s^j = k]\frac{1}{|\argmax_{j\ne i} b_s^j| + 1}.\]
Clearly, 
\begin{equation}\label{eq:P-Q}
0 ~ \le ~ \frac{1}{N} P_t^i(k) ~ \le ~ Q_t^i(k) ~ \le ~ \frac{1}{2} P_t^t(k) ~ \le ~ \frac{1}{2}.
\end{equation} 
Recall that $\alpha_t^i(k)$ is bidder $i$'s average utility by bidding $k$ in the first $t$ rounds.  We can write $\alpha_t^i(k)$ using $P_t^i(0:k-1)$ and $Q_i^t(k)$: 
\begin{equation}\label{eq:alpha-P-Q}
    \alpha_t^i(k) = (v^i-k) \big( P_t^i(0:k-1) + Q_t^i(k) \big). 
\end{equation}
We use $H_t$ to denote the history of the first $t$ rounds, which includes the realization of all randomness in the first $t$ rounds.  Bidders themselves do not necessarily observe the full history $H_t$.  Given $H_{t-1}$, each bidder's mixed strategy $\bm x_t^i$ at round $t$ is fully determined, and the $k$-th component of $\bm x_t^i$ is $\Pr[b_t^i = k\mid H_{t-1}]$.  The following lemma shows that, if other bidders rarely bid $0$ to $k-1$ in history, then bidder $i\in M^1$ will not bid $k$ with large probability in round~$t$, for $k\le v^1-3$.
\begin{lemma}\label{lemma:mean-based}
Assume $v^1\ge 3$.  For any $i\in M^1$, any $k\in \{0, 1, \ldots, v^1-3\}$, any $t$ such that $\gamma_t < \frac{1}{12NV}$, if the history $H_{t-1}$ of the first $t-1$ rounds satisfies $P_{t-1}^i(0:k-1) \le \frac{1}{3NV}$, then $\Pr[b_t^i = k \mid H_{t-1} ] \le \gamma_t$. 
\end{lemma}

The intuition behind Lemma~\ref{lemma:mean-based} is that, if other bidders never bid $0$ to $k-1$, then bidder $i$'s bid $k$ will be dominated by a mixed strategy between bidding $k+1$ and $v^1-1$. However, Lemma~\ref{lemma:mean-based} is stronger than the classical notion of iterative dominance: it requires bidder $i$'s bid $k$ to be dominated even when other bidders bid $0$ to $k-1$ with a small constant probability $\frac{1}{3NV}$.  This stronger property is crucial to our proof of Theorem \ref{thm:main-M1-3}. 

%\tao{I added the above intuition paragraph.  Let me know if you think it is unnecessary. }

\begin{proof}{Proof of Lemma~\ref{lemma:mean-based}}
%Assume $P_{t-1}^i(0:k-1)\le \frac{1}{2VN} - 2\gamma_t$.
If $\alpha_{t-1}^i(k+1) - \alpha_{t-1}^i(k) > V \gamma_{t}$, then by the mean-based property, the conditional probability $\Pr[b_t^i = k \mid H_{t-1}] \le \gamma_t$, so the lemma holds. Then, we consider the case where $\alpha_{t-1}^i(k+1) - \alpha_{t-1}^i(k) \le V \gamma_t$.  Using \eqref{eq:alpha-P-Q} and \eqref{eq:P-Q}, 
\begin{align*}
    V\gamma_t & \ge \alpha_{t-1}^i(k+1) - \alpha_{t-1}^i(k) \\
    & = (v^1-k-1) \Big( P_{t-1}^i(0:k) + Q_{t-1}^i(k+1) \Big) \\
    & \quad - (v^1-k) \Big( P_{t-1}^i(0:k-1) + Q_{t-1}^i(k) \Big) \\
    & \ge (v^1-k-1) P_{t-1}^i(k) \\
    & \quad - P_{t-1}^i(0:k-1) - (v^1-k) Q_{t-1}^i(k) \\
    & \ge (v^1-k-1) P_{t-1}^i(k) \\
    & \quad  - P_{t-1}^i(0:k-1)  - (v^1-k) \tfrac{P_{t-1}^i(k)}{2},
\end{align*}
which implies 
\begin{equation}\label{eq:P-gamma-Pk-1}
    P_{t-1}^i(k) \le \tfrac{2}{v^1-k-2}\big( V\gamma_t + P_{t-1}^i(0:k-1) \big).
\end{equation}
We then upper bound $\alpha_{t-1}^i(k)$ as follows: 
\begin{align*}
     & \alpha_{t-1}^i (k) ~ \le ~ (v^1-k)\big(P_{t-1}^i(0:k-1) + \tfrac{1}{2} P_{t-1}^i(k)\big) \\
      & \le (v^1-k) P_{t-1}^i(0:k-1) \hspace{6em} \text{by \eqref{eq:P-gamma-Pk-1}}\\
      & \quad + \tfrac{v^1-k}{v^1-k-2}\big(V \gamma_t + P_{t-1}^i(0:k-1)\big) \\
    & = \tfrac{v^1-k}{v^1-k-2}V\gamma_t + \big( v^1-k + \tfrac{v^1-k}{v^1-k-2} \big) P_{t-1}^i(0:k-1) \\
    & \le 3V\gamma_t + \big( v^1-k + 3 \big) P_{t-1}^i(0:k-1) \\
     & \le 3V\gamma_t + 2V P_{t-1}^i(0:k-1). 
\end{align*}
By the condition of the lemma, we have $P_{t-1}^i(0:k-1) \le \frac{1}{3NV} < \frac{1}{2NV} - 2\gamma_t$. So 
\[ \alpha_{t-1}^i(k) < 3V\gamma_t + 2V \left(\tfrac{1}{2NV} - 2\gamma_t\right) = \tfrac{1}{N} - V \gamma_t. \]
Then, we note 
% It then follows
that $\alpha_{t-1}^i(v^1-1) = P_{t-1}^i(0:v^1-2) + Q_{t-1}^i(v^1-1) \ge \frac{1}{N} P_{t-1}^i(0:v^1-1) = \frac{1}{N}\cdot 1$ where the last equality 
%is 
holds
because no bidder bids above $v^1-1$ by assumption. 
Therefore, 
\[ \alpha_{t-1}^i(v^1-1) -  \alpha_{t-1}^i(k) > \tfrac{1}{N} - \left(\tfrac{1}{N} - V \gamma_t\right) = V\gamma_t.\] 
% 
% We then note that $\alpha_{t-1}^i(v^1-1) = P_{t-1}^i(0:v^1-2) + Q_{t-1}^i(v^1-1) \ge \frac{1}{N} P_{t-1}^i(0:v^1-1) = \frac{1}{N}$ where the last equality is because no bidder bids above $v^1-1$ by assumption. 
% It follows that 
% \begin{align*}
%     \alpha_t^i(v^1-1) & - \alpha_t^i(k) \\
%     & \ge \frac{1}{N} - (v^1-k)\Big(P_{t-1}^i(0:k-1) + \frac{1}{2} P_{t-1}^i(k)\Big) \\
%     \text{(by \eqref{eq:P-gamma-Pk-1}) } & \ge \frac 1 N - (v^1-k) P_{t-1}^i(0:k-1) \\
%     & \hspace{4em} - \frac{v^1-k}{v^1-k-2}\Big(V \gamma_t + P_{t-1}^i(0:k-1)\Big) \\
%     & = \frac 1 N - \frac{v^1-k}{v^1-k-2}V\gamma_t - \left( v^1-k + \frac{v^1-k}{v^1-k-2} \right) P_{t-1}^i(0:k-1) \\
%     (\frac{v^1-k}{v^1-k-2} & \le 3)
%      \ge \frac 1 N - 3V\gamma_t - \left( v^1-k + 3 \right) P_{t-1}^i(0:k-1) \\
%      & \ge \frac 1 N - 3V\gamma_t - 2V P_{t-1}^i(0:k-1) \\
%      \text{(by assumption) } & \ge \frac 1 N - 3V\gamma_t - 2V \left(\frac{1}{2VN} - 2\gamma_t\right) \\
%      & = V \gamma_t.
% \end{align*}
%Then by 
From the mean-based property, we obtain $\Pr[b_t^i = k \mid H_{t-1}] \le \gamma_t$. 
% implying
%Therefore,
% $\Pr[b_t^i = k \mid H_{t-1}] \le \gamma_t$. 
 \end{proof}

The following lemma is for $k = v^1-2$: if bidders rarely bid $0$ to $v^1-3$ in history and $|M^1| \ge 3$, then bidder $i\in M^1$ will not bid $v^1-2$ with large probability in round~$t$. Intuitively, this is because $v^1 - 2$ is dominated by $v^1-1$. See Appendix \ref{app:mean-based-v-2} for the proof. 
\begin{lemma}\label{lemma:mean-based-v-2}
Suppose $|M^1|\ge 3$ and $v^1\ge 2$.  For any $t$ such that $\gamma_t < \frac{1}{12N V}$, if the history $H_{t-1}$ of the first $t-1$ rounds satisfies $\frac{1}{t-1} \sum_{s=1}^{t-1} \mathbb{I}[\exists i\in M^1, b_s^i \le v^1-3] \le \frac{1}{3NV}$, then, $\forall$ $i\in M^1$, $\Pr[b_t^i = v^1- 2\mid H_{t-1} ] \le \gamma_t$.  
\end{lemma}

% \subsection{Base Case: Eliminating Bid $0$}

\subsection{Iteratively Eliminating Bids $0, 1, \ldots, v^1-3$}\label{sec:first-stage}
In this subsection we will use an induction argument to prove that, after a sufficiently long time, bidders in $M^1$ will rarely bid $0, 1, \ldots, v^1-3$ (Corollary~\ref{cor:first-stage-end}). 
We partition the time horizon into $v^1-3$ periods.
% and using an induction from $0$ to $v^1-3$. 
Let constants $c = 1+\frac{1}{12NV}$ and $d = \lceil \log_c (8NV) \rceil$. Let $T_b$ be any (constant) integer such that $\gamma_{T_b} < \frac{1}{12N^2 V^2}$ and $\exp\left( -\frac{(c-1)T_b}{1152N^2V^2} \right) \le \frac{1}{2}$.
%We partition the times as follows to get a constant bound for bidding $0, 1, \ldots, v^1-3$ with high probability.
% We partition the rounds from the beginning as follows: 
Let $T_0 = 12N V T_b$ and $T_k = c^d T_{k-1} = c^{dk}T_0 \ge (8N V)^kT_0$ for $k \in \{1, 2, \ldots, v^1-3\}$. Let $A_k$ be event
$$ A_k = \left[ \frac{1}{T_k} \sum_{t=1}^{T_k} \mathbb{I}[\exists i\in M^1, b_t^i\le k]\le \frac{1}{4N V} \right], $$
which says that bidders in $M^1$ bid $0, 1, \ldots, k$ not too often in the first $T_k$ rounds. Our goal is to show that $\Pr[A_{v^1-3}]$ is high. 
%The following lemma is for the base case of the proof of Lemma~\ref{lemma:first-stage}. 

%The following lemma is for avoiding  bidding $0$.
The \textbf{base case} is to show that $A_0$ happens with high probability. 
\begin{lemma}\label{lem:base-case}
%If $|M^1|\ge 2$, then
$\Prx{A_0} \ge 1 - \exp\left(-\frac{T_b}{24NV}\right)$.
%\begin{align*}
 %   & \hspace{-1em} \Pr\left[\frac{1}{T} \sum_{t=1}^T \mathbb{I}[\exists i\in M^1, b_t^i = 0] \le \frac{1}{T} \bigg(\Delta  + T_c + \sum_{t=T_b+1}^T |M^1|\gamma_t\bigg) \right] \\ &\ge 1 - \exp\bigg(-\frac{\Delta^2}{2(T-T_b)}\bigg). 
%\end{align*}
\end{lemma}

\begin{proof}
Consider any round $t\ge T_b$.  For any $i\in M^1$, given any history $H_{t-1}$ of the first $t-1$ rounds, it holds that $P_{t-1}^i(0:-1) = 0\le \frac{1}{3NV}$.  Hence, by Lemma~\ref{lemma:mean-based}, we have $\Pr[b_t^i = 0\mid H_{t-1}] \le \gamma_t$.
Using a union bound over $i\in M^1$, 
\[\Pr[\exists i\in M^1, b_t^i = 0\mid H_{t-1}] \le |M^1|\gamma_t.\]
Let $Z_t = \mathbb{I}[\exists i\in M^1, b_t^i = 0] - |M^1|\gamma_t$ and let $X_t=\sum_{s=T_b+1}^t Z_s$.  We have $\Ex{Z_t \mid H_{t-1}} \le 0$.  Therefore, the sequence $X_{T_b+1}, X_{T_b+2}, \ldots, X_{T_0}$ is a supermartingale (with respect to the sequence of history $H_{T_b}, H_{T_b+1}, \ldots, H_{T_0-1}$).  By Azuma's inequality, for any $\Delta>0$, we have 
\begin{equation*}
    \Pr\bigg[\sum_{t=T_b+1}^{T_0} Z_t \ge \Delta \bigg] \le \exp\left(-\frac{\Delta^2}{2(T_0 - T_b)}\right). 
\end{equation*}
Let $\Delta = T_b$. We have with probability at least $1 - \exp\left(-\frac{\Delta^2}{2(T_0 - T_b)}\right) \ge 1 - \exp\left(-\frac{T_b}{24N V}\right)$, $\sum_{t=T_b+1}^{T_0} Z_t < T_b$ holds,
namely, $\sum_{t=T_b+1}^{T_0} \mathbb{I}[\exists i\in M^1, b_t^i = 0] < T_b + \sum_{t=T_b+1}^{T_0} |M^1|\gamma_t$, 
which implies
\begin{align*}
    & \frac{1}{T_0} \sum_{t=1}^{T_0} \mathbb{I}[\exists i\in M^1, b_t^i = 0] \\
    & \le \frac{1}{T_0} \bigg( T_b + \sum_{t=T_b+1}^{T_0} \mathbb{I}[ \exists i\in M^1, b_t^i = 0]\bigg)\\& < \frac{1}{T_0} \bigg( 2T_b +  \sum_{t=T_b+1}^{T_0} |M^1|\gamma_t\bigg) ~ \le ~ \frac{1}{4NV}, 
\end{align*}
where the last inequality is because $\frac{T_b}{T_0} = \frac{1}{12N V}$ and $\frac{1}{T_0}\sum_{t=T_b+1}^{T_0} |M^1|\gamma_{t} \le |M^1|\gamma_{T_b} \le \frac{1}{12N V}$.
 \end{proof}

%The following lemma shows that if the number of bidders with the highest type is at least two, then after sufficiently many rounds, the frequency of bids less then or equal to $v^1-3$ is bounded above by $\frac{1}{4NV}$, with high probability. If $|M^1|\ge3$, then the low frequency of the low bids implies that bidders with the highest type would bid $v^1-1$ with high probability in each round. This concludes the last-iterate convergence result to Nash equilibrium of the one-shot first-price auction when $|M^1|\ge3$. If $|M^1| = 2$, then we need a more sophisticated analysis. 

Then, we use \textbf{induction} to show that, if bidders in $M^1$ seldom bid $0, 1, \ldots, k$ in the first $T_k$ rounds, then they will also seldom bid $0, 1, \ldots, k, k+1$ in the first $T_{k+1}$ rounds, with high probability. 
\begin{lemma}\label{lemma:first-stage}
Suppose $|M^1|\ge 2$. For every $k \in[ 0, v^1-4]$, $\Prx{A_{k+1} \mid A_k} \ge 1 - \sum_{j=1}^d \exp\Big( -\frac{|\Gamma_{\ell}^j|}{1152N^2V^2} \Big)$.     
\end{lemma}

To prove Lemma~\ref{lemma:first-stage}, we divide the rounds in $[T_k,T_{k+1}]$ to $d = \lceil \log_c (8N V) \rceil$ episodes such that $T_k = T_k^0 < T_k^1 < \cdots < T_k^{d} = T_{k+1}$ where $T_k^{j} = cT_k^{j-1}$. Let $\Gamma_k^j = [T_k^{j-1}+1,T_k^{j}]$, with $|\Gamma_k^j| = T_k^j - T_k^{j-1}$.  
We define a series of events $B_k^j$ for $j \in [0,d]$. $B_k^0$ is the same as $A_k$. For $j \in [1, d]$,
\begin{equation*}
    B_k^j = \bigg[ \sum_{t\in \Gamma_k^j} \mathbb{I}[\exists i\in M^1, b_t^{i} \le k+1] \le \frac{|\Gamma_k^j|}{8N V} \bigg]. 
\end{equation*}

\begin{claim}\label{claim:induction-B}
%Suppose $|M^1|\ge 2$. For every $k \in[ 0, v^1-4]$,
For every $j \in [0,d-1]$, $\Pr\big[B_k^{j+1}\mid A_k, B_k^1, \ldots, B_k^j\big] \ge 1 - \exp\Big(-\frac{|\Gamma_k^{j+1}|}{1152N^2 V^2}\Big)$.
\end{claim}

\begin{proof}
Suppose $A_k, B_k^1, \ldots, B_k^j$ happen. 
We denote $A_k^j = [A_k, B_k^1, \ldots, B_k^j]$.
We argue that event $A_k^j$ implies $P_{t-1}^i(0:k) \le \frac{1}{3NV}$ for every bidder $i\in M^1$ and every round $t \in \Gamma_k^j = [T_k^j+1, T_k^{j+1}]$. 
% consider $P_t^i(0:k)$.
Recall that $P_{t-1}^i(0:k) = \frac{1}{t-1} \sum_{s=1}^{t-1} \mathbb{I}[\max_{i'\ne i} b_s^{i'} \le k]$. 
Because $|M^1| \ge 2$, the event $\mathbb{I}[\max_{i'\ne i} b_s^{i'} \le k]$ implies that there exists $i^*\in M^1$, $i^*\ne i$, such that $b_s^{i^*} \le k$. % Thus, $P_{t-1}^i(0:k) \le$. 
% Given event $A_k^j$, we have 
Thus, 
\begin{align*}
    P_{t-1}^i(0:k) & ~ \le ~ \frac{1}{t-1} \sum_{s=1}^{t-1} \mathbb{I}[\exists i\in M^1, b_s^i \le k] \\
    & ~ = ~ \frac{1}{t-1} \bigg( \sum_{s=1}^{T_k}  \mathbb{I}[\exists i\in M^1, b_s^i \le k] \\
    & \qquad + \sum_{s \in \Gamma_k^1 \cup \cdots \cup \Gamma_k^j}  \mathbb{I}[\exists i\in M^1, b_s^i \le k] \\
    & \qquad +  \sum_{s=T_k^j+1}^{t-1} \mathbb{I}[\exists i\in M^1, b_s^i \le k] \bigg) \\
     % & \le  \frac{1}{t-1} \left( T_k^j P_{T_k^j}^i(0:k) + (t-1-T_k^j) \right) \\
    \text{(by event $A_k^j$)} & ~ \le ~ \frac{1}{t-1} \bigg( T_k \frac{1}{4NV} ~ + ~ (T_k^j - T_k)\frac{1}{8NV} \\
    & \qquad +  \sum_{s=T_k^j+1}^{t-1} \mathbb{I}[\exists i\in M^1, b_s^i \le k] \bigg) \\
    & ~ \le ~ \frac{1}{t-1} \bigg( T_k^j \frac{1}{4NV} + (t - 1 - T_k^j) \cdot 1 \bigg).
\end{align*}
Because $T_k^j \le t-1 \le T_k^{j+1} = c T_k^j$ with $c = 1 + \frac{1}{12NV}$, we have 
\begin{equation*}
     P_{t-1}^i(0:k) ~ \le ~ \frac{1}{4NV} + \frac{T_k^{j+1}-T_k^j}{T_{k}^{j+1}} ~ \le ~ \frac{1}{3N V} 
     % & < \frac{1}{2NV} - 2\gamma_t \quad {\text{(because $\gamma_t \le \gamma_{T_b} < \tfrac{1}{12N V}$)}}.
\end{equation*}
for any round $t \in \Gamma_k^{j+1}$. 
%if the history $H_{t-1}$ satisfies $A_k^j$, then we have $P_{t-1}^i(0:k) \le \frac{1}{2NV} - 2\gamma_t$ and according to
Then, Lemma \ref{lemma:mean-based} implies
$\Prx{b_t^i = b \mid H_{t-1},A_k^j} \le \gamma_t$ for every $b \le k+1$. 
Consider the event $[\exists i\in M^1, b_t^{i} \le k+1]$.  Using union bounds over $i\in M^1$ and $b\in\{0, 1, \ldots, k+1\}$,  
\begin{align*}
     & \Pr \Big[\exists i\in M^1, b_t^{i} \le k+1 \Biggiven H_{t-1}, A_k^j \Big] \\
     &  \le |M^1| \cdot \Prx{b_t^i \le k+1 \;\Big|\; H_{t-1}, A_k^j} \\ & \le |M^1|(k+2)\gamma_t ~~ \le ~~ |M^1|V\gamma_t. 
\end{align*}
\noindent Let $Z_t = \mathbb{I}[\exists i\in M^1, b_t^i \le k+1] - |M^1|V\gamma_t$ and let $X_t=\sum_{s=T_k^j+1}^t Z_s$.  We have $\Ex{Z_t \mid A_k^j, H_{t-1}} \le 0$.  Therefore, the sequence $X_{T_k^j+1}, X_{T_k^j+2}, \ldots, X_{T_k^{j+1}}$ is a supermartingale (with respect to the sequence of history $H_{T_k^j}, H_{T_k^j+1}, \ldots, H_{T_k^{j+1}-1}$).  By Azuma's inequality, for any $\Delta > 0$, we have 
\begin{equation*}
    \Pr\Bigg[\sum_{t=T_k^j+1}^{T_k^{j+1}} Z_t \ge \Delta \bigggiven A_k^j \Bigg] \le \exp\bigg(-\frac{\Delta^2}{2|\Gamma_k^{j+1}|}\bigg). 
\end{equation*}
Let $\Delta = \frac{|\Gamma_k^{j+1}|}{24N V}$. 
Then with probability at least 
$
1-\exp\Big(-\frac{|\Gamma_k^{j+1}|}{1152N^2 V^2}\Big)
$
we have $\sum_{t \in \Gamma_k^{j+1}} Z_t < \frac{|\Gamma_k^{j+1}|}{24N V}$, implying $\sum_{t\in \Gamma_k^{j+1}}  \mathbb{I}[\exists i\in M^1, b_t^{i} \le k+1] < \frac{|\Gamma_k^{j+1}|}{24N V} + \sum_{t \in \Gamma_k^{j+1}} |M^1|V \gamma_t ~ \le ~ \frac{|\Gamma_k^{j+1}|}{24N V} + |M^1|V \frac{|\Gamma_k^{j+1}|}{12N^2V^2} ~ \le ~ \frac{|\Gamma_k^{j+1}|}{8NV}$,
which proves the claim. 
 \end{proof}

\begin{proof}{Proof of Lemma~\ref{lemma:first-stage}}
Suppose $A_k$ holds. We want to show that $A_{k+1}$ holds with high probability.
Using Claim~\ref{claim:induction-B} with $j=0, 1, \ldots, d-1$, we have, with probability at least $1 - \sum_{j=1}^d\exp\Big(-\frac{|\Gamma_k^j|}{1152N^2 V^2}\Big)$, all the events $B_k^1, \ldots, B_k^d$ hold, which implies
\begin{align*}
    & \frac{1}{T_{k+1}} \sum_{t=1}^{T_{k+1}} \mathbb{I}[\exists i\in M^1, b_t^{i} \le k+1 ] \\
    & \le \frac{1}{T_{k+1}} \bigg( T_k\cdot 1 + \sum_{t\in \Gamma_k^1\cup\cdots \cup\Gamma_k^d} \mathbb{I}[\exists i\in M^1, b_t^{i} \le k+1]\bigg) \\
    & \le \frac{1}{T_{k+1}} \Big( T_k \cdot 1 + (T_{k+1}-T_k) \cdot \frac{1}{8NV}\Big) \\
    & \le \frac{1}{8NV} + \Big(1-\frac{T_k}{T_{k+1}}\Big)\frac{1}{8NV} ~~ \le ~~ \frac{1}{4N V}, 
\end{align*}
where in the third inequality we used $T_{k+1} \ge (8NV) T_k$.  
Thus $A_{k+1}$ holds. 
 \end{proof}

Using induction from $k=0, 1, \ldots$ to $v^1-4$, we have all events $A_0, A_1, \ldots, A_{v^1-3}$ happen with probability at least $1 - \exp\big(-\frac{T_b}{24NV}\big) - \sum_{k=0}^{v^1-4}\sum_{j=1}^{d}\exp\big(-\frac{|\Gamma_k^j|}{1152N^2 V^2}\big)$.
We then lower bound the probability, obtaining the following corollary: 
% Then we bound the probability. 
% Note that $|\Gamma_k^j| = T_k^j - T_k^{j-1} = cT_k^{j-1} - cT_k^{j-2} = c|\Gamma_k^{j-1}|$ for any $k \in \{0, 1, \ldots, v^1-4\}$ and $j \in \{2, \ldots, d\}$, and that $|\Gamma_k^1| = c|\Gamma_{k-1}^d|$ for any $k \in \{1, 2, \ldots, v^1-4\}$.  We also note that $|\Gamma_0^1| = (c-1)T_0 = T_b$. Thus,
% $\sum_{k=0}^{v^1-4}\sum_{j=1}^d \exp\big( -\frac{|\Gamma_{k}^j|}{1152N^2V^2} \big) = \sum_{s=0}^{(v^1-3)d-1}\exp\big(-\frac{c^sT_b}{1152N^2V^2}\big)$.  Moreover, we can show that $\sum_{s=0}^{(v^1-3)d-1}\exp\big(-\frac{c^sT_b}{1152N^2V^2}\big) \le 2\exp\left(-\frac{T_b}{1152N^2V^2} \right)$.
% Hence, we obtain the following corollary (see Appendix~\ref{app:main-proof} for a detailed proof): 
\begin{corollary}\label{cor:first-stage-end}
Suppose $|M^1|\ge 2$. $\Prx{A_{v^1-3}} \ge 1 - \exp\Big(-\frac{T_b}{24NV}\Big) - 2\exp\Big(-\frac{T_b}{1152N^2V^2} \Big)$. 
\end{corollary}
\begin{proof}
Using Lemma~\ref{lem:base-case} and Lemma~\ref{lemma:first-stage} from $k=0$ to $v_1-4$, we get
\begin{align*}
    & \Prx{A_{v^1-3}} \ge \Prx{A_0, A_1, \ldots, A_{v^1-3}} \\
    & \ge 1 - \exp\left(-\tfrac{T_b}{24NV}\right) - \sum_{k=0}^{v^1-4}\sum_{j=1}^d \exp\left( -\tfrac{|\Gamma_{k}^j|}{1152N^2V^2} \right).
\end{align*}
Note that $|\Gamma_k^j| = T_k^j - T_k^{j-1} = cT_k^{j-1} - cT_k^{j-2} = c|\Gamma_k^{j-1}|$ for any $k \in \{0, 1, \ldots, v^1-4\}$ and $j \in \{2, \ldots, d\}$, and that $|\Gamma_k^1| = c|\Gamma_{k-1}^d|$ for any $k \in \{1, 2, \ldots, v^1-4\}$.  We also note that $|\Gamma_0^1| = (c-1)T_0 = T_b$. Thus,
\begin{align*}
    & \sum_{k=0}^{v^1-4}\sum_{j=1}^d \exp\left( -\tfrac{|\Gamma_{k}^j|}{1152N^2V^2} \right) \\
    & = \sum_{s=0}^{(v^1-3)d-1}\exp\left(-\tfrac{c^sT_b}{1152N^2V^2}\right) \\
    & \le \sum_{s=0}^{\infty}\exp\left(-\tfrac{c^sT_b}{1152N^2V^2}\right)\\
    & =  \exp\left(-\tfrac{T_b}{1152N^2V^2}\right)\left(1+ \sum_{s=1}^{\infty}\exp\left(-\tfrac{(c^s-1)T_b}{1152N^2V^2}\right)\right). 
\end{align*}
It remains to prove $\sum_{s=1}^{\infty}\exp\left(-\frac{(c^s-1)T_b}{1152N^2V^2}\right) \le 1$.
Since $c^s-1 \ge c-1+(s-1)(c^2-c),\forall s\ge 1$, we have
\begin{align*}
    &  \sum_{s=1}^{\infty}\exp\left(-\tfrac{(c^s-1)T_b}{1152N^2V^2}\right)\\
    &\le \sum_{s=1}^{\infty}\exp\left(-\tfrac{(c-1)T_b}{1152N^2V^2}\right) \left(\exp\left(-\tfrac{(c^2-c)T_b}{1152N^2V^2}\right)\right)^{s-1}\\
    &\le \sum_{s=1}^{\infty}\left(\tfrac{1}{2}\right)^s = 1,
\end{align*}
where the second ``$\le$'' is because $\exp\left(-\frac{(c^2-c)T_b}{1152N^2V^2}\right)\le \exp\left(-\frac{(c-1)T_b}{1152N^2V^2}\right) \le \frac{1}{2}$ by the choice of $T_b$. 
\end{proof}

\subsection{Eliminating $v^1-2$} \label{sec:main-M1-3-proof}
%In this subsection we show that for $|M^1|\ge 2$, after a sufficiently long time, the frequency of bids in $\{0, 1, \ldots, v^1-3\}$ submitted by bidders in $M^1$ will decrease to a small constant level, with high probability (Corollary~\ref{cor:first-stage-end}).     
%We show this by partitioning the rounds from the beginning and using an induction from $0$ to $v^1-3$. 
%Let constants $c = 1+\frac{1}{12NV}$ and $d = \lceil \log_c (8NV) \rceil$. Let $T_b$ be any (constant) integer such that $\gamma_{T_b} \le \frac{1}{12N^2 V^2}$ and $\exp\left( -\frac{(c-1)T_b}{1152N^2V^2} \right) \le \frac{1}{2}$.
%We partition the times as follows to get a constant bound for bidding $0, 1, \ldots, v^1-3$ with high probability.
% We partition the rounds from the beginning as follows: 
%Let $T_0 = 12N V T_b$ and $T_k = c^d T_{k-1} = c^{dk}T_0 \ge (8N V)^kT_0$ for $k \in \{1, 2, \cdots, v^1-3\}$. Let $A_k$ be the event that $\frac{1}{T_k} \sum_{t=1}^{T_k} \mathbb{I}[\exists i\in M^1, b_t^i\le k]\le \frac{1}{4N V}$ holds.

Assume that $A_{v^1-3}$ has happened.  We continue partitioning the time horizon after $T_{v^1-3}$, all the way to infinity, to show two points: (1) the frequency of bids in $\{0, 1, \ldots, v^1-3\}$ from bidders in $M^1$ approaches $0$; (2) the frequency of $v^1-2$ also approaches $0$. 
% Point (1) is also used in the proof of Theorem~\ref{thm:main-M1-2}.
%In this subsection, we continue partitioning the rounds after $T_{v^1-3}$ to get an approaching $0$ bound of bidding $\{0, 1, \ldots, v^1-3\}$ with high probability, which is useful for the proof of Theorem \ref{thm:main-M1-3} and \ref{thm:main-M1-2}. In the meanwhile, we get an approaching $0$ bound of bidding $v^1-2$ for the case of $|M^1|\ge 3$, which prove the theorem \ref{thm:main-M1-3}. 
Again let $ c=1+\frac{1}{12NV}$. Let $T_a^0=T_{v^1-3}, T_a^{k+1}=c T_a^k, \Gamma_a^{k+1}=[T_a^{k}+1, T_a^{k+1}], k\ge 0$.
Let $\delta_t=(\frac{1}{t})^{\frac{1}{3}}, t\ge 0$. For each $k\ge 0$, define
\[F_{T_a^k}= \tfrac{1}{4NV c^k}+\sum_{s=0}^{k-1} \tfrac{c-1}{c^{k-s}}\delta_{T_a^s}+ \sum_{s=0}^{k-1}|M^1|V \tfrac{c-1}{c^{k-s}} \gamma_{T_a^s}, \] 
and
\[\widetilde{F}_{T_a^k}=\tfrac{1}{c^k} +\sum_{s=0}^{k-1} \tfrac{c-1}{c^{k-s}}\delta_{T_a^s}+ \sum_{s=0}^{k-1}|M^1|V \tfrac{c-1}{c^{k-s}} \gamma_{T_a^s}. \] 
% and let $F_t=c F_{T_a^k}$, for every $t \in \Gamma_a^k$, $k\ge 1$.

\begin{claim}\label{claim:F_k-bound}
If $T_b$ is sufficiently large such that $\delta_{T_a^k}+|M^1|V\gamma_{T_a^k}\le \frac{1}{4NV}$, then 
$F_{T_a^{k+1}} \le F_{T_a^k}\le \frac{1}{4NV}$ for every $k\ge 0$ and $\lim_{k\to\infty}F_{T_a^k} = \lim_{k \to\infty}\widetilde{F}_{T_a^k} = 0$.
\end{claim}

\begin{lemma}\label{lemma:second-stage-M1=3}
Suppose $|M^1|\ge 2$.  Let $T_b$ be any sufficiently large constant. 
Let $A_a^k$ be the event 
\begin{equation*}
    \bigg[ \forall s\le k, ~ \frac{1}{T_a^s} \sum_{t=1}^{T_a^s} \mathbb{I}[\exists i\in M^1, b_t^i\le v^1-3]\le F_{T_a^s} \bigg]
\end{equation*}
Then, $\Pr[A_a^k]\ge 1 - \exp\left(-\frac{T_b}{24NV}\right) - 2\exp\left(-\frac{T_b}{1152N^2V^2} \right) - 2\exp\left(-(\frac{ T_b}{1152N^2V^2})^{\frac{1}{3}}\right)$.
Moreover, if $|M^1|\ge 3$, we can include the following in event $A_a^k$: $\forall s\le k$, $\frac{1}{T_a^s} \sum_{t=1}^{T_a^s} \mathbb{I}[\exists i\in M^1, b_t^i\le v^1-2]\le \widetilde{F}_{T_a^s}$.
\end{lemma}

\noindent The proof of Lemma~\ref{lemma:second-stage-M1=3} is similar to Lemma~\ref{lemma:first-stage} except that we use Lemma~\ref{lemma:mean-based-v-2} to argue that bidders bid $v^1-2$ with low frequency. See details in Appendix \ref{proof:lemma:second-stage-M1=3}. 

% \begin{corollary}\label{cor:second-stage-M=3}
% $\Pr[A_a^k]\ge 1 - \exp\left(-\frac{T_b}{24NV}\right) - 2\exp\left(-\frac{T_b}{1152N^2V^2} \right) - 2\exp\left(-(\frac{ T_b}{1152N^2V^2})^{\frac{1}{3}}\right)$.
% \end{corollary}
~\\
\noindent \textbf{Proof of Theorem~\ref{thm:main-M1-3}.} Suppose $|M^1|\ge 3$.  We note that the event $A_a^k$ implies that for any time $t\in \Gamma_a^{k} = [T_a^{k-1}+1, T_a^k]$, 
\begin{align} 
    & \frac{1}{t}\sum_{s=1}^t \mathbb{I}[\exists i  \in M^1, b_s^i\le v^1-2] \nonumber \\
    & \le \frac{1}{t}\sum_{s=1}^{T_a^k} \mathbb{I}[\exists i  \in M^1, b_s^i\le v^1-2]  \nonumber \\
    & \le \frac{1}{t} T_a^k \widetilde F_{T_a^k} \nonumber \\
    & \le c \widetilde F_{T_a^k} \quad ~~ \text{(because $t \ge \tfrac{1}{c} T_a^k$).} \label{eq:M=3-t-frequency-bound}
\end{align}
We note that $A_a^{k-1} \supseteq A_a^k$, so by Lemma~\ref{lemma:second-stage-M1=3} with probability at least $\Pr[\cap_{k=0}^\infty A_a^k] = \lim_{k\to\infty} \Pr[A_a^k] \ge 1 - \exp\left(-\frac{T_b}{24NV}\right) - 2\exp\left(-\frac{T_b}{1152N^2V^2} \right) - 2\exp\left(-(\frac{ T_b}{1152N^2V^2})^{\frac{1}{3}}\right)$ all events $A_a^0, A_a^1, \ldots, A_a^k, \ldots$ happen. 
Then, according to \eqref{eq:M=3-t-frequency-bound} and Claim~\ref{claim:F_k-bound}, we have
\begin{equation*}
\lim_{t\to\infty} \frac{1}{t}\sum_{s=1}^t \mathbb{I}[\exists i  \in M^1, b_s^i\le v^1-2] \le \lim_{k\to\infty} c \widetilde F_{T_a^k} = 0.
\end{equation*}
Letting $T_b\to \infty$ proves the first result of the theorem.  The second result follows from the observation that, when $\frac{1}{t}\sum_{s=1}^t \mathbb{I}[\exists i \in M^1, b_s^i\le v^1-2] \le \frac{1}{3NV}$, all bidders in $M^1$ will choose bids in $\{0, 1, \ldots, v^1-2\}$ with probability at most $(v^1-1)\gamma_{t+1}$ in round $t+1$ according to Lemmas~\ref{lemma:mean-based} and \ref{lemma:mean-based-v-2}, and that $(v^1-1)\gamma_{t+1}\to 0$ as $t\to\infty$.

\section{Experimental Results}
\label{sec:body-exp}
% \subsection{A Same Algorithm May Converge to Two Equilibria When $|M^1|=2$}
% Code for the experiments can be found at \url{https://github.com/tao-l/FPA-mean-based}. 

\subsection{$|M^1|=2$: Convergence to Two Equilibria}

% \begin{figure}
%     \centering
    
%     \begin{subfigure}[b]{0.52\columnwidth}
%     \centering
%     \includegraphics[width=\textwidth]{figures/M=2/Greedy_v1-1_player-1_frequency.png}
%     \caption{Player 1's bid frequency}
%     \end{subfigure}
%     \hfill 
%     \hspace{-3em}
%     \begin{subfigure}[b]{0.52\columnwidth}
%     \centering
%     \includegraphics[width=\textwidth]{figures/M=2/Greedy_v1-1_player-2_frequency.png}
%     \caption{Player 2's bid frequency}
%     \end{subfigure}
    
%     \vspace{1em}
    
%     \begin{subfigure}[b]{0.52\columnwidth}
%     \centering
%     \includegraphics[width=\textwidth]{figures/M=2/Greedy_v1-1_player-1_strategy.png}
%     \caption{Player 1's mixed strategy}
%     \end{subfigure}
%     \hfill 
%     \hspace{-3em}
%     \begin{subfigure}[b]{0.52\columnwidth}
%     \centering
%     \includegraphics[width=\textwidth]{figures/M=2/Greedy_v1-1_player_2_strategy.png}
%     \caption{Player 2's mixed strategy}
%     \end{subfigure}
    
%     \caption{$|M^1|=2$, $\eps_t$-Greedy, $v^1=v^1=4$, converging to $v^1-1 = 3$.  The four curves in each plot show (a) (b) the changes of frequencies of bids $0, 1, 2, 3$ and (c) (d) the changes of mixed strategies, in one simulation.  The frequency of $3$ approaches $1$.  The two regions show the $[10\%, 90\%]$-confidence intervals of the corresponding curves (the upper is for bid $3$, the lower is for bid $2$), among all simulations that converge to $v^1-1$.}
%     \label{fig:M=2-eps-greedy_v-1}
% \end{figure}

\begin{figure}
    \centering
    \includegraphics[width=\textwidth]{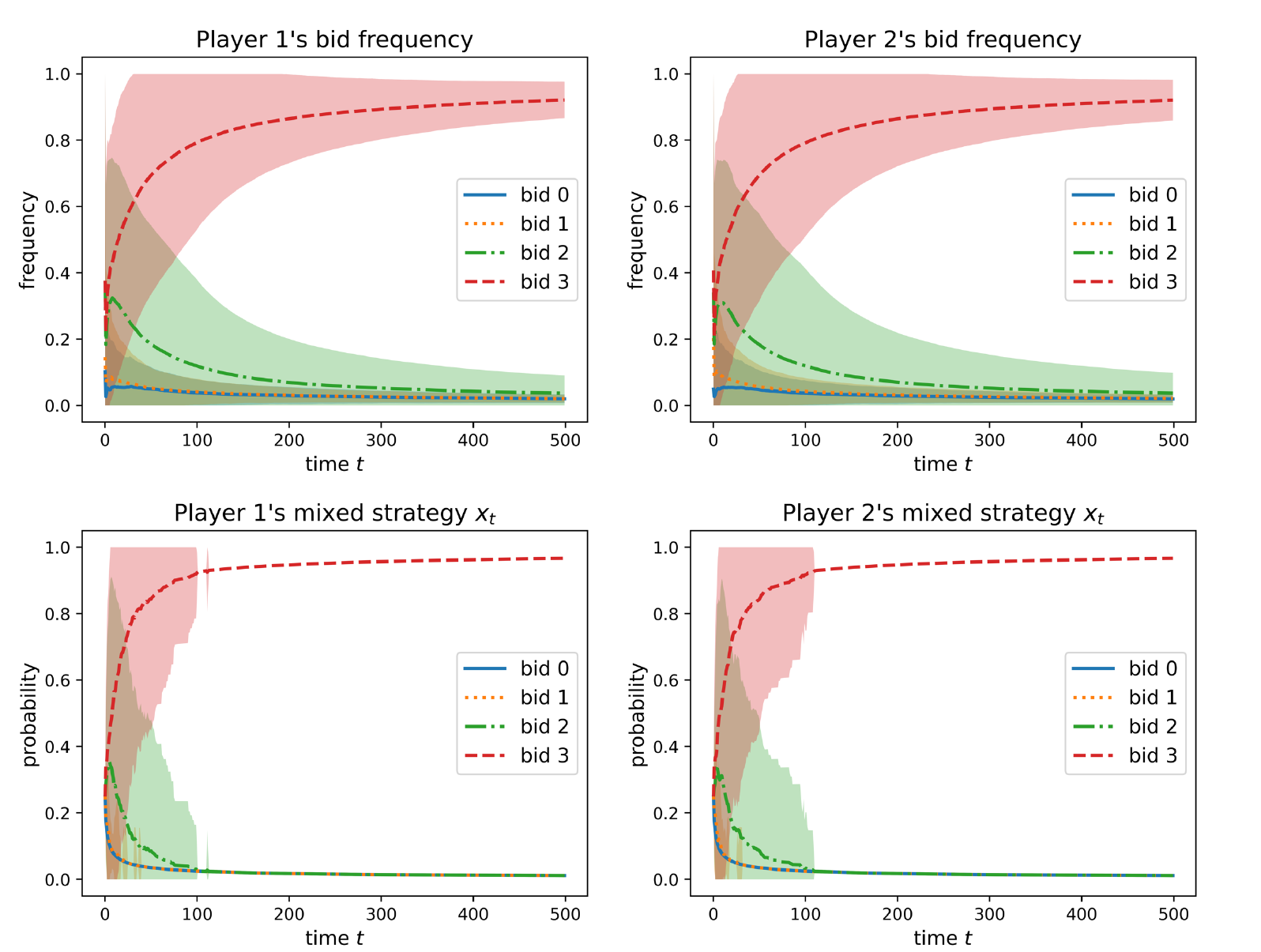}
    \caption{Player 1 and 2's bid frequencies and mixed strategies in the case of $|M^1|=2$, $v^1=v^2=4$, using $\eps_t$-Greedy algorithm, and converging to the $v^1-1=3$ equilibrium. Converging to $v^1-1$ happens in 774 out of 1000 simulations. The curves and the shaded regions are the means and $2$-standard deviation intervals among the 774 simulations.}
     \label{fig:M=2-eps-greedy_v-1}
\end{figure}

For the case of $|M^1|=2$, we showed in Theorem~\ref{thm:main-M1-2} that any mean-based algorithm must converge to one of the two equilibria where the two players in $M^1$ bid $v^1-1$ or $v^1-2$.
% One may wonder whether a fixed mean-based algorithm always converges to a fixed equilibrium or it may converge to different equilibria when run for multiple times.    
One may wonder whether there is a theoretical guarantee of which equilibrium will be obtained. 
We give experimental results to show that, in fact, \emph{both} equilibria can be obtained under a \emph{same} randomized mean-based algorithm in different runs.  We demonstrate this by the $\eps_t$-Greedy algorithm (defined in Example~\ref{ex:mean-based-algorithms}).  Interestingly, 
under the same setting, the MWU algorithm
%in another experiment of MWU algorithm under the same setting, it 
always converges to the equilibrium of $v^1-1$. 
% the latter case is true. 
%
In the experiment, we let $n=|M^1|=2$, $v^1=v^2=V=4$. 

\paragraph{$\eps_t$-Greedy converges to two equilibria}
We run $\eps_t$-Greedy with $\eps_t=\sqrt{1/t}$ for 1000 simulations.
% \begin{algorithm}
% \caption{Experts-setting $\eps_t$-Greedy algorithm with decreasing $\eps_t$}\label{alg:decreasing-eps-greedy}
% \begin{algorithmic}[1]
% \REQUIRE $i$, player $i$'s value $v^i$, and a decreasing sequence $(\eps_t)_{t=1}^\infty$ 
% \FOR{$t=1, 2, \ldots, $}{
%  \STATE With probability $1-\eps_t$, choose bid $b_t^i = \argmax_{b\in \mathcal B^i}\sum_{s=1}^t r^i_{b, s}$ (if there are multiple argmax bids, choose an arbitrary one); with probability $\eps_t$, choose a bid in $\mathcal B^i=\{0, 1, \ldots, v^i-1\}$ uniformly at random. 
%  \STATE Receive the reward $r_{b', t}^i = u^i(b', \bm b_t^{-i})$ for each bid $b'\in \mathcal B^i$. 
% }
% \ENDFOR
% \end{algorithmic}
% \end{algorithm}
% 
In each simulations, we run it for $T=1000$ rounds.  After it finishes, we use the frequency of bids from bidder $1$ to determine which equilibrium the algorithm will converge to: if the frequency of bid $2$ is above $0.8$, we consider it converging to the equilibrium of $v^1-2$; if the frequency of bid $3$ is above $0.8$, we consider it converging to the equilibrium of $v^1-1$; if neither happens, we consider it as ``not converged yet''. Among the $1000$ simulations, we found $774$ times of converging to $v^1-1$, $226$ times of converging to $v^1-2$, and $0$ times of ``not converged yet''. % ; namely, the probability of converging to $v^1-2$ is roughly $87\%$. 

We give two figures of the changes of bid frequencies and mixed strategies of player $1$ and $2$: Figure~\ref{fig:M=2-eps-greedy_v-1} is for $v^1-1$; Figure~\ref{fig:M=2-eps-greedy_v-2} is for the case of converging to $v^1-2$.  The x-axis is round number $t$ and the y-axis is the frequency $\frac{1}{t}\sum_{s=1}^t \mathbb{I}[b_s^i=b]$ of each bid $b\in\{0, 1, 2, 3\}$ or the mixed strategy $\bm x_t^i = (x_t^i(0), x_t^i(1), x_t^i(2), x_t^i(3))$.  For clarity, we only show the first $500$ rounds.  

\begin{figure}
    \centering
    \includegraphics[width=\textwidth]{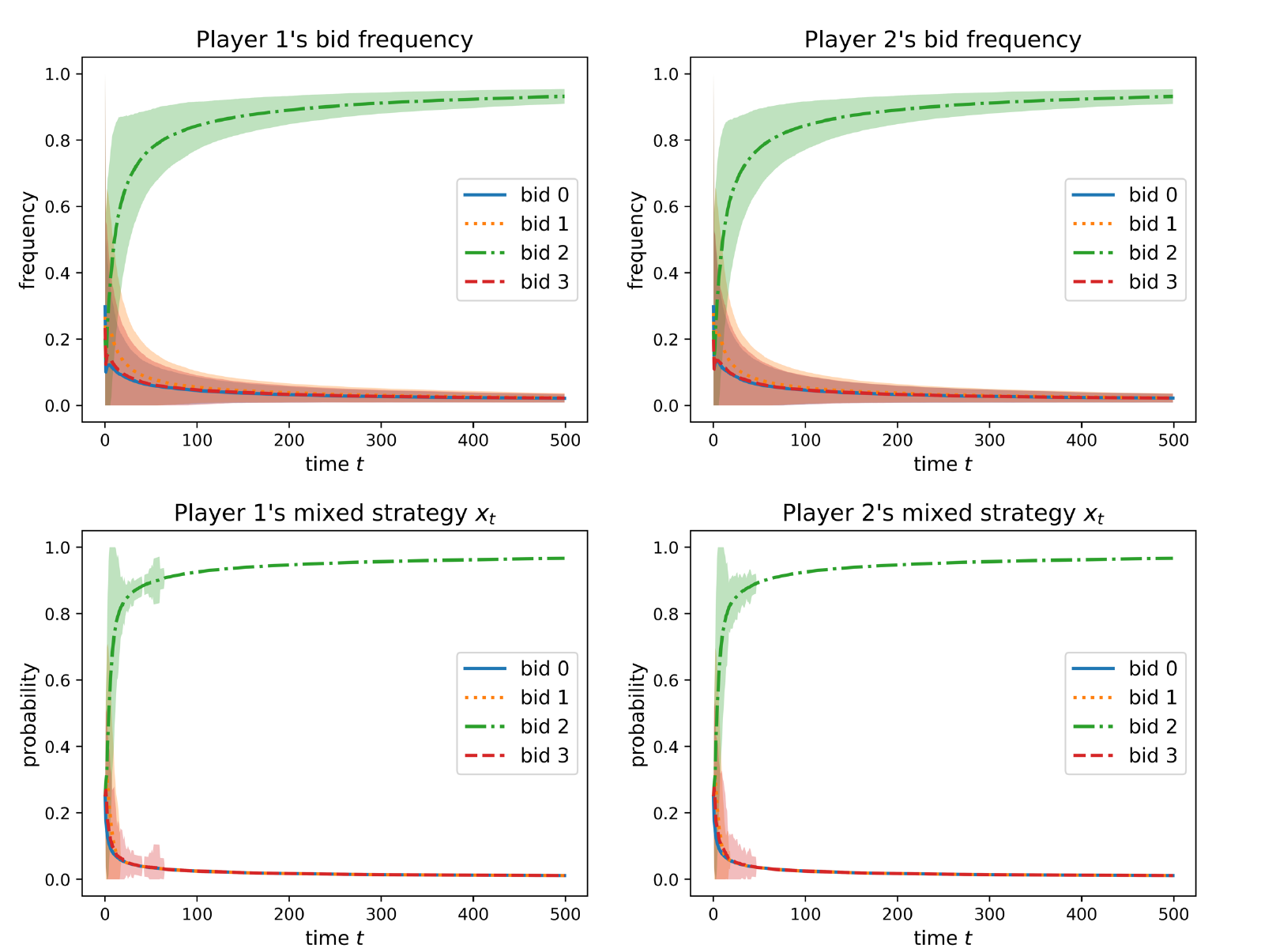}
    \caption{Player 1 and 2's bid frequencies and mixed strategies in the case of $|M^1|=2$, $v^1=v^2=4$, using $\eps_t$-Greedy algorithm, and converging to the $v^1-2=2$ equilibrium. Converging to $v^1-2$ happens in 226 out of 1000 simulations. The curves and the shaded regions are the means and $2$-standard deviation intervals among the 226 simulations.}
    \label{fig:M=2-eps-greedy_v-2}
\end{figure}

\paragraph{MWU always converges to $v^1-1$}
We run MWU with $\eps_t=\sqrt{1/t}$. 
%See \citet{braverman_selling_2018} for why the algorithm is $\gamma_t$-mean-based.
% \begin{algorithm}
% \caption{Multiplicative Weights Update algorithm with decreasing $\eps_t$}\label{alg:decreasing-MWU}
% \begin{algorithmic}[1]
% \REQUIRE $i$,player $i$'s value $v^i$, and a decreasing sequence $(\eps_t)_{t=1}^\infty$ 
% \FOR{$t=1, 2, \ldots$}{
%  \STATE Choose bid $b_t^i = b$ with probability $\frac{\exp(\eps_t \sum_{s=1}^t r^i_{b, s})}{\sum_{b'\in\mathcal B^i} \exp(\eps_t \sum_{s=1}^t r^i_{b', s})}$. 
%  \STATE Receive the reward $r_{b', t}^i = u^i(b', \bm b_t^{-i})$ for each bid $b'\in \mathcal B^i$. 
% }
% \ENDFOR
% \end{algorithmic}
% \end{algorithm}
Same as the previous experiment, we run the algorithm for $1000$ simulations and count how many times the algorithm converges to the equilibrium of $v^1-2$ and $v^1-1$.  We found that, in all $1000$ simulations, MWU converged to $v^1-1$. 
Figure~\ref{fig:M=2-MWU} shows the changes of bid frequencies and mixed strategies of both players.

% \begin{figure}
%     \centering
    
%     \begin{subfigure}[b]{0.52\columnwidth}
%     \centering
%     \includegraphics[width=\textwidth]{figures/M=2/MWU_v1-1_player-1_frequency.png}
%     \caption{Player 1's bid frequency}
%     \end{subfigure}
%     \hfill 
%     \hspace{-3em}
%     \begin{subfigure}[b]{0.52\columnwidth}
%     \centering
%     \includegraphics[width=\textwidth]{figures/M=2/MWU_v1-1_player-2_frequency.png}
%     \caption{Player 2's bid frequency}
%     \end{subfigure}
    
%     \vspace{1em}
    
%     \begin{subfigure}[b]{0.52\columnwidth}
%     \centering
%     \includegraphics[width=\textwidth]{figures/M=2/MWU_v1-1_player-1_strategy.png}
%     \caption{Player 1's mixed strategy}
%     \end{subfigure}
%     \hfill 
%     \hspace{-3em}
%     \begin{subfigure}[b]{0.52\columnwidth}
%     \centering
%     \includegraphics[width=\textwidth]{figures/M=2/MWU_v1-1_player_2_strategy.png}
%     \caption{Player 2's mixed strategy}
%     \end{subfigure}
    
%     \caption{$|M^1|=2$, MWU, $v^1=v^1=4$, converging to $v^1-1 = 3$.  The four curves in each plot show (a) (b) the changes of frequencies of bids $0, 1, 2, 3$ and (c) (d) the changes of mixed strategies, in one simulation.  The frequency of $3$ approaches $1$.  The two regions show the $[10\%, 90\%]$-confidence intervals of the corresponding curves (the upper is for bid $3$, the lower is for bid $2$), among 1000 simulations.}
%     \label{fig:M=2-MWU}
% \end{figure}

\begin{figure}
    \centering
    \includegraphics[width=\textwidth]{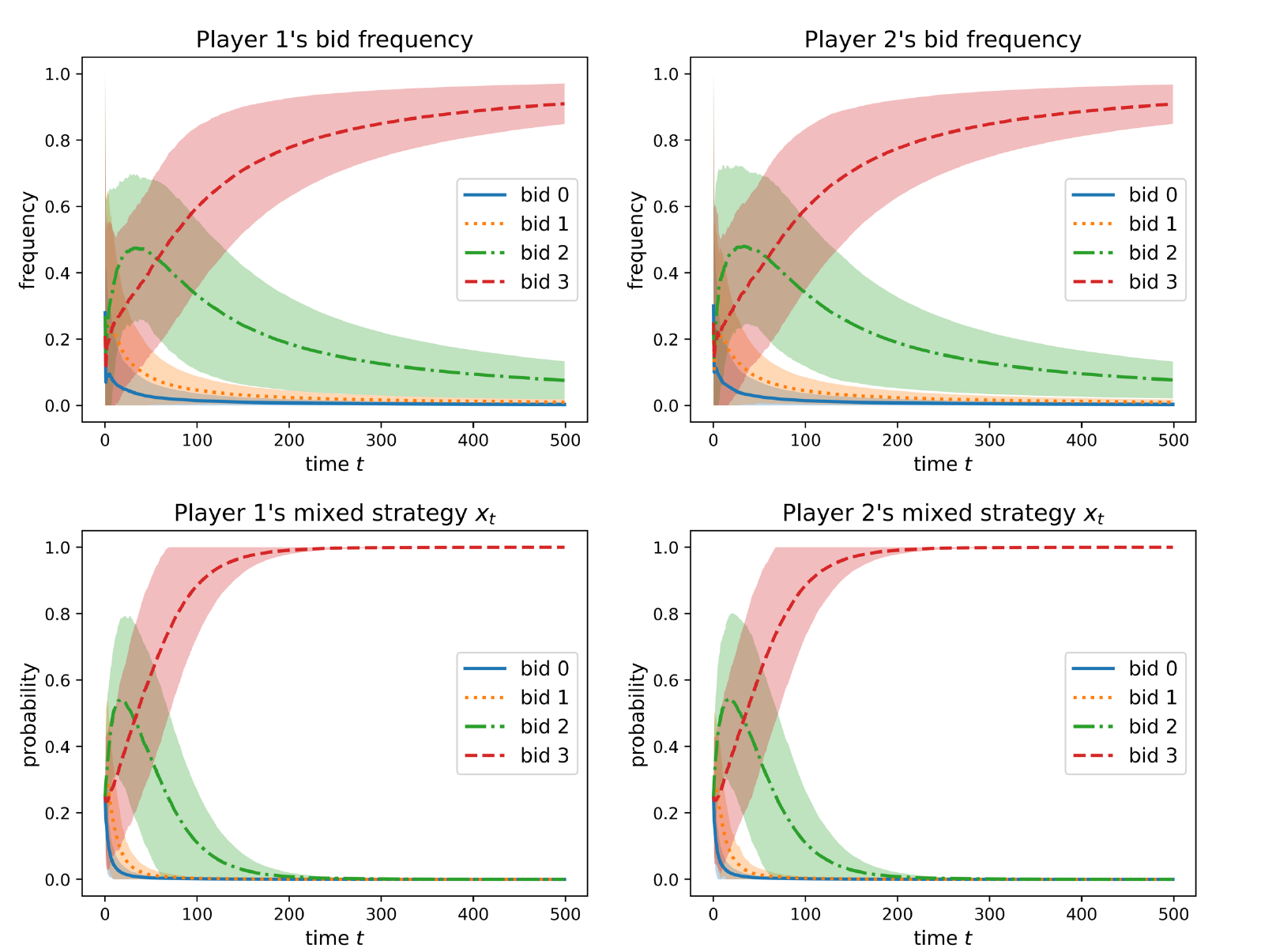}
    \caption{Player 1 and 2's bid frequencies and mixed strategies in the case of $|M^1|=2$, $v^1=v^2=4$, using MWU algorithm. The curves and the shaded regions are the means and $2$-standard deviation intervals of 1000 simulations.}
     \label{fig:M=2-MWU}
\end{figure}

\subsection{$|M^1|=1$: Non-Convergence}
For the case of $|M^1|=1$ we showed that not all mean-based algorithms can converge to equilibrium, using the example of Follow the Leader (Example~\ref{ex:M1=1}).
Here we experimentally demonstrate that such non-convergence phenomena can also happen with more natural (and even no-regret) mean-based algorithms like $\eps$-Greedy and MWU.

In the experiment we let $n=2$, $v^1=8, v^2=6$.  We run $\eps_t$-Greedy and MWU both with $\eps_t=1/\sqrt{t}$ for $T=20000$ rounds. 

% For $\eps_t$-Greedy, Figure~\ref{fig:M=1-eps-greedy} shows that: (a) player 1's frequency of bidding $v^2=6$ seems to converge to $1$; but (b) player 2's bid frequency oscillates; (c) player 1's mixed strategy does not last-iterate converge; (d) player 2 switches between bids $3$ and $5$.  Intuitively, this phenomenon is because: in the $\eps_t$-Greedy algorithm, when player $1$ bids $v^2=6$ with high probability, she also sometimes chooses bids uniformly at random, in which case the best response for player $2$ is to bid $v^2/2=3$; but after player $2$ switches to $3$, player $1$ will find it beneficial to lower her bid from $6$ to $5$; then player $2$ will switch to $5$ to win the item with $1/2$ probability; but then player $1$ will increase to $6$ to outbid player 2; ...; hence entering a cycle. 

For $\eps_t$-Greedy, Figure~\ref{fig:M=1-eps-greedy} shows that the two bidders do not converge to a pure-strategy equilibrium, either in time-average or last-iterate.  According to Proposition~\ref{prop:Nash}, a pure-strategy equilibrium must have bidder $1$ bidding $v^2 = 6$ and bidder $2$ bidding $v^2 - 1 = 5$.  But figure (b) shows that bidder $2$'s frequency of bidding $5$ does not converge to $1$.  The frequency oscillates and we do not know whether it will stabilize at some limit less than $1$.
Looking closer, we see that bidder $2$ constantly switches between bids $5$ and $3$, and bidder $1$ switches between $5$ and $6$.
Intuitively, this is because: in the $\eps_t$-Greedy algorithm, when bidder $1$ bids $v^2=6$ with high probability, she also sometimes (with probability $\eps_t$) chooses bids uniformly at random, in which case the best response for bidder $2$ is to bid $v^2/2=3$; but after bidder $2$ switches to $3$, bidder $1$ will find it beneficial to lower her bid from $6$ to $5$; then, bidder $2$ will switch to $5$ to compete with bidder $1$, winning the item with probability $1/2$; but then bidder $1$ will increase to $6$ to outbid bidder $2$; ...  In this way, they enter a cycle. 

% \begin{figure}
%     \centering
%     \begin{subfigure}[b]{0.51\columnwidth}
%     \centering
%     \includegraphics[width=\textwidth]{new_figures/eps_greedy_8_6_player_1_frequency.png}
%     \caption{Player 1's bid frequency}
%     \end{subfigure}
%     \hfill 
%     \hspace{-3em}
%     \begin{subfigure}[b]{0.51\columnwidth}
%     \centering
%     \includegraphics[width=\textwidth]{new_figures/eps_greedy_8_6_player_2_frequency.png}
%     \caption{Player 2's bid frequency}
%     \end{subfigure}
%     \begin{subfigure}[b]{0.51\columnwidth}
%     \vspace{1em}
%     \centering
%     \includegraphics[width=\textwidth]{new_figures/eps_greedy_8_6_player_1_strategy.png}
%     \caption{Player 1's mixed strategy}
%     \end{subfigure}
%     \hfill 
%     \hspace{-3em}
%     \begin{subfigure}[b]{0.51\columnwidth}
%     \centering
%     \includegraphics[width=\textwidth]{new_figures/eps_greedy_8_6_player_2_strategy.png}
%     \caption{Player 2's mixed strategy}
%     \end{subfigure}
%     \vspace{0.5em}
%     \caption{$|M^1|=1$, $\eps_t$-Greedy, $v^1=8, v^2=6$. 
%     (a) Player 1's frequency of bid $6$ seems to converge to $1$.
%     (b) Player 2's bid frequency oscillates.
%     (c) Player 1's mixed strategy does not last-iterate converge; she switches between bids $5$ and $6$.
%     (d) Player 2 switches between bids $3$ and $5$.
%     The curves show the results from one simulation.  The shaped regions are $2$-standard deviation intervals from 100 simulations. }
%     \label{fig:M=1-eps-greedy}
% \end{figure}

\begin{figure}
    \centering
    \includegraphics[width=\textwidth]{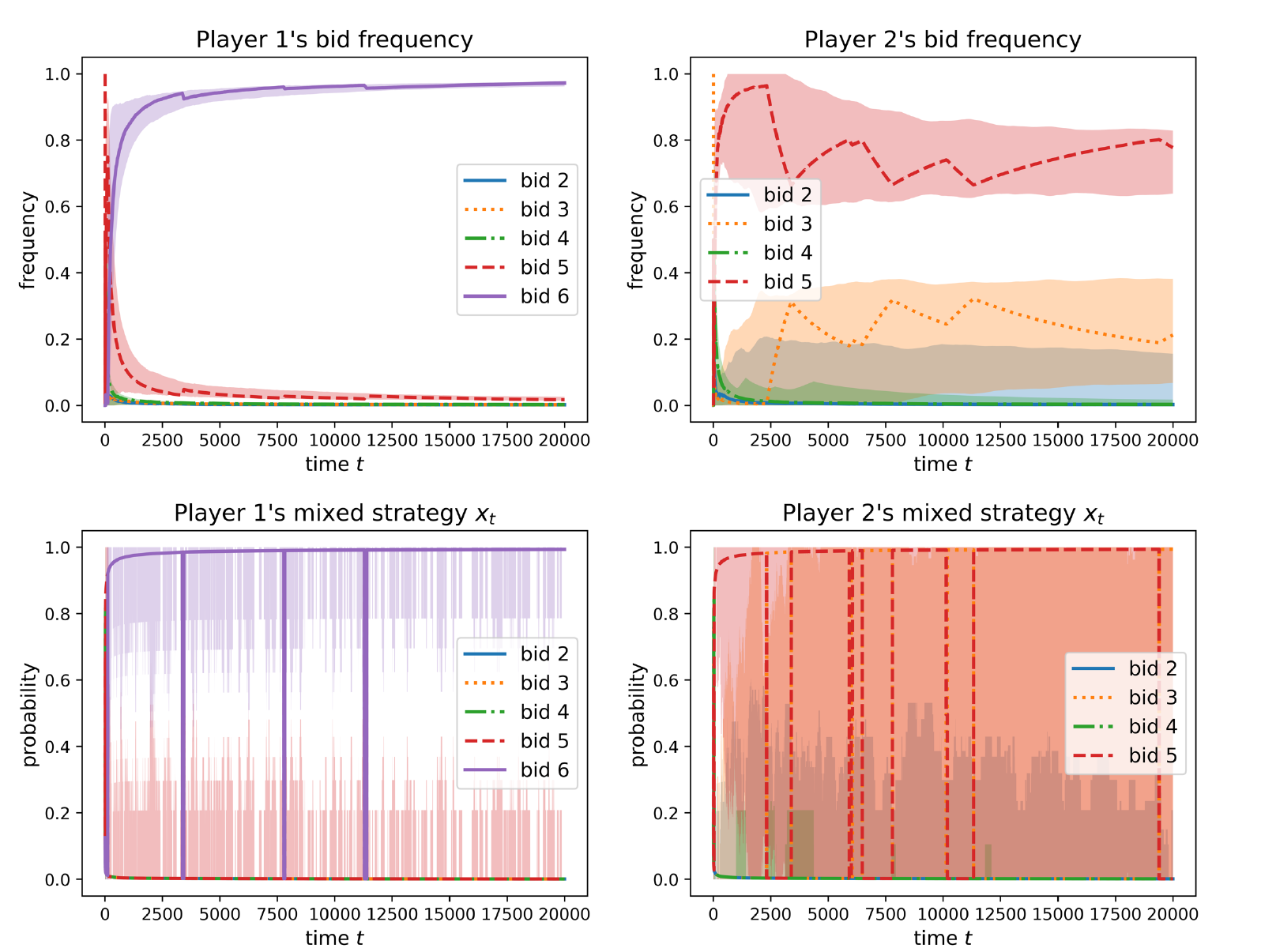}
    \caption{Player 1 and 2's bid frequencies and mixed strategies in the case of $|M^1|=1$, $v^1=8, v^2=6$, using $\eps_t$-Greedy algorithm.
    The curves are the results from one simulation. 
    The shaped regions are $2$-standard deviation intervals from 100 simulations.
    Player 1's frequency of bid $6$ seems to converge to $1$, but the mixed strategy does not last-iterate converge; it switches between bids $5$ and $6$.
    Player 2's bid frequency oscillates; the mixed strategy switches between bids $3$ and $5$.}
    \label{fig:M=1-eps-greedy}
\end{figure}

For MWU, Figure~\ref{fig:M=1-MWU} shows that the two bidders do not converge to a pure-strategy Nash equilibrium.  In particular, the left two figures show that bidder 1 converges to bidding $v^2=6$, but the right two figures show that bidder 2 seems to converge to a mixed strategy. 
% to do not seem to converge. 
%\tao{Bidder $2$ seems to converge to a mixed strategy. }

\begin{figure}
    \centering
    \includegraphics[width=\textwidth]{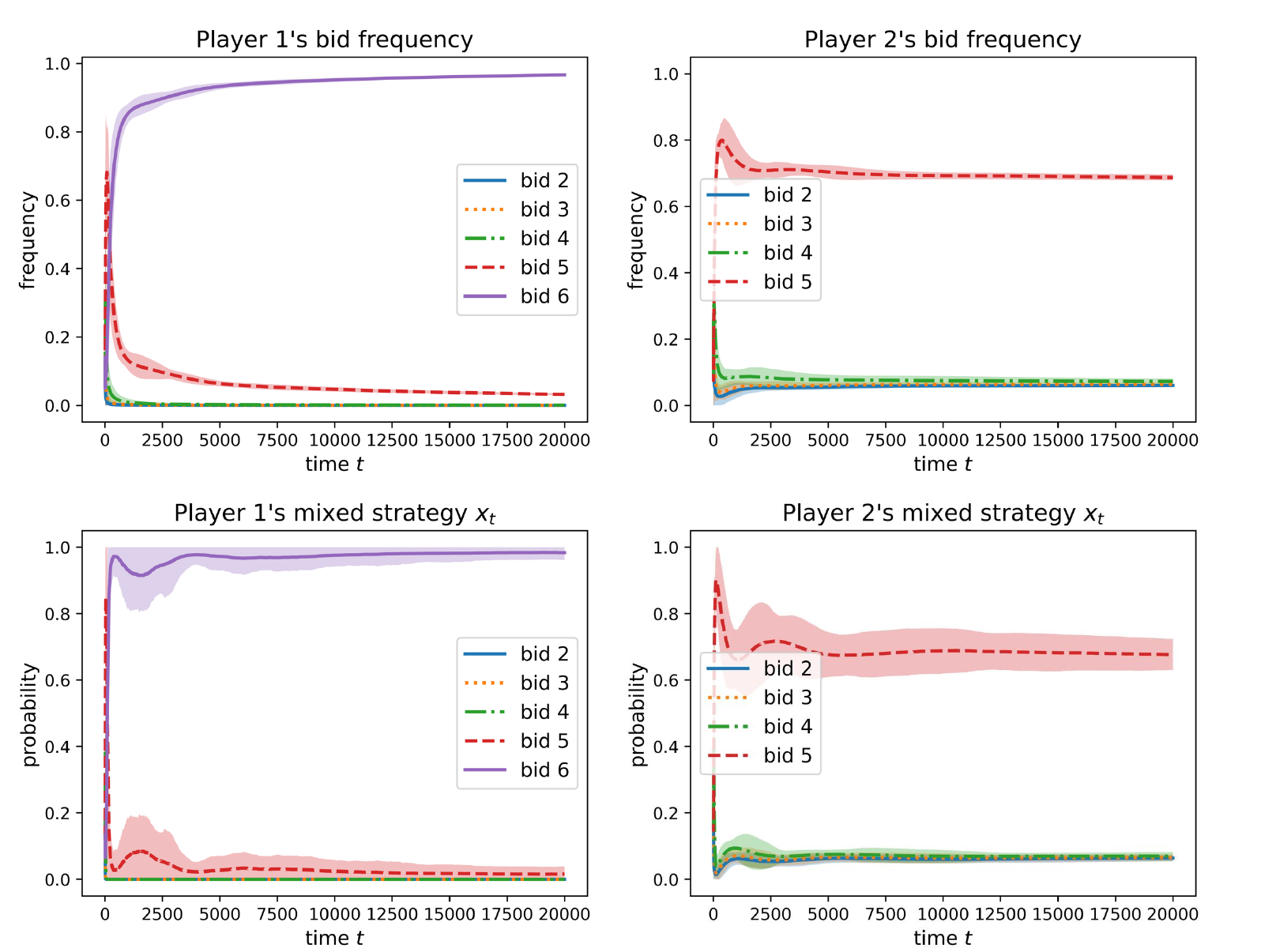}
    \caption{Player 1 and 2's bid frequencies and mixed strategies in the case of $|M^1|=1$, $v^1=8, v^2=6$, using MWU algorithm. The curves and the shaded regions are the means and $2$-standard deviation intervals of 100 simulations. }
    \label{fig:M=1-MWU}
%     \centering
%     \begin{subfigure}[b]{0.52\columnwidth}
%     \centering
%     \includegraphics[width=\textwidth]{new_figures/MWU_8_6_player_1_frequency.png}
%     \subcaption{Player 1's bid frequency}
%     \end{subfigure}
%     \hfill 
%     \hspace{-3em}
%     \begin{subfigure}[b]{0.52\columnwidth}
%     \centering
%     \includegraphics[width=\textwidth]{new_figures/MWU_8_6_player_2_frequency.png}
%     \caption{Player 2's bid frequency}
%     \end{subfigure}
%     \begin{subfigure}[b]{0.52\columnwidth}
%     % \vspace{1em}
%     \centering
%     \includegraphics[width=\textwidth]{new_figures/MWU_8_6_player_1_strategy.png}
%     \caption{Player 1's mixed strategy}
%     \end{subfigure}
%     \hfill 
%     \hspace{-3em}
%     \begin{subfigure}[b]{0.52\columnwidth}
%     \centering
%     \includegraphics[width=\textwidth]{new_figures/MWU_8_6_player_2_strategy.png}
%     \caption{Player 2's mixed strategy}
%     \end{subfigure}
%     \caption{$|M^1|=1$, MWU, $v^1=8, v^2=6$. % Player 1's bid frequency (a) and mixed strategy (c) seem to converge to bidding $6$.
%     % But player 2's bid frequency (b) and mixed strategy (d) do not seem to converge.
%     The curves and the shaded region show the means and $2$-standard deviation intervals of 100 simulations. }
%     \label{fig:M=1-MWU}
\end{figure}

\section{Conclusions and Future Directions}
\label{sec:conclusion}
% In this work we show that, in repeated first-price auctions, mean-based learning bidders with fixed values converge to a Nash equilibrium of the auction if and only if the number of bidders with the highest value is at least two. 
%when bidders learn to bid using mean-based algorithms, they may or may not converge to the Nash equilibrium of the auction.
% Studying first-price auctions with varying values is a natural but challenging future direction given the complexity of equilibria of such auctions. 
% Studying varying values is a natural yet challenging future direction.
% given the complexity of equilibria of general first-price auctions.  

% Our convergence results apply to a wide class of no-regret algorithms but our non-convergence examples use mean-based algorithms that are not necessarily no-regret.  Our experiments suggest non-convergence for some no-regret algorithms.  It would be interesting for future works to prove non-convergence results for no-regret algorithms.  

In this work we showed that, in repeated fixed-value first-price auctions, mean-based online learning bidders converge to a Nash equilibrium in the presence of \emph{competition}, in the sense that at least two bidders share the highest value.
% Without competition, we gave non-convergence examples using mean-based algorithms that are not necessarily no-regret. 
In the case of no competition, we gave examples and experimental results showing that mean-based learning algorithms do not always converge to Nash equilibria.
There are several interesting open directions: 
\begin{itemize}
\item \textbf{The case of $|M^1|=1$.}
% in the study of learning dynamics in fixed-value first-price auctions.
% In fact, \cite{kolumbus_auctions_2022} show that some non-mean-based no-regret algorithms do not converge. 
% It is hence open to prove (non-)convergence for mean-based no-regret algorithms.
Although %our Example \ref{ex:M1=1} and experimental results show that
mean-based learning algorithms do not always converge to Nash equilibria in this case, we observed that the highest-value bidder's strategy still seems to ``converge'' in the sense that the bidder wins and pays the second highest value eventually. 
Understanding the convergence property of online learning algorithms in the absence of competition (namely, $|M^1| = 1$) is a natural and interesting future direction. 

\item \textbf{Convergence rate of mean-based learning dynamics.}
The convergence result we give is in the limit sense.  As observed by \cite{wu_multi-agent_2022}, many no-regret algorithms actually need an exponential time to converge to Nash equilibria in some iterative-dominance-solvable game. Our theoretical analysis for the first-price auction demonstrates a $T = O(c^{O(v^1)})$ upper bound on the convergence time for the case of $|M^1|=3$.  But the convergence time in our experiments is significantly shorter.  The exact convergence rate remains open. 

\item \textbf{Beyond fixed value setting.} As discussed in the Introduction, while the fixed value setting is well motivated, understanding the learning dynamics in first-price auctions in the Bayesian setting is also interesting yet challenging.  There are some recent progress in the symmetric value distribution setting \citep{ahunbay2025uniqueness}, while the general asymmetric setting remains open. 

\item \textbf{Multi-unit setting.} Our results focus on \emph{single-unit} first-price auctions. Understanding the conditions under which last-iterate convergence holds for learning dynamics in \emph{multi-unit} first-price (pay-as-bid) auctions is an interesting future direction. This question is also highlighted in a recent work by~\citet{galgana2025learning}, who studied optimization of bidding strategies from the perspective of a single bidder in multi-unit first-price auctions

\end{itemize}

\bibliographystyle{apalike}
\bibliography{bibfile}

%\end{multicols}

\appendix
\section{Missing Proofs from Section~\ref{sec:main-result}}\label{app:main-result-proof}
\subsection{Proof of Theorem~\ref{thm:main-M1-2}}
Suppose $|M^1|=2$. 
We will prove that, for any sufficiently large integer $T_b$, with probability at least $1 - \exp\big(-\frac{T_b}{24NV}\big) - 2\exp\big(-\frac{T_b}{1152N^2V^2} \big) - \frac{6}{e-2}\big(\frac{48NV}{T_b}\big)^{3e/4}$, one of following two events must happen: 
\begin{itemize}
    \item $\lim_{t\to\infty}\frac{1}{t}\sum_{s=1}^{t} \mathbb{I}[\forall i\in M^1, b_s^i=v^1-2]=1$;
    \item $\lim_{t\to\infty}\frac{1}{t}\sum_{s=1}^{t} \mathbb{I}[\forall i\in M^1, b_s^i=v^1-1]=1$ and $\lim_{t\to\infty}\Pr[b_t^i=v^1-1]=1$.
\end{itemize}
And if $n\ge 3$ and $v^3=v^1-1$, only the second event happens. 
Letting $T_b\to\infty$ proves Theorem~\ref{thm:main-M1-2}.

% \paragraph{Proof Sketch} ==== to write =====
%We first show that if the frequency of bidding $v^1-1$ of any bidder $i\in M^1$ is beyond a threshold at some time after $T_a=T_{v^1-3}$, then both the time-average frequency and the last-iterate probability of bidding $v^1-1$ will converge to 1 with high probability (as shown in Case 1). Furthermore, if the frequency of some bidder $i\in M^1$ is beyond a larger threshold at some time after $T_a=T_{v^1-3}$, which is Case 2, then it will satisfy the assumption of Case 1 after one stage of time.

%Otherwise if the frequency of bidding $v^1-1$ of any bidder $i\in M^1$ is always below the threshold (Case 3), the frequency of bidding $v^1-2$ of any bidder $i\in M^1$ will converge to 1 since the threshold converges to $0$. 

%Moreover, if $n\ge3$ and $v^3=v^1-1$, then the assumption of Case 2 will always hold.

We reuse the argument in Section~\ref{sec:first-stage}.  Assume $v^1\ge 3$.\footnote{If $v^1=1$, Theorem~\ref{thm:main-M1-2} trivially holds.  If $v^1=2$, we let $T_{v^1-3} = T_0 = T_b$; $A_{v^1-3}$ holds with probability $1$ since $\frac{1}{T_{v^1-3}} \sum_{t=1}^{T_{v^1-3}} \mathbb{I}[\exists i\in M^1, b_t^i\le v^1-3] = 0$; the argument for $v^1\ge 3$ will still apply. }   Recall that we defined $ c=1+\frac{1}{12NV}$, $d = \lceil \log_c (8NV) \rceil$; $T_b$ is any integer such that $\gamma_{T_b} < \frac{1}{12N^2 V^2}$ and $\exp\left( -\frac{(c-1)T_b}{1152N^2V^2} \right) \le \frac{1}{2}$;
$T_0 = 12N V T_b$; $T_{v^1-3} = c^{(v^1-3)d}T_0$.
%, and $T_k = c^d T_{k-1} = c^{dk}T_0 \ge (8N V)^kT_0$ for $k \in \{1, 2, \cdots, v^1-3\}$. 
We defined $A_{v^1-3}$ to be the event $\frac{1}{T_{v^1-3}} \sum_{t=1}^{T_{v^1-3}} \mathbb{I}[\exists i\in M^1, b_t^i\le v^1-3]\le \frac{1}{4N V}$.
According to Corollary~\ref{cor:first-stage-end}, $A_{v^1-3}$ holds with probability at least $1 - \exp\left(-\frac{T_b}{24NV}\right) - 2\exp\left(-\frac{T_b}{1152N^2V^2} \right)$.
Suppose $A_{v^1-3}$ holds. 

Now we partition the time horizon after $T_{v^1-3}$ as follows:
let $T_a^0=T_{v^1-3}, T_a^k = C(k+24NV)^2$, $\forall k\ge 0$, where $C=\frac{T_{v^1-3}}{(24NV)^2}$, so that $T_a^0 = C(0+24NV)^2$. Denote $\Gamma_a^{k+1} = [T_a^k+1, T_a^{k+1}]$, with $|\Gamma_a^{k+1}| = T_a^{k+1} - T_a^k$. (We note that the notations here have different meanings than those in Section~\ref{sec:main-M1-3-proof}.) 
We define $\delta_t=(\frac{1}{t})^{1/8}, t\ge 0$. For each $k\ge 0$, we define 
\begin{equation*}
    F_{T_a^k}=\frac{T_a^0}{T_a^k} \frac{1}{4NV} + \sum_{s=0}^{k-1} \frac{T_a^{s+1} - T_a^s}{T_a^k}\delta_{T_a^s}+ \sum_{s=0}^{k-1} \frac{T_a^{s+1} - T_a^s}{T_a^k}|M^1|V\gamma_{T_a^s}.
\end{equation*}
Let $A_a^k$ be event
\begin{equation*}
    A_a^k = \left[ \frac{1}{T_a^k} \sum_{t=1}^{T_a^k} \mathbb{I}[\exists i\in M^1, b_t^i\le v^1-3]\le F_{T_a^k} \right].
\end{equation*}
We note that $A_a^0 = A_{v^1-3}$ because $F_{T_a^0} = \frac{1}{4NV}$. 

In the proof we will always let $T_b$ to be sufficiently large.  This implies that all the times $T_0, T_{v^1-3}, T_a^0, T_a^k$, etc., are sufficiently large. 

\subsubsection{Additional Notations, Claims, and Lemmas}
\begin{claim}\label{claim:F_T_a_k}
When $T_b$ is sufficiently large, 
\begin{itemize}
    \item $F_{T_a^{k+1}} \le F_{T_a^k} \le \frac{1}{4NV}$ for every $k\ge 0$. 
    \item $\lim_{k\to\infty} F_{T_a^k} = 0$. 
\end{itemize}
\end{claim}
\begin{proof}
Since $\delta_{T_a^0} \to 0$ and $\gamma_{T_a^0} \to 0$ as $T_b\to\infty$, when $T_b$ is sufficiently large we have
\[F_{T_a^1}=\frac{T_a^0}{T_a^1} \frac{1}{4NV} + \frac{T_a^{1} - T_a^0}{T_a^1}\left(\delta_{T_a^0}+ |M^1|V\gamma_{T_a^0}\right) \le \frac{T_a^0}{T_a^1} \frac{1}{4NV} + \frac{T_a^{1} - T_a^0}{T_a^1}\frac{1}{4NV} = \frac{1}{4NV} = F_{T_a^0}. \]
% 
% \[ 
% F_{T_a^{k+1}} = \frac{T_a^k }{T_a^{k+1}} F_{T_a^k} + \frac{T_a^{k+1} - T_a^k}{T_a^{k+1}} \left(\delta_{T_a^k} + |M^1|V \gamma_{T_a^k}\right)
% \]
% 
% \[F_{T_a^{k}}=\frac{T_a^0}{T_a^k} \frac{1}{4NV} + \sum_{s=0}^{k-1} \frac{T_a^{s+1} - T_a^s}{T_a^k}\delta_{T_a^s}+ \sum_{s=0}^{k-1} \frac{T_a^{s+1} - T_a^s}{T_a^k}|M^1|V\gamma_{T_a^s}.  \]
% 
Since $\delta_{T_a^s}$ and $\gamma_{T_a^s}$ are both decreasing, we have
\begin{align*}
    F_{T_a^{k}}
    & > \sum_{s=0}^{k-1} \frac{T_a^{s+1} - T_a^s}{T_a^k}\delta_{T_a^s}+ \sum_{s=0}^{k-1} \frac{T_a^{s+1} - T_a^s}{T_a^k}|M^1|V\gamma_{T_a^s} \\
    & \ge \sum_{s=0}^{k-1} \frac{T_a^{s+1} - T_a^s}{T_a^k}\delta_{T_a^k}+ \sum_{s=0}^{k-1} \frac{T_a^{s+1} - T_a^s}{T_a^k}|M^1|V\gamma_{T_a^k} ~=~ \delta_{T_a^k} + |M^1|V \gamma_{T_a^k}.
\end{align*}
Thus,
\begin{align*}
    F_{T_a^{k+1}} &\stackrel{\text{by definition}}{=} \frac{T_a^k }{T_a^{k+1}} F_{T_a^k} + \frac{T_a^{k+1} - T_a^k}{T_a^{k+1}} \left(\delta_{T_a^k} + |M^1|V \gamma_{T_a^k}\right) < \frac{T_a^k }{T_a^{k+1}} F_{T_a^k} + \frac{T_a^{k+1} - T_a^k}{T_a^{k+1}} F_{T_a^k} = F_{T_a^k}.
\end{align*}

Then we prove $\lim_{k\to\infty} F_{T_a^{k}} = 0$. 
For every $ 0<\eps<\frac{1}{4NV}$, we can find $k$ sufficiently large such that $\delta_{T_a^{k}} \le \frac{\eps}{6}$, and $\gamma_{T_a^{k}} \le \frac{\eps}{6|M^1|V}$. For any $l \ge \lceil k/\eps \rceil$, we have $\frac{T_a^0}{T_a^{l}} \le \frac{T_a^{k}}{T_a^{l}} \le \frac{\eps}{6}$. Then 
\begin{align*}
    F_{T_a^{l}}&=\frac{T_a^0}{T_a^l} \frac{1}{4NV} + \sum_{s=0}^{l-1} \frac{T_a^{s+1} - T_a^s}{T_a^{l}}(\delta_{T_a^s}+|M^1|V\gamma_{T_a^s})\\
    &\le \frac{\eps}{3} + 2\sum_{s=0}^{k-1} \frac{T_a^{s+1} - T_a^s}{T_a^l} + \sum_{s=k}^{l-1} \frac{T_a^{s+1} - T_a^s}{T_a^{l}}(\delta_{T_a^k}+|M^1|V\gamma_{T_a^k})\\
    &\le \frac{\eps}{3} + 2\frac{T_a^k}{T_a^l} + \delta_{T_a^k}+|M^1|V\gamma_{T_a^k} \\
    &\le \frac{\eps}{3} + \frac{\eps}{3} + \frac{\eps}{3}  = \eps.
\end{align*}
Since $F_{T_a^k}$ is non-negative, we have $\lim_{k\to\infty} F_{T_a^k} = 0$. 
 \end{proof}

\begin{claim}\label{claim:new-delta-probability-sum}
$\sum_{s=0}^\infty \exp\left(-\frac{1}{2}|\Gamma_a^{s+1}|\delta_{T_a^s}^2\right) \le \frac{2}{e-2}\frac{1}{C^{3e/4}} \le \frac{2}{e-2}\big(\frac{48NV}{T_b}\big)^{3e/4}$.
%In particular, this expression goes to $0$ when $C\to\infty$. 
\end{claim}
\begin{proof}
Recall that $|\Gamma_a^{s+1}| = T_a^{s+1} - T_a^s$, $\delta_{T_a^s}^2 = (\frac{1}{T_a^s})^{1/8}$, and $T_a^s = C(s+24NV)^2$.  Hence, 
\begin{align*}
    \sum_{s=0}^\infty \exp\left(-\frac{1}{2}|\Gamma_a^{s+1}|\delta_{T_a^s}^2\right) & = \sum_{s=0}^\infty \exp\left(-\frac{1}{2}\big(T_a^{s+1} - T_a^s\big)\big(\frac{1}{T_a^s}\big)^{1/4}\right) \\
    & = \sum_{s=0}^\infty \exp\left(-\frac{1}{2}C\big(2(s+24NV)+1\big)\big(\frac{1}{C(s+24NV)^2}\big)^{1/4}\right) \\
    & \le \sum_{s=0}^\infty \exp\left(-C^{3/4}\big(s+24NV\big)\big(\frac{1}{s+24NV}\big)^{1/2}\right) \\
    & = \sum_{s=0}^\infty \exp\left(-C^{3/4}\sqrt{s+24NV}\right) \\
    & \le \sum_{x=2}^\infty \exp\left(-C^{3/4}\sqrt{x}\right) \\
    & \le \int_{x=1}^\infty \exp\left(-C^{3/4}\sqrt{x}\right) \dd x \\
    \text{(using $e^x \ge x^e$ for $x\ge 0$) } & \le \int_{x=1}^\infty \frac{1}{(C^{3/4}\sqrt{x})^e} \dd x = \frac{1}{C^{3e/4}} \cdot \frac{2}{e-2}.
\end{align*}
Substituting $C = \frac{T_{v^1-3}}{(24NV)^2} = \frac{c^{(v^1-3)d}12NV T_b}{(24NV)^2} \ge \frac{12 NV T_b}{(24NV)^2} = \frac{T_b}{48NV}$ proves the claim. 
 \end{proof}

\begin{claim}\label{claim:T_ratio_bound}
$\frac{T_a^k}{T_a^{k+1}} \ge 1 - \frac{2}{k+24NV}$. 
\end{claim}
\begin{proof}
By definition,
$
\frac{T_a^k}{T_a^{k+1}} = \frac{(k+24NV)^2}{(k+24NV+1)^2} = 1 - \frac{2(k+24NV)+1}{(k+24NV+1)^2} \ge 1 - \frac{2}{k+24NV+1} \ge 1 - \frac{2}{k+24NV}. 
$
 \end{proof}

\begin{claim}\label{claim:A_a^k-to-t}
When $A_a^k$ holds, we have, for every $t\in \Gamma_a^{k+1} = [T_a^k+1, T_a^{k+1}]$, 
$
\frac{1}{t-1}\sum_{s=1}^{t-1} \mathbb{I}[\exists i  \in M^1, b_s^i\le v^1-3]\le F_{T_a^k} + \frac{2}{k + 24NV} \le \frac{1}{2NV} - 2\gamma_t.
$
\end{claim}
\begin{proof}
When $A_a^k$ holds, for every $t\in \Gamma_a^{k+1}$,
\begin{align*}
    \frac{1}{t-1}\sum_{s=1}^{t-1} \mathbb{I}[\exists i \in M^1, b_s^i\le v^1-3]&\le \frac{1}{t-1} \left( T_a^k F_{T_a^k} + (t-1 - T_a^k) \right) \\
    \text{(since $T_a^k\le t-1 \le T_a^{k+1}$)}&\le F_{T_a^k} + \frac{T_a^{k+1} - T_a^k}{T_a^{k+1}} \\
    \text{(by Claim~\ref{claim:T_ratio_bound})} & \le F_{T_a^k} + \frac{2}{k+24NV}.
\end{align*}
Since $F_{T_a^k}\le \frac{1}{4NV}$ by Claim~\ref{claim:F_T_a_k} and $\gamma_t \le \frac{1}{12N^2V^2}$ by assumption, the above expression is further bounded by $\frac{1}{4NV} + \frac{2}{k+24NV}\le \frac{1}{4NV} + \frac{2}{24NV} = \frac{1}{3NV} \le \frac{1}{2NV} - 2\gamma_t$.
 \end{proof}

\begin{lemma}\label{lem:A_a-k-to-A_a-k+1}
For every $k\ge 0$, $\Pr[A_a^{k+1} \mid A_a^k] \ge 1 - \exp\left(-\frac{1}{2} |\Gamma_a^{k+1}| \delta_{T_a^k}^2 \right)$.
\end{lemma}
\begin{proof}
Given $A_a^k$, according to Claim~\ref{claim:A_a^k-to-t}, it holds that for every $t\in \Gamma_a^{k+1}$,  $\frac{1}{t-1}\sum_{s=1}^{t-1} \mathbb{I}[\exists i  \in M^1, b_s^i\le v^1-3] \le \frac{1}{2NV} - 2\gamma_t$.  Then according to Lemma~\ref{lemma:mean-based}, % for any history $H_{t-1}$, 
\[ \Pr[\exists i\in M^1, b_t^i\le v^1-3 \mid H_{t-1}, A_a^k] \, \le\, |M^1|V \gamma_t. \]
Let $Z_t = \mathbb{I}[\exists i\in M^1, b_t^i \le v^1-3] - |M^1|V\gamma_t$ and let $X_t=\sum_{s=T_a^k+1}^t Z_s$.  We have $\Ex{Z_t \mid H_{t-1}, A_a^k} \le 0$.  Therefore, the sequence $X_{T_a^k+1}, X_{T_a^k+2}, \ldots, X_{T_a^{k+1}}$ is a supermartingale (with respect to the sequence of history $H_{T_a^k}, H_{T_a^k+1}, \ldots, H_{T_a^{k+1}-1}$).  By Azuma's inequality, for any $\Delta > 0$, we have 
\begin{equation*}
     \Pr\Bigg[ \sum_{t \in \Gamma_a^{k+1}} Z_t \ge \Delta \bigggiven A_a^k \Bigg] \,\le\, \exp\left(-\frac{\Delta^2}{2|\Gamma_a^{k+1}|}\right). 
\end{equation*}
Let $\Delta = |\Gamma_a^{k+1}|\delta_{T_a^k}$. Then with probability at least $1-\exp\big(-\frac{1}{2}|\Gamma_a^{k+1}|\delta_{T_a^k}^2\big)$, 
we get $\sum_{t \in \Gamma_a^{k+1}} \mathbb{I}[\exists i\in M^1, b_t^{i} \le v^1-3] \;<\; \Delta + |M^1|V \sum_{t \in \Gamma_a^{k+1}} \gamma_t \; \le\; |\Gamma_a^{k+1}|\delta_{T_a^k} + |M^1|V|\Gamma_a^{k+1}| \gamma_{T_a^k}$, 
which implies
\begin{align*}
    & \frac{1}{T_a^{k+1}} \sum_{t=1}^{T_a^{k+1}} \mathbb{I}[\exists i\in M^1, b_t^{i} \le v^1-3]
    \\
    & = \frac{1}{T_a^{k+1}} \bigg(\sum_{t=1}^{T_a^k} \mathbb{I}[\exists i\in M^1, b_t^{i} \le v^1-3] + \sum_{t \in \Gamma_a^{k+1}} \mathbb{I}[\exists i\in M^1, b_t^{i} \le v^1-3]\bigg)\\
    & \le \frac{1}{T_a^{k+1}} \left(T_a^k F_{T_a^k} + |\Gamma_a^{k+1}|\delta_{T_a^k} + |M^1|V|\Gamma_a^{k+1}| \gamma_{T_a^k}\right)\\
    {\text{(by definition) }} &=F_{T_a^{k+1}}
\end{align*}
and thus $A_a^{k+1}$ holds.
 \end{proof}

Denote by $f_t^i(b)$ the frequency of bid $b$ in the first $t$ rounds for bidder $i$: $f_t^i(b)=\frac{1}{t}\sum_{s=1}^t \mathbb{I}[b_s^i=b].$
Let  $f_t^i(0:v^1-3)=\frac{1}{t}\sum_{s=1}^t \mathbb{I}[b_s^i\le v^1-3]$.
\begin{claim}\label{claim:prob-v-2}
If the history $H_{t-1}$ satisfies $f_{t-1}^i(v^1 -1)> 2(X + V \gamma_t)$ and $\frac{1}{t-1}\sum_{s=1}^{t-1} \mathbb{I}[\exists i  \in M^1, b_s^i\le v^1-3]\le X$ for some $X\in [0, 1]$, then we have $\Pr[b_t^{i'}=v^1-2\mid H_{t-1}]\le \gamma_t$ for the other $i'\neq i \in M^1$.
\end{claim}
\begin{proof}
% Given $f_{t-1}^i(v^1-1)\ge 2 (X+ V\gamma_t)$ and $\frac{1}{t-1}\sum_{s=1}^{t-1} \mathbb{I}[\exists i  \in M^1, b_s^i\le v^1-3]\le X$, we have, on the one hand, 
Consider $\alpha_{t-1}^{i'}(v^1-1)$ and $\alpha_{t-1}^{i'}(v^1-2)$. On the one hand, 
\begin{align}\label{eq:v^1-1}
    \alpha_{t-1}^{i'}(v^1-1) = 1\times (1-f_{t-1}^i(v^1-1))+\frac{1}{2}\times f_{t-1}^i(v^1-1)= 1-\frac{1}{2}f_{t-1}^i(v^1-1). 
\end{align}
On the other hand, since having more bidders with bids no larger than $v^1-2$ only decreases the utility of a bidder who bids $v^1-2$, we can upper bound $\alpha_{t-1}^{i'}(v^1-2)$ by
\begin{align}\label{eq:v^1-2}
    \alpha_{t-1}^{i'}(v^1-2) &\le  2 \times f_{t-1}^i (0:v^1 -3) + 1\times(1-f_{t-1}^i(v^1-1)-f_{t-1}^i(0:v^1-3))\nonumber\\
    &= 1 - f_{t-1}^i(v^1-1) + f_{t-1}^i(0:v^1-3)\nonumber \\
    & \le 1 - f_{t-1}^i(v^1-1) + X,
\end{align}
where the last inequality holds because $f_{t-1}^i (0:v^1 -3) \le \frac{1}{t-1}\sum_{s=1}^{t-1} \mathbb{I}[\exists i  \in M^1, b_s^i\le v^1-3]\le X$. Combining (\ref{eq:v^1-1}) and (\ref{eq:v^1-2}), we get
\begin{align*}
    \alpha_{t-1}^{i'}(v^1-1)-\alpha_{t-1}^{i'}(v^1-2) \ge (1-\frac{1}{2}f_{t-1}^i) - (1-f_{t-1}^i+X)= \frac{1}{2}f_{t-1}^i(v^1-1) - X > V\gamma_t.
\end{align*}
This implies $\Pr[b_t^{i'}=v^1-2 \mid H_{t-1}]\le \gamma_t$ according to the mean-based property.
 \end{proof}

\subsubsection{Proof of the General Case}
We consider $k=0, 1, \ldots$ to $\infty$. 
For each $k$, we suppose $A_a^0, A_a^1, \ldots, A_a^k$ hold, which happens with probability at least $1 - \sum_{s=0}^{k-1} \exp\big(-\frac{1}{2}|\Gamma_a^{s+1}|\delta_{T_a^s}^2\big)$ according to Lemma~\ref{lem:A_a-k-to-A_a-k+1}, given that $A_a^0=A_{v^1-3}$ already held.
The proof is divided into two cases based on $f_{T_a^k}^i (v^1-1)$.

\paragraph{Case 1:}  \textit{For all $k\ge 0$, $f_{T_a^k}^i(v^1-1) \le 16(F_{T_a^k}+\frac{2}{k+24NV} + V\gamma_{T_a^k})$ for both $i \in M^1$.} 

We argue that the two bidders in $M^1$ converge to playing $v^1-2$ in this case. 

According to Lemma~\ref{lem:A_a-k-to-A_a-k+1}, all events $A_a^0, A_a^1, \ldots, A_a^k, \ldots $ happen with probability at least $1 - \sum_{k=0}^\infty \exp\left(-\frac{1}{2} |\Gamma_a^{k+1}|\delta_{T_a^k}^2\right)$. Claim~\ref{claim:A_a^k-to-t} and Claim~\ref{claim:F_T_a_k} then imply that, for both $i\in M^1$,  
\[ \lim_{t\to\infty} f_t^i(0:v^1-3) \le \lim_{k\to\infty} \left( F_{T_a^k} + \frac{2}{k+24NV} \right) = 0. \]
Because for every $t \in \Gamma_a^{k+1} = [T_a^{k}+1, T_a^{k+1}]$ we have $
    f_{t}^i (v^1-1) \le \frac{T_a^{k+1}}{t} f_{T_a^k}^i(v^1-1) \le \frac{T_a^{k+1}}{T_a^k} f_{T_a^k}^i(v^1-1) \le 2f_{T_a^k}^i(v^1-1)$
and by condition $f_{T_a^k}^i(v^1-1)\to 0$ as $k\to \infty$, we have $ \lim_{t\to\infty} f_t^i(v^1-1) = 0.$
Therefore, 
$ \lim_{t\to\infty} f_t^i(v^1-2) = \lim_{t\to\infty} 1 - f_t^i(0:v^1-3) - f_t^i(v^1-1) = 1, $
which implies
\[\lim_{t \to \infty} \frac{1}{t}\sum_{s=1}^{t} \mathbb{I}[ \forall i\in M^1, b_s^i=v^1-2] = 1.\]

\paragraph{Case 2:} \textit{There exists $k \ge 0$ such that $f_{T_a^{k}}^i (v^1-1) > 16(F_{T_a^{k}}+\frac{2}{k+24NV}+V\gamma_{T_a^{k}})$ for some $i\in M^1$.}

If this case happens, we argue that the two bidders in $M^1$ converge to playing $v^1-1$. 

We first prove that, after $\ell = k+24NV$ periods (i.e., at time $T_a^{k+\ell}$), the frequency of $v^1-1$ for \emph{both} bidders in $M^1$ is greater than  $4(F_{T_a^{k+\ell}}+\frac{2}{(k+\ell)+24NV}+V\gamma_{T_a^{k+\ell}})$, with high probability.
\begin{lemma}
Suppose that, at time $T_a^k$, $A_a^k$ holds and for some $i\in M^1$, $f_{T_a^k}^i (v^1- 1)> 16(F_{T_a^k}+\frac{2}{k+24NV}+V\gamma_{T_a^k})$ holds.  Then, with probability at least $1 - 2\sum_{j=k}^{k+\ell-1} \exp\left(-\frac{1}{2}|\Gamma_a^{j+1}|\delta_{T_a^j}^2\right)$, the following events happen at time $T_a^{k+\ell}$, where $\ell = k+24NV$: 
\begin{itemize}
    \item $A_a^{k+\ell}$; 
    \item For both $i\in M^1$, $f_{T_a^{k+\ell}}^i (v^1-1)> 4(F_{T_a^{k+\ell}}+\frac{2}{(k+\ell)+24NV}+V\gamma_{T_a^{k+\ell}})$. 
\end{itemize}
\end{lemma}
\begin{proof}
We prove by an induction from $j =  k$ to $k+\ell - 1$.
Given $A_a^j$, $A_a^{j+1}$ happens with probability at least $1 - \exp\left(-\frac{1}{2}|\Gamma_a^{j+1}|\delta_{T_a^j}^2\right)$ according to Lemma~\ref{lem:A_a-k-to-A_a-k+1}.  Hence, with probability at least $1 - \sum_{j=k}^{k+\ell-1} \exp\left(-\frac{1}{2}|\Gamma_a^{j+1}|\delta_{T_a^j}^2\right)$, all events $A_a^k, A_a^{k+1}, \ldots, A_a^{k+\ell}$ happen. 

Now we consider the second event. 
% Then we consider the second event. 
% We will use an induction to prove that 
% $f_{T_a^{s}}^i (v^1-1)$. 
For all $t\in \Gamma_a^{j+1}$, noticing that $\frac{T_a^k}{t-1} \ge \frac{T_a^k}{T_a^{j+1}} \ge \frac{T_a^k}{T_a^{k+\ell}} = \frac{(k+24NV)^2}{(2(k+24NV))^2} = \frac{1}{4}$, we have 
\begin{align}
    f_{t-1}^i (v^1-1) \ge \frac{T_a^k}{t-1}f_{T_a^k}^i(v^1-1) & \ge \frac{1}{4} f_{T_a^k}^i(v^1-1) \nonumber \\
    \text{(by condition) }& > 4(F_{T_a^k} + \frac{2}{k+24NV} + V\gamma_{T_a^k})  \label{eq:f-v-1-i-eq-1} \\ 
     \text{($F_{T_a^k}$ and $\gamma_{T_a^k}$ are decreasing in $k$) }& \ge 4(F_{T_a^j} + \frac{2}{j+24NV} + V\gamma_{T_a^j}). \nonumber 
     %\\ 
    % \text{($\gamma_t$ is decreasing) } & \ge 4(F_{T_a^s} + \frac{2}{s+24NV} + V\gamma_t). 
\end{align}
According to Claim~\ref{claim:A_a^k-to-t}, given $A_a^j$ we have $\frac{1}{t-1}\sum_{s=1}^{t-1} \mathbb{I}[\exists i  \in M^1, b_s^i\le v^1-3]\le F_{T_a^j} + \frac{2}{j + 24NV} \le \frac{1}{2NV} - 2\gamma_t$.  Using Claim~\ref{claim:prob-v-2} with $X = F_{T_a^j} + \frac{2}{j + 24NV}$, we have, for bidder $i'\ne i, i'\in M^1$, $\Pr[b_t^{i'}=v^1-2\mid H_{t-1}]\le \gamma_t$.
By Lemma~\ref{lemma:mean-based}, $\Pr[b_t^{i'}\le v^1-3\mid H_{t-1}]\le (V-1)\gamma_t$.   Combining the two, we get
$\Pr[b_t^{i'}=v^1-1 \mid H_{t-1}]\ge 1-V\gamma_t.$
Let $\Delta=|\Gamma_a^{k+1}|\delta_{T_a^k}$.
Similar to the proof of Lemma~\ref{lem:A_a-k-to-A_a-k+1}, we can use Azuma's inequality to argue that, with probability at least $1-\exp(-\frac{1}{2}|\Gamma_a^{k+1}|\delta_{T_a^k}^2)$, it holds that 
\[\sum_{t\in \Gamma_a^{j+1}}\mathbb{I}[b_t^{i'}=v^1-1] 
\ge \sum_{t\in \Gamma_a^{j+1}} (1-V\gamma_t-\delta_{T_a^j})
\ge |\Gamma_a^{j+1}|(1-V\gamma_{T_a^j}-\delta_{T_a^j}).\]
An induction shows that, with probability at least $1 - \sum_{j=k}^{k+\ell-1} \exp\left(-\frac{1}{2}|\Gamma_a^{j+1}|\delta_{T_a^j}^2\right)$, $\sum_{t\in \Gamma_a^{j+1}}\mathbb{I}[b_t^{i'}=v^1-1]
\ge |\Gamma_a^{j+1}|(1-V\gamma_{T_a^j}-\delta_{T_a^j})$ holds for all $j\in\{k, \ldots, k+\ell-1\}$.  
Therefore, 
\begin{align*}
    f_{T_a^{k+\ell}}^{i'}(v^1-1) & \ge \frac{1}{T_a^{k+\ell}}\left(0 + \sum_{t\in \Gamma_a^{k+1} \cup \cdots \cup \Gamma_a^{k+\ell}} \mathbb{I}[b_t^{i'}=v^1-1]  \right) \\
    & \ge \frac{1}{T_a^{k+\ell}}\left( |\Gamma_a^{k+1}|(1-V\gamma_{T_a^k}-\delta_{T_a^k}) + \cdots + |\Gamma_a^{k+\ell}|(1-V\gamma_{T_a^{k+\ell-1}}-\delta_{T_a^{k+\ell-1}})\right) \\
    &\ge \frac{1}{T_a^{k+\ell}}\left( (|\Gamma_a^{k+1}| + \cdots + |\Gamma_a^{k+\ell}|) \cdot (1-V\gamma_{T_a^k}-\delta_{T_a^k}) \right)\\
    &= \frac{T_a^{k+\ell} - T_a^k}{T_a^{k+\ell}}(1-V\gamma_{T_a^k}-\delta_{T_a^k})\\
    & = \frac{4(k+24NV)^2 - (k+24NV)^2}{4(k+24NV)^2}(1-V\gamma_{T_a^k}-\delta_{T_a^k}) \\
    & = \frac{3}{4} (1-V\gamma_{T_a^k}-\delta_{T_a^k})\\
    & 
   \stackrel{ \text{(assuming $T_b$ is large enough) }}{>} 4\left(F_{T_a^{k+\ell}}+\frac{2}{(k+\ell)+24NV}+V\gamma_{T_a^{k+\ell}}\right). 
\end{align*}
This proves the claim for $i'\in M^1$. 
The claim for $i\in M^1$ follows from \eqref{eq:f-v-1-i-eq-1} and the fact that $F_{T_a^k}$ and $\gamma_{T_a^{k}}$ are decreasing in $k$. 
 \end{proof} 

We denote by $k_0 = k + \ell$ the time period at which $f_{T_a^{k_0}}^i(v^1-1) > 4(F_{T_a^{k_0}} + \frac{2}{k_0+24NV} + V\gamma_{T_a^{k_0}})$ for both $i\in M^1$.  
We continuing the analysis for each period $k\ge k_0$. 
Define sequence $(G_{T_a^k})$:
\[G_{T_a^k}=\frac{T_a^{k_0}}{T_a^k}\cdot 4\left(F_{T_a^{k_0}} + \frac{2}{k_0+24NV} + V\gamma_{T_a^{k_0}}\right) + \sum_{s=k_0}^{k-1}\frac{T_a^{s+1} - T_a^s}{T_a^k}(1-V\gamma_{T_a^s}-\delta_{T_a^s}), ~ \text{ for } k \ge k_0,\]
where we recall that $\delta_t=(\frac{1}{t})^{1/8}$. 
We note that $f_{T_a^{k_0}}^i(v^1-1)> G_{T_a^{k_0}} = 4\big(F_{T_a^{k_0}} + \frac{2}{k_0+24NV} + V\gamma_{T_a^{k_0}}\big)$. 
%\footnote{where is the definition of $\delta$?} 
% Recall that $T_0 = 12NV T_b$ and $T_a^0 = c^{d(v^1-3)}T_0$.
% {\color{red} In the following proof, we choose $T_b$ to be sufficiently large such that $(2c+1)V\gamma_{T_a}+2cF_{T_a}+\delta_{T_a}\le 1$.
% The following claim gives some useful properties of the two sequences.} 

\begin{claim}\label{claim:G_k-bound}
When $T_b$ is sufficiently large,
\begin{itemize}
    \item $G_{T_a^k} \ge 4\big(F_{T_a^{k_0}} + \frac{2}{k_0+24NV} + V\gamma_{T_a^{k_0}}\big)$ for every $k\ge k_0$.
    \item $\lim_{k\to\infty}G_{T_a^k}=1$.
\end{itemize} 
\end{claim}
\begin{proof}
Since $1-V \gamma_{T_a^s} - \delta_{T_a^s} \to 1$ as $T_b\to\infty$, for sufficiently large $T_b$ we have $1-V \gamma_{T_a^s} - \delta_{T_a^s} \ge 4\big(F_{T_a^{k_0}} + \frac{2}{k_0+24NV} + V\gamma_{T_a^{k_0}}\big)$ and hence $G_{T_a^k} \ge 4\big(F_{T_a^{k_0}} + \frac{2}{k_0+24NV} + V\gamma_{T_a^{k_0}}\big)$. 

Now we prove $\lim_{k\to\infty}G_{T_a^k}=1$. 
%The first term in $G_{T_a^k}$
% \[ \frac{T_a^{k_0}}{T_a^k}\cdot 4\left(F_{T_a^{k_0}} + \frac{2}{k_0+24NV} + V\gamma_{T_a^{k_0}}\right) = O\left(\frac{(k_0)^2}{k^2}\right) \to 0  \]
% as $k\to\infty$.
Consider the second term in $G_{T_a^k}$, $\sum_{s=k_0}^{k-1}\frac{T_a^{s+1} - T_a^s}{T_a^k}(1-V\gamma_{T_a^s}-\delta_{T_a^s})$. 
Since
$ 
\sum_{s=\sqrt{k}}^{k-1} \frac{T_a^{s+1} - T_a^s}{T_a^k} = \sum_{s=\sqrt k}^{k-1} \frac{2(s+24NV)+1}{(k+24NV)^2} = \frac{(k+\sqrt k + 48NV)(k-\sqrt k)}{(k+24NV)^2} \to 1
$
and $1 - V\gamma_{T_a^k} - \delta_{T_a^k} \to 1$ as $k\to \infty$, 
for any $\eps > 0$ we can always find $K\ge k_0$ such that $\sum_{s=\sqrt{k}}^{k-1} \frac{T_a^{s+1} - T_a^s}{T_a^k} \ge 1 - \eps/2$ for every $k\ge K$ and $1 - V\gamma_{T_a^s} - \delta_{T_a^s} \ge 1-\eps/2$ for every $s\ge \sqrt{k}$.  Hence,
$
    G_{T_a^k}  \ge \sum_{s=\sqrt k}^{k-1}\frac{T_a^{s+1} - T_a^s}{T_a^k}(1-V\gamma_{T_a^s}-\delta_{T_a^s}) \ge (1-\eps/2)(1-\eps/2) \ge 1-\eps.
$
In addition, $G_{T_a^k} \le 1$ when $T_b$ is sufficiently large.  Therefore $\lim_{k\to\infty} G_{T_a^k} = 1$. 
 \end{proof}

\begin{lemma}\label{lemma:4event}
Fix any $k$. 
Suppose $A_a^k$ holds and $f_{T_a^k}(v^1-1) > G_{T_a^k}$ holds for both $i\in M^1$.  Then, the following four events happen with probability at least $1 - 3\exp\left(-\frac{1}{2} |\Gamma_a^{k+1}|\delta_{T_a^k}^2\right)$: 
\begin{itemize}
    \item $A_a^{k+1}$;
    \item $f_{T_a^{k+1}}^i(v^1-1) > G_{T_a^{k+1}}$ holds for both $i\in M^1$; 
    \item $f_{t}^i(v^1-1) > (1-\frac{2}{k+24NV})G_{T_a^{k}}$ holds for both $i\in M^1$, for any $t\in \Gamma_a^{t+1}$. 
    \item $\bm x_t^i(v^1-1) = \Pr[b_t^i=v^1-1\mid H_{t-1}]\ge 1-V\gamma_t$ for both $i\in M^1$, for any $t \in \Gamma_a^{k+1}$. 
\end{itemize}
\end{lemma}

\begin{proof}
By Lemma~\ref{lem:A_a-k-to-A_a-k+1}, $A_a^{k+1}$ holds with probability at least $1 - \exp\left(-\frac{1}{2} |\Gamma_a^{k+1}|\delta_{T_a^k}^2\right)$.  Now we consider the second event. 
For every $t\in \Gamma_a^{k+1}$, we have
\begin{align}
f_{t-1}^i(v^i-1) & \ge \frac{T_a^k}{T_a^{k+1}} f_{T_a^k}(v^i-1) \nonumber \\
\text{(by condition) } & >  \frac{T_a^k}{T_a^{k+1}} G_{T_a^k}  \nonumber \\
\text{(by Claim~\ref{claim:T_ratio_bound}) } & \ge  \left(1-\frac{2}{k+24NV}\right)G_{T_a^k}  \label{eq:three-events-event-3} \\
& \ge \frac{1}{2}G_{T_a^k} \nonumber \\
\text{(by Claim~\ref{claim:G_k-bound}) } & \ge 2\left(F_{T_a^k} + \frac{2}{k+24NV} + V\gamma_{T_a^k}\right).  \nonumber 
\end{align}
In addition, according to Claim~\ref{claim:A_a^k-to-t} $A_a^k$ implies
$
\frac{1}{t-1}\sum_{s=1}^{t-1} \mathbb{I}[\exists i  \in M^1, b_s^i\le v^1-3]\le F_{T_a^k} + \frac{2}{k+24NV} \le \frac{1}{2NV} - 2\gamma_t.
$
%according to Claim \ref{claim:allfrequencysmallhappens}. 
Using Claim \ref{claim:prob-v-2} with $X=F_{T_a^k} + \frac{2}{k+24NV}$, we get $\Pr[b_t^i=v^1-2\mid H_{t-1}]\le \gamma_t$. Additionally, by Lemma~\ref{lemma:mean-based} we have $\Pr[b_t^i\le v^1-3 \mid H_{t-1}]\le (V-1)\gamma_t$. Therefore,
\begin{equation}\label{eq:mixed-strategy-v-1}
 \Pr[b_t^i=v^1-1 \mid H_{t-1}] \ge 1-V\gamma_t.
\end{equation}
Using Azuma's inequality with $\Delta=|\Gamma_a^{k+1}|\delta_{T_a^k}$, we have with probability at least $1-\exp(-\frac{1}{2}|\Gamma_a^{k+1}|\delta_{T_a^k}^2)$, 
$
v\sum_{t\in \Gamma_a^{k+1}}\mathbb{I}[b_t^i=v^1-1] 
> \sum_{t\in \Gamma_a^{k+1}} (1-V\gamma_t-\delta_{T_a^k})
\ge |\Gamma_a^{k+1}|(1-V\gamma_{T_a^k}-\delta_{T_a^k}).
$
It follows that
$
f_{T_a^{k+1}}^i(v^1-1) > \frac{1}{T_a^{k+1}}\left(T_a^{k} G_{T_a^k}+ |\Gamma_a^{k+1}|(1-V\gamma_{T_a^k}-\delta_{T_a^k})\right)=G_{T_a^{k+1}}
$
by definition. 

Using a union bound, the first event $A_a^{k+1}$ and the second event that $f_{T_a^{k+1}}^i(v^1-1) > G_{T_a^{k+1}}$ holds for both $i\in M^1$ happen with probability at least $1 - 3\exp(-\frac{1}{2}|\Gamma_a^{k+1}|\delta_{T_a^k}^2)$.  The third event is given by~\eqref{eq:three-events-event-3} and the forth event is given by~\eqref{eq:mixed-strategy-v-1}. 
 \end{proof}

We use Lemma~\ref{lemma:4event} from $k$ to $\infty$; from its third and fourth events, combined with Claim~\ref{claim:G_k-bound}, we get
$\lim_{t\to \infty} f_t^i (v^1-1) \ge \lim_{k\to \infty}\left(1-\frac{2}{k+24NV}\right)G_{T_a^k} = 1 ~\text{and}~\lim_{t\to \infty} \bm x_t^i = \bm 1_{v^1-1},$
which happens with probability at least $1-3\sum_{k=0}^{\infty}\exp(-\frac{1}{2}|\Gamma_a^{k+1}|\delta_{T_a^k}^2)$.
This concludes the analysis for Case 2. 
% By Claim~\ref{claim:new-delta-probability-sum}, this probability goes to $1$ if we let $T_b\to\infty$. 
% Furthermore, according to the proof of Lemma~\ref{lemma:second-stage-M1=3}, the probability is at least $1- 2\exp\left(-\left(\frac{ T_b}{1152N^2V^2}\right)^{\frac{1}{3}}\right)$.

Combining Case 1 and Case 2, we have that either $\lim_{t\to\infty}\frac{1}{t}\sum_{s=1}^{t} \mathbb{I}[\forall i\in M^1, b_s^i=v^1-2]=1$ happens or $\lim_{t\to\infty}\frac{1}{t}\sum_{s=1}^{t} \mathbb{I}[\forall i\in M^1, b_s^i=v^1-1]=1$ happens (in which case we also have $\lim_{t\to \infty} \bm x_t^i = \bm 1_{v^1-1}$) with overall probability at least $1 - \exp\big(-\frac{T_b}{24NV}\big) - 2\exp\big(-\frac{T_b}{1152N^2V^2} \big) - 3\sum_{k=0}^{\infty}\exp(-\frac{1}{2}|\Gamma_a^{k+1}|\delta_{T_a^k}^2)$.  Using Claim~\ref{claim:new-delta-probability-sum} concludes the proof. 

\subsubsection{The special case of $v^3=v^1-1$}
\begin{claim}\label{claim:bidder3}
Given $f_t^i(v^1-2)\ge 1-\frac{1}{4+2NV}$ for all $i\in M^1$, we have $\Pr[b_t^3=v^1-2\mid H_{t-1}]\ge 1-V\gamma_t$.
\end{claim}
\begin{proof}
If $f_t^i(v^1-2)\ge 1-\eps$, $\eps=\frac{1}{4+2NV}$, for all $i\in M^1$ then the frequency of the maximum bid to be $v^1-2$ is at least $1-2\eps$, which implies
$\alpha_{t-1}^3(v^1-2)\ge 2 \frac{1}{N} (1-2\eps).$
For any $b\le v^1-3$, 
$\alpha_{t-1}^3(b)\le V 2\eps.$
Since $\gamma_t< \frac{1}{12N^2V^2}< \frac{1}{NV}$, we have $\alpha_{t-1}^3(v^1-2)-\alpha_{t-1}^3(b)\ge 2 \frac{1}{N} (1-2\eps) - 2V\eps > V\gamma_t$, which implies, according to mean-based property, 
$\Pr[b_t^3=v^1-2]\ge 1-V\gamma_t. $
 \end{proof}

\begin{claim}\label{claim:v-2}
If history $H_{t-1}$ satisfies $f_{t-1}^i(v^1-2)\ge \frac{9}{10}$ for $i\in M^1$ and $f_{t-1}^3(v^1-2)\ge \frac{9}{10}$, then $\Pr[b_{t-1}^{i'}=v^1-2\mid H_{t-1}]\le \gamma_t$.
\end{claim}

\begin{proof}
If $f_{t-1}^i(v^1-2)\ge \frac{9}{10}$ for $i\in M^1$ and $f_{t-1}^3(v^1-2)\ge \frac{9}{10}$, then we have 
$
\frac{1}{t-1}\sum_{s=1}^{t-1}\mathbb{I}[|\{ i\notin M^1 : b_s^i = v^1-2 \}|\ge 2]\ge 1 - 2\times \frac{1}{10} = \frac{4}{5}$ and
$
P_{t-1}^{i'}(0:v^1-3)
%\le \frac{1}{t-1}\sum_{s=1}^{t-1}\mathbb{I}[b_s^j\le v^1-3, \forall j \neq i']
\le 1-f_{t-1}^3(v^1-2)\le \frac{1}{10}.$

Recall that $P_t^i(k) = \frac{1}{t} \sum_{s=1}^t \mathbb{I}[\max_{j\ne i} b_s^j = k]$. By $P_t^i(0:k)$ we mean $\sum_{\ell = 0}^k P_t^i(\ell)$.
And we can calculate
\begin{align*}
    & \alpha_{t-1}^{i'}(v^1-1)-\alpha_{t-1}^{i'}(v^1-2)\\ 
    &\ge P_{t-1}^{i'}(v^1-1)\times(\frac{1}{2}-0)+\frac{1}{t-1}\sum_{s=1}^{t-1}\mathbb{I}[|\{ i\notin M^1 : b_s^i = v^1-2 \}|\ge 2]\times(1-\frac{2}{3}) \\
    & \quad + P_{t-1}^{i'}(0:v^1-3)\times(1-2)\\
    &\ge 0+\frac{1}{3}\times\frac{4}{5}-\frac{1}{10}=\frac{1}{6}\\
    &> V\gamma_t,
\end{align*}
which implies $\Pr[b_{t-1}^{i'}=v^1-2\mid H_{t-1}]\le \gamma_t$ according to mean-based property.
 \end{proof}

%We consider each time point $T_a^k$ step by step, $k\ge 0$. We have $A_a^k$ happens with high probability according to Lemma~\ref{lemma:4event}. If $f_{T_a^k}^i(v^1-1)\ge 16(F_{T_a^k}+\frac{2}{k+24NV} + V\gamma_{T_a^k})$ for some $i\in M^1$, then go to Case 2 immediately. Else according to Claim~\ref{claim:bidder3}, bidder 3 will bid $v^1-2$ with high probability. Using Azuma's inequality, with high probability, the frequency of bidder 3 bidding $v^1-2$ in this stage will be approximately $1$, which increases $f_{T_a^{k+1}}^3(v^1-2)$.
%the frequency up to the next time point. 

%If we don't go to Case 2 and the above analysis continues to sufficiently large time points, then with high probability $A_a^k$ happens and 
%$f_{T_a^{k+1}}^3(v^1-2)$
%the frequency of bidder 3 bidding $v^1-2$ 
%is approaching $1$, which forces bidder $i\in M^1$ to bid $v^1-1$ with high probability according to mean-based property and thus we will eventually go to Case 2 after sufficiently many time points.

We only provide a proof sketch here; the formal proof is complicated but similar to the above proof for Case 2 and hence omitted.  We prove by contradiction.  Suppose Case 1 happens, that is, at each time step $T_a^k$ the frequency of $v^1-1$ for both bidders $i\in M^1$, $f_{T_a^k}^i(v^1-1)$, is upper bounded by the threshold $16(F_{T_a^k}+\frac{2}{k+24NV} + V\gamma_{T_a^k})$, which approaches $0$ as $k\to\infty$.  Assuming $A_a^0, \ldots, A_a^k$ happen (which happens with high probability), the frequency of $0:v^1-3$ is also low.  Thus, $f_t^i(v^1-2)$ must be close to $1$.  Then, according to Claim~\ref{claim:bidder3}, bidder $3$ will bid $v^1-2$ with high probability. Using Azuma's inequality, with high probability, the frequency of bidder $3$ bidding $v^1-2$ in all future periods will be approximately $1$, which increases $f_t^3(v^1-2)$ to be close to $1$ after several periods.  Then, according Claim~\ref{claim:v-2}, bidder $i\in M^1$ will switch to bid $v^1-1$.  After several periods, the frequency $f_{T_a^k}^i(v^1-1)$ will exceed $16(F_{T_a^k}+\frac{2}{k+24NV} + V\gamma_{T_a^k})$ and thus satisfy Case 2.  This leads to a contradiction.

\subsection{Proof of Proposition \ref{prop:M2counterexample}}
\label{proof:M2counterexample}
%\begin{proposition}
%When $|M^1| = 2$, there exists a mean-based algorithm \alg{} such that with positive probability, the mixed strategy profile of the players following \alg{} does not converge. 
%\end{proposition}
%\begin{proof}
We consider a simple case where there are only two bidders with the same type $v^1 = v^2 = 3$. Let $V=3$. The set of possible bids is $\mathcal B^1 = \mathcal B^2=\{0, 1, 2\}.$ Denote $f_t^i(b) = \frac{1}{t}\sum_{s=1}^t\mathbb{I}[b_s^i = b]$ the frequency of bidder $i$'s bid in the first $t$ rounds. 
\begin{claim}\label{Claim:ex-1}
For $i\in \{1,2\}$, $\alpha_t^i(1) - \alpha_t^i(2) = f_t^{3-i}(0) - \frac{f_t^{3-i}(2)}{2}$ and $\alpha_t^i(1) - \alpha_t^i(0) = f_t^{3-i}(1) + \frac{f_t^{3-i}(0)}{2}$.
\end{claim} 
\begin{proof}
We can express $\alpha_t^i(b)$ using the frequencies as the following:
   $ \alpha_t^i(0)  = \frac{ 3f_t^{3-i}(0)}{2};
    \alpha_t^i(1) = f_t^{3-i}(1) + 2f_t^{3-i}(0) = 1 + f_t^{3-i}(0) - f_t^{3-i}(2);
     \alpha_t^i(2) = \frac{f_t^{3-i}(2)}{2} + 1-f_t^{3-i}(2).
$
% Then the claim the straightforward.
Then the claim follows from direct calculation. 
%\begin{align*}
 %   \alpha_t^i(1) - \alpha_t^i(2) = f_t^{3-i}(0) - \frac{f_t^{3-i}(2)}{2}\\
  %  \alpha_t^i(1) - \alpha_t^i(0) = 1-  \frac{f_t^{3-i}(0)}{2} - f_t^{3-i}(2) = f_t^{3-i}(1) + \frac{f_t^{3-i}(0)}{2}
%\end{align*}
 \end{proof}

We construct a $\gamma_t$-mean-based algorithm $\alg$ (Algorithm~\ref{alg:ex}) with $\gamma_t = O(\frac{1}{t^{1/4}})$ such that, with constant probability, $\lim_{t\to\infty} f^i_t(1) = 1$ but in infinitely many rounds the mixed strategy $\bm x_t^i = \bm 1_2$. 
The key idea is that, when $\alpha_t^i(1) - \alpha_t^i(2)$ is positive but lower than $V\gamma_t$ in some round $t$ (which happens infinitely often), we let the algorithm bid $2$ with certainty in round $t+1$.  This does not violate the mean-based property.

We note that this algorithm has no randomness in the first $T_0$ rounds. It bids $1$ in the first $T_0 - T_0^{2/3}$ rounds and bid $0$ in the remaining $T_0^{2/3}$ rounds. Define round $T_k = 32^k T_0$ for $k \ge 0$. Let $\gamma_t = 1$ for $1 \le t \le T_0$ and $\gamma_t = T_{k}^{-1/4} = O(t^{-1/4})$ for  $t \in [T_k+1,T_{k+1}]$ and all $k \ge 0$. 

\begin{claim}
Algorithm \ref{alg:ex} is a $\gamma_t$-mean-based algorithm with $\gamma_t = O(t^{-1/4})$.
\end{claim}
\begin{proof}
We only need to verify the mean-based property in round $t \ge T_{0}+1$ since $\gamma_t = 1$ for $t \le T_0$.
The proof follows by the definition and is straightforward: 
If the condition in Line 5 holds, where $\argmax_{b} \alpha_{t-1}(b) = 1$ and $\alpha_{t-1}^i(1) - \alpha_{t-1}^i(2) \le V\gamma_t$, then the mean-based property does not apply to bids $1$ and $2$ and the algorithm bids $0$ with probability $0\le \gamma_t$. 
Otherwise, according to Line 8, the algorithm bids $b' \notin \argmax_b \alpha_{t-1}(b)$ with probability at most $T_{k+1}^{-1/3} \le \gamma_t$. 
% We have two cases. The first case is when $T_k +2 \le t \le T_{k+1}$ for some $k \ge 0$. From the definition we know that Algorithm \ref{alg:ex} bids $b' \notin \argmax_b \alpha_{t-1}(b)$ with probability at most $T_{k+1}^{-\frac{1}{3}} < \gamma_t$. The second case is when $t = T_k +1$ for some $k \ge 0$. From the definition of Algorithm \ref{alg:ex}, we know that if $ \max\alpha_{t-1}(b) - \alpha_{t-1}(b') \ge V\gamma_t$, then Algorithm \ref{alg:ex} submits $b'$ with probability at most $T_{k+1}^{-\frac{1}{3}} < \gamma_t$.
 \end{proof}

For $k \ge 0$, let $A_k$ be the event that for both $i \in \{1,2\}$, it holds that $ T_{k}^{-\frac{1}{3}} \le f_{T_k}^i(0) \le 2T_{k}^{-\frac{1}{3}}$ and $f_{T_k}^i(2) = \frac{k}{T_k}$. Since both bidders submit deterministic bids in the first $T_0$ rounds, it is easy to check that $A_0$ holds probability 1. 

The following two claims show that if $A_0, A_1 , \ldots$ all happen, then the dynamics time-average converges to $1$ while in the meantime, both of the bidders bid 2 at round $T_k+1$ for all $k \ge 0$.

\begin{claim}\label{claim:time-average}
For any $k \ge 0$ and $i \in \{1,2\}$, if $A_{k+1}$ holds, then $f_t^i(1) \ge 1- 64 T_{k+1}^{-\frac{1}{3}}- \frac{32k}{T_{k+1}}$ holds for any $t \in [T_k, T_{k+1}]$. In particular, if $A_k$ holds for all $k \ge 0$, then $\lim_{t \rightarrow \infty} f_t^i(1) = 1$ for $i \in \{1,2\}$.
\end{claim} 
\begin{proof}
Let $A_{k+1}$ holds. Then $2 T_{k+1}^{-\frac{1}{3}} \ge f_{T_{k+1}}^i(0) 
    \ge \frac{t f_{t}^i(0)}{T_{k+1}}
 \ge \frac{f_t^i(0)}{32},$ which implies that $f_{T_k}^i(0) \le 64 T_{k+1}^{-\frac{1}{3}}$. Similarly, we have $f_{t}^i(2) \le \frac{32k}{T_{k+1}}$. The claim follows by $f_t^i(1) = 1 -f_{t}^i(0) - f_t^i(2)$.
 \end{proof}

\begin{claim}\label{claim:bid2}
If $A_k$ happens, then both of the bidders bid 2 at round $T_k+1$.
\end{claim}
\begin{proof}
According to Claim \ref{Claim:ex-1}, we know that for any $i \in \{1,2\}$ and any $t > T_0$,
$
    \alpha_{t-1}^i(1) - \alpha_{t-1}^i(0) = f_{t-1}^{3-i}(1) + \frac{f_{t-1}^{3-i}(0)}{2} > 0.$
Thus $\argmax_b\{\alpha_{t-1}^i(b)\} \ne 0$ for any history $H_{t-1}$.
Again by Claim \ref{Claim:ex-1}, we have for any $i \in \{1,2\}$, $
  0 < T_k^{-\frac{1}{3}} - \frac{k}{T_k} \le \alpha_{T_k}^{i}(1) - \alpha_{T_k}^i(2) = f_{T_k}^{3-i}(0) - \frac{f_{T_k}^{3-i}(2)}{2} \le f_{T_k}^{3-i}(0) \le 2T_{k+1}^{-\frac{1}{3}} < 3T_{k+1}^{-\frac{1}{4}} = V\gamma_{T_k+1}. $
It follows from Lines 5-6 of Algorithm \ref{alg:ex} that both bidders bid 2 at round $T_k+1$.
 \end{proof}

We now bound the probability of $A_{k+1}$ given $A_k$ happens.  This will be used later to derive a constant lower bound on the probability that $A_k$ happens for all $k \ge 0$.
\begin{lemma}\label{lemma:induction}
For any $k \ge 0$,
$\Pr[A_{k+1}\mid A_k] \ge  1- 4\exp\Big( \frac{T_{k+1}^{\frac{1}{3}}}{900}\Big)$.
\end{lemma}
\begin{proof}
Suppose $A_k$ happens. We know from Claim \ref{claim:bid2} that both bidders bid 2 in round $T_k+1$. The following claim shows the behaviour of the algorithm in rounds $[T_k +2, T_{k+1}]$.

\begin{claim}\label{claim:tmp}
For any $i \in \{1,2\}$ and any $t \in [T_k+2,T_{k+1}]$, 
$\Pr[b_t^i=1\mid A_k] = 1-T_{k+1}^{-\frac{1}{3}}$, and $\Pr[b_t^i = 0\mid A_k] = T_{k+1}^{-\frac{1}{3}}$. 
\end{claim}

\begin{proof}
According to the definition of Algorithm \ref{alg:ex}, it suffices to prove that for any $t \in [T_k+2,T_{k+1}]$ and $i \in \{1,2\}$, $\argmax_b\{\alpha_{t-1}^i(b)\} = 1$ holds. 
We prove it by induction. For the base case, it is easy to verify that $\alpha_{T_k+1}^i(1) - \alpha_{T_k+1}^i(2) = f_{T_k+1}^{3-i}(0) - \frac{f_{T_k+1}^{3-i}(2)}{2} > 0,\forall i \in \{1,2\}$. Suppose the claim holds for all of the rounds $[T_k +2,t]$. Then none of the bidders bids 2 in rounds $[T_k +2,t]$. It follows that for any $i \in \{1,2\}$,
$
    \alpha_{t}^i(1) - \alpha_{t}^i(2) = f_{t}^{3-i}(0) - \frac{f_{t}^{3-i}(2)}{2}
    \ge \frac{f_{T_k}^{3-i}(0)}{32} - \frac{k+1}{2T_k}
    \ge \frac{1}{32T_k^{\frac{1}{3}}} - \frac{k+1}{T_k} 
    > 0 \text{ (since $T_0 > 64^{\frac{3}{2}}$)}.
$
Therefore $\argmax_b\{\alpha_{t-1}^i(b)\} = 1$. This completes the induction step. 
 \end{proof}

From the above proof we can also conclude that for $i \in \{1,2\}$, $f_{T_{k+1}}^i(2) = \frac{k+1}{T_{k+1}}$.

Note that the bidding strategies of both bidders at different rounds in $[T_k+2, T_{k+1}]$ are independent. Specifically, we have $\Pr[b^i_s = 0] = T_{k+1}^{-\frac{1}{3}}$ for $i \in \{1,2\}$ and $s \in [T_{k}+2, T_{k+1}]$. Therefore, by Chernoff bound, we have for $i \in \{1,2\}$, the total number of bids $0$ between rounds $[T_{k}+2, T_{k+1}]$ lies between $[\frac{29}{30}(T_{k+1}-T_k-1)T_{k+1}^{-\frac{1}{3}}, \frac{31}{30}(T_{k+1}-T_k-1)T_{k+1}^{-\frac{1}{3}}]$ with probability at least $1- 2\exp\left( -\frac{T_{k+1}-T_k-1}{450T_{k+1}^{\frac{2}{3}}} \right)  \ge 1- 2\exp\left(- \frac{T_{k+1}^{\frac{1}{3}}}{900}\right)$.
% \begin{align*}
%     & \Pr\Bigg[\frac{29}{30}\frac{T_{k+1}-T_k-1}{T_{k+1}^{\frac{1}{3}}} \le \sum_{s=T_k+2}^{T_{k+1}} \bm 1[b_s^i = 0] \le \frac{31}{30}\frac{T_{k+1}-T_k-1}{T_{k+1}^{\frac{1}{3}}} \bigggiven A_k \Bigg] \\
%     & \ge 1- 2\exp\left( -\frac{T_{k+1}-T_k-1}{450T_{k+1}^{\frac{2}{3}}} \right)\\
%     & \ge 1- 2\exp\left(- \frac{T_{k+1}^{\frac{1}{3}}}{900}\right).
% \end{align*}
Therefore, with probability at least $1- 4\exp\left(- \frac{T_{k+1}^{\frac{1}{3}}}{900}\right)$, both of the above events happens. It implies that for $i \in \{1,2\}$, the frequency of bids $0$ is at least
\begin{align*}
    f_{T_{k+1}}^{i}(0) &\ge \frac{1}{T_{k+1}}\left( T_k f_{T_k}^i(0) + \frac{29}{30}\frac{T_{k+1}-T_k-1}{T_{k+1}^{\frac{1}{3}}} \right)\\
    & \ge \frac{1}{T_{k+1}}\left( \frac{T_k}{T_k^{\frac{1}{3}}} + \frac{29}{30}\frac{\frac{30}{32}T_{k+1}}{T_{k+1}^{\frac{1}{3}}} \right)\\
    & = \frac{32^{\frac{1}{3}}}{32T_{k+1}^{\frac{1}{3}}} + \frac{29}{32 T_{k+1}^{\frac{1}{3}}} \ge \frac{1}{T_{k+1}^{\frac{1}{3}}},
\end{align*}
and the frequency of bids $0$ is at most
\begin{align*}
    f_{T_{k+1}}^{i}(0) &\le \frac{1}{T_{k+1}}\left( T_k f_{T_k}^i(0) + \frac{31}{30}\frac{T_{k+1}-T_k-1}{T_{k+1}^{\frac{1}{3}}} \right)\\
    & \le \frac{1}{T_{k+1}}\left( \frac{2T_k}{T_k^{\frac{1}{3}}} + \frac{31}{30}\frac{T_{k+1}}{T_{k+1}^{\frac{1}{3}}} \right)\\
    & = \frac{2\times 32^{\frac{1}{3}}}{32T_{k+1}^{\frac{1}{3}}} + \frac{31}{30 T_{k+1}^{\frac{1}{3}}} \le \frac{2}{T_{k+1}^{\frac{1}{3}}}.
\end{align*}
Therefore, $A_{k+1}$ holds. This completes the proof of Lemma~\ref{lemma:induction}.
 \end{proof}
Using a union bound, we have $\Pr[\forall k \ge 0, A_k \text{ holds}] = \Pr[A_0]\prod_{k=0}^{\infty}\Pr[A_{k+1} \mid A_k]$ is at least $1 - 4\sum_{j=1}^\infty \exp(- T_{j}^{\frac{1}{3}} / 900 )$. Since $T_0 = 10^{12}$, $\exp(-T_0^{1/3}/900) \le \frac{1}{16}$ and $T_j = 32^j T_0$, we can further lower bound the probability by $1 - \frac{1}{4}\sum_{j=1}^\infty \exp\left(- 32^{j/3}\right) \ge \frac{1}{2}$. 
% \begin{align*}
%     % & \Pr[\forall k \ge 0, A_k \text{ holds}] \\
%     % &\ge \Pr[A_0]\prod_{k=0}^{\infty}\Pr[A_{k+1} \mid A_k]\\
%     &1 - 4\sum_{j=1}^\infty \exp\left(- \frac{T_{j}^{\frac{1}{3}}}{900}\right)\\
%     & \ge 1 - 4\sum_{j=1}^\infty \exp\left(- \frac{T_{0}^{\frac{1}{3}}3^j}{900}\right) \\
%     & \ge 1 - \frac{1}{4}\sum_{j=1}^\infty \exp\left(- 3^j\right) \\
%     %& = 1 - 4\exp\left(- \frac{T_{0}^{\frac{1}{3}}}{300}\right)\left(1 + \sum_{j=2}^\infty \exp\left(- \frac{T_{0}^{\frac{1}{3}}(3^j-3)}{900}\right)\right)\\
%     %&\ge 1 - 8\exp\left(- \frac{T_{0}^{\frac{1}{3}}}{300}\right)\\
%     &\ge \frac{1}{2},
% \end{align*}
%where in the second last inequality we use the fact $T_0 = 10^{12}$ and $\exp(-\frac{T_0^{1/3}}{900}) \le \frac{1}{16}$.
Therefore, with probability at least $\frac{1}{2}$, $A_k$ holds for all $k \ge 0$. By Claim~\ref{claim:time-average}, the dynamics time-average converges to the equilibrium of $1$, but by Claim~\ref{claim:bid2}, both bidders' mixed strategies do not converge in the last-iterate sense. This completes the proof of Proposition \ref{prop:M2counterexample}.

%% =======================================================================
%% =======================================================================

\section{Missing Proofs from Section~\ref{sec:main-proof}}
\label{app:main-proof}

\subsection{Proof of Lemma~\ref{lemma:mean-based-v-2}}
\label{app:mean-based-v-2}
Let $\Gamma=\{s\leq t-1 | \exists i\in M^1, b^i_s \leq v^1-3\}$.
The premise of the lemma says $\frac{|\Gamma|}{t-1} \le \frac{1}{3NV}$.  First, note that 
\begin{align}\label{eq:M1-3-P-v-3}
P_{t-1}^i(0:v^1-3) & = \frac{1}{t-1}\sum_{s=1}^{t-1} \mathbb{I}[\max_{i'\ne i} b_s^{i'} \le v^1-3] \nonumber \\
& \le \frac{1}{t-1}\sum_{s=1}^{t-1} \mathbb{I}[\exists i\in M^1, b_s^{i} \le v^1-3]= \frac{|\Gamma|}{t-1} \le \frac{1}{3NV}.
\end{align}
Then, according to \eqref{eq:alpha-P-Q}, 
\begin{align}\label{eq:M1-3-alpha}
     & \alpha_{t-1}^i(v^1-1) - \alpha_{t-1}^i(v^1-2)  \nonumber \\
     & = Q_{t-1}^i(v^1-1) + P_{t-1}^i(v^1-2)  - 2Q_{t-1}^i(v^1-2) - P_{t-1}^i(0:v^1-3). 
\end{align}
Using $Q_{t-1}^i(v^1-1) \ge \frac{1}{N} P_{t-1}^i(v^1-1)$ and $Q_{t-1}^i(v^1-2)\le \frac{1}{2} P_{t-1}^i(v^1-2)$ from \eqref{eq:P-Q}, we can lower bound \eqref{eq:M1-3-alpha} by
$ 
\frac{1}{N} P_{t-1}^i(v^1-1) - P_{t-1}^i(0:v^1-3). 
$
With \eqref{eq:M1-3-P-v-3}, we get
$
\alpha_{t-1}^i(v^1-1) - \alpha_{t-1}^i(v^1-2) \ge \frac{1}{N} P_{t-1}^i(v^1-1) - \frac{1}{3NV}.
$

If $\frac{1}{N} P_{t-1}^i(v^1-1) - \frac{1}{3NV} > V\gamma_t$, then $\alpha_{t-1}^i(v^1-1) - \alpha_{t-1}^i(v^1-2) > V\gamma_t$.  By the mean-based property, $\Pr[b_t^i = v^1-2 \mid H_{t-1}] \le \gamma_t$. 

Suppose $\frac{1}{N} P_{t-1}^i(v^1-1) - \frac{1}{3NV} \le V\gamma_t$, which is equivalent to
$ 
P_{t-1}^i(v^1-1) \le \frac{1}{3V} + NV\gamma_t.
$
Consider $Q_{t-1}^i(v^1-2)$.  By the definition of $\Gamma$, in all rounds $s\notin \Gamma$ and $s\le t-1$, all bidders in $M^1$ bid $v^1-2$ or $v^1-1$. 
If bidder $i$ wins with bid $v^1-2$ in round $s\notin \Gamma$, she must be tied with at least two other bidders in $M^1$ since $|M^1|\ge 3$; if bidder $i$ wins with bid $v^1-2$ (tied with at least one other bidder) in round $s\in \Gamma$, that round contributes at most $\frac{1}{2}$ to the summation in $Q_{t-1}^i(v^1-2)$.  Therefore, 
\begin{align}\label{eq:M1-3-Q}
    Q_{t-1}^i(v^1-2) \le \frac{1}{t-1} \left( \frac{(t-1) - |\Gamma|}{3} + \frac{|\Gamma|}{2} \right)  = \frac{1}{3} + \frac{1}{6}\frac{|\Gamma|}{t-1} \le \frac{1}{3} + \frac{1}{18 NV}. 
\end{align}
We then consider $P_{t-1}^i(v^1-2)$. Since $P_{t-1}^i(0:v^1-3) + P_{t-1}^i(v^1-2) + P_{t-1}^i(v^1-1) = 1$, and recalling that $P_{t-1}^i(0:v^1-3)\le \frac{1}{3NV}$ and $P_{t-1}^i(v^1-1) \le \frac{1}{3V} + NV\gamma_t$, we get
\begin{align}\label{eq:M1-3-P}
    P_{t-1}^i(v^1-2) = 1 - P_{t-1}^i(0:v^1-3) - P_{t-1}^i(v^1-1) \ge 1 - \frac{1}{3NV} - \frac{1}{3V} - NV\gamma_t. 
\end{align}
Combining \eqref{eq:M1-3-alpha} with \eqref{eq:M1-3-P-v-3}, \eqref{eq:M1-3-Q}, and \eqref{eq:M1-3-P}, we get 
\begin{align*}
     & \alpha_{t-1}^i (v^1-1) - \alpha_{t-1}^i(v^1-2) \\
    & \ge 0 + \left(1 - \frac{1}{3NV} - \frac{1}{3V} - NV\gamma_t\right) - 2\left(\frac{1}{3} + \frac{1}{18 NV}\right) - \frac{1}{3NV} \\
    & = \frac{1}{3} - \frac{3N+7}{9NV}  - NV\gamma_t \\
    & \ge \frac{1}{3} - \frac{3N+7}{9NV} - \frac{1}{12} \quad \text{(because $\gamma_t\le \frac{1}{12NV}$)} \\
    & \ge \frac{1}{4} - \frac{3N+7}{27 N} \quad \text{(because $V\ge 3$)} \\
    & \ge \frac{1}{12N} \quad \text{(because $N\ge 3$)} \\
    & \ge V \gamma_t \quad \text{(because $\gamma_t\le \frac{1}{12NV}$)}.
\end{align*}
Therefore, by the mean-based property, $\Pr[b_t^i=v^1-2 \mid H_{t-1}] \le \gamma_t$.

\subsection{Proof of Claim \ref{claim:F_k-bound}}
Since $\delta_{T_a^0} \to 0$ and $\gamma_{T_a^0} \to 0$ as $T_b \to \infty$, when $T_b$ is sufficiently large we have 
\[F_{T_a^1} = \frac{1}{c}\frac{1}{4NV} + \frac{c-1}{c} \left(\delta_{T_a^0} + |M^1|V \gamma_{T_a^0}\right) \le \frac{1}{c}\frac{1}{4NV} + \frac{c-1}{c} \frac{1}{4NV} \le \frac{1}{4NV} = F_{T_a^0}.\]
By definition, for every $k\ge 1$
\[F_{T_a^{k+1}}=\frac{1}{c}F_{T_a^k}+ \frac{c-1}{c}\left(\delta_{T_a^k}+|M^1|V\gamma_{T_a^k}\right), \quad F_{T_a^{k}}=\frac{1}{c}F_{T_a^{k-1}}+ \frac{c-1}{c}\left(\delta_{T_a^{k-1}}+|M^1|V\gamma_{T_a^{k-1}}\right). \]
Using the fact that $F_{T_a^k} \le F_{T_a^{k-1}}$ and that $\delta_{T_a^k}+|M^1|V\gamma_{T_a^k}$ is decreasing in $k$, we have $F_{T_a^{k+1}}\le F_{T_a^{k}} \le \frac{1}{4NV}$. Similarly, we have $\widetilde{F}_{T_a^{k+1}} \le \widetilde{F}_{T_a^k}$ for any $k \ge 0$.

Note that $\delta_{T_a^k} \to 0$ and $\gamma_{T_a^0} \to 0$ as $k \to +\infty$. Therefore, for any $0< \eps \le \frac{1}{4NV}$, we can find $k$ sufficiently large such that $\frac{1}{c^{k/2}} \le \frac{\eps}{6}$, $\delta_{T_a^s} \le \frac{\eps}{6}$, and $\gamma_{T_a^s} \le \frac{\eps}{6|M^1|V}$. Then we have
\begin{align*}
   F_{T_a^k} \le \widetilde{F}_{T_a^k} &= \frac{1}{c^k} +\sum_{s=0}^{k-1} \frac{c-1}{c^{k-s}}\delta_{T_a^s}+ \sum_{s=0}^{k-1}|M^1|V \frac{c-1}{c^{k-s}} \gamma_{T_a^s} \\
   & \le \frac{\eps}{3} + 2\sum_{s=0}^{k/2-1}\frac{c-1}{c^{k-s}} + \sum_{s=k/2}^{k-1}\frac{c-1}{c^{k-s}} (\delta_{T_a^{k/2}}+ |M^1|V \frac{c-1}{c^{k-s}} \gamma_{T_a^{k/2}})\\
   & \le \frac{\eps}{3} + 2\frac{1}{c^{k/2}} + \frac{\eps}{3} \sum_{s=k/2}^{k-1}\frac{c-1}{c^{k-s}} \\
   & \le \frac{\eps}{3} + \frac{\eps}{3} + \frac{\eps}{3} = \eps.
\end{align*}
Thus for any $l \ge k$, we have $F_{T_a^l} \le \widetilde{F}_{T_a^l} \le \eps$. Since $F_{T_a^k}$ and $\widetilde{F}_{T_a^k}$ are both positive, we have $\lim_{k\to \infty}F_{T_a^k} =\lim_{k\to \infty} \widetilde{F}_{T_a^k} =0$. 

\subsection{Proof of Lemma \ref{lemma:second-stage-M1=3}}
\label{proof:lemma:second-stage-M1=3}
We will use an induction to prove the following: 
\begin{equation*}
    \Pr[A_a^{k+1}] \ge 1 - \exp\left(-\frac{T_b}{24NV}\right) - 2\exp\left(-\frac{T_b}{1152N^2V^2} \right) - \sum_{s=0}^{k} \exp\left(-\frac{1}{2}|\Gamma_a^{s+1}|\delta_{T_a^s}^2\right).
\end{equation*}
We do not assume $|M^1|\ge 3$ for now. 
The base case follows from Corollary \ref{lemma:first-stage} because $A_a^0$ is the same as $A_{v^1-3}$.
Suppose $A_a^k$ happens. Consider $A_a^{k+1}$. 
For any round $t\in \Gamma_a^{k+1}$,
\begin{align*}
    P_{t-1}^i(0:v^1-3) &\le \frac{1}{t-1} \sum_{s=1}^{t-1} \mathbb{I}[\exists i\in M^1, b_s^i \le v^1-3] \\
    &  = \frac{1}{t-1} \bigg( \sum_{s=1}^{T_a^k} \mathbb{I}[\exists i\in M^1, b_s^i \le v^1-3] + \sum_{s=T_a^k +1}^{t-1} \mathbb{I}[\exists i\in M^1, b_s^i \le v^1-3] \bigg) \\
     {\text{(by $F_{T_a^k}\le \frac{1}{4NV}$) }} & \le  \frac{1}{t-1} \left(\frac{T_a^k}{4NV} + (t-1-T_a^k) \right) \\
     {\text{(by $T_a^k \le t-1\le T_a^{k+1}$) }} & \le \frac{1}{T_a^{k}} \left( \frac{T_a^k}{4N V} + T_a^{k+1} -T_a^k \right) \\
     {\text{(by $T_a^{k+1} = c T_a^k$) }} & = \frac{1}{3N V}. % \stackrel{\text{($\gamma_t < \frac{1}{12N V}$)}}{<} \frac{1}{2NV} - 2\gamma_t.
\end{align*}
By Lemma \ref{lemma:mean-based}, %and a similar analysis to Claim \ref{claim:induction-B},
for any history $H_{t-1}$ that satisfies $A_a^k$, we have 
\begin{equation}\label{eq:M1-V-gamma_t}
\Pr[\exists i \in M^1, b_t^i\le v^1-3 \given H_{t-1}, A_a^k]\le |M^1|V\gamma_t.
\end{equation}
Let $Z_t = \mathbb{I}[\exists i\in M^1, b_t^i \le v^1-3] - |M^1|V\gamma_t$ and let $X_t=\sum_{s=T_a^k+1}^t Z_s$.  We have $\Ex{Z_t \mid A_a^k, H_{t-1}} \le 0$.  Therefore, the sequence $X_{T_a^k+1}, X_{T_a^k+2}, \ldots, X_{T_a^{k+1}}$ is a supermartingale (with respect to the sequence of history $H_{T_a^k}, H_{T_a^k+1}, \ldots, H_{T_a^{k+1}-1}$).  By Azuma's inequality, for any $\Delta > 0$, we have 
\begin{equation*}
     \Prx{\sum_{t \in \Gamma_a^{k+1}} Z_t \ge \Delta \Biggiven A_a^k} \le \exp\left(-\frac{\Delta^2}{2|\Gamma_a^{k+1}|}\right). 
\end{equation*}
Let $\Delta = |\Gamma_a^{k+1}|\delta_{T_a^k}$. Then with probability at least $1-\exp\big(-\frac{1}{2}|\Gamma_a^{k+1}|\delta_{T_a^k}^2\big)$, 
we have 
\begin{align}\label{eq:second-stage-azuma-consequence-M=3}
\sum_{t \in \Gamma_a^{k+1}} \mathbb{I}[\exists i\in M^1, b_t^{i} \le v^1-3] & \;<\; \Delta + |M^1|V \sum_{t \in \Gamma_a^{k+1}} \gamma_t \nonumber \\
& \; \le\; |\Gamma_a^{k+1}|\delta_{T_a^k} + |M^1|V|\Gamma_a^{k+1}| \gamma_{T_a^k}, 
\end{align}
which implies
\begin{align*}
    \frac{1}{T_a^{k+1}} & \sum_{t=1}^{T_a^{k+1}} \mathbb{I}[\exists i\in M^1, b_t^{i} \le v^1-3] \\
    & = \frac{1}{T_a^{k+1}} \bigg(\sum_{t=1}^{T_a^k} \mathbb{I}[\exists i\in M^1, b_t^{i} \le v^1-3] + \sum_{t \in \Gamma_a^{k+1}} \mathbb{I}[\exists i\in M^1, b_t^{i} \le v^1-3]\bigg)\\
    & \le \frac{1}{T_a^{k+1}} \left(T_a^k F_{T_a^k} + |\Gamma_a^{k+1}|\delta_{T_a^k} + |M^1|V|\Gamma_a^{k+1}| \gamma_{T_a^k}\right)\\
    {\text{(since $T_a^{k+1}=c T_a^k$) }} &=\frac{1}{c}F_{T_a^k}+ \frac{c-1}{c}\delta_{T_a^k}+|M^1|V\frac{c-1}{c}\gamma_{T_a^k}\\
    {\text{(by definition) }} &=F_{T_a^{k+1}}
\end{align*}
and thus $A_a^{k+1}$ holds.

Now we suppose $|M^1|\ge 3$. We can change \eqref{eq:M1-V-gamma_t} to $ \Pr[\exists i \in M^1, b_t^i\le v^1-2 \given H_{t-1}, A_a^k]\le |M^1|V\gamma_t$
due to Lemma~\ref{lemma:mean-based-v-2} and the fact that $\frac{1}{t-1} \sum_{s=1}^{t-1} \mathbb{I}[\exists i\in M^1, b_s^i \le v^1-3] \le \frac{1}{3NV}$.  The definition of $Z_t$ is changed accordingly, and \eqref{eq:second-stage-azuma-consequence-M=3} becomes
\begin{align*}
\hspace{-2em} \sum_{t \in \Gamma_a^{k+1}} \mathbb{I}[\exists i\in M^1, b_t^{i} \le v^1-2] <  |\Gamma_a^{k+1}|\delta_{T_a^k} + |M^1|V|\Gamma_a^{k+1}| \gamma_{T_a^k},
\end{align*}
which implies 
\begin{align*}
    \frac{1}{T_a^{k+1}} \sum_{t=1}^{T_a^{k+1}} \mathbb{I}[\exists i\in M^1, b_t^{i} \le v^1-2] \le \frac{1}{T_a^{k+1}} \left(T_a^k \widetilde{F}_k + |\Gamma_a^{k+1}|\delta_{T_a^k} + |M^1|V|\Gamma_a^{k+1}| \gamma_{T_a^k}\right) = \widetilde{F}_{T_a^{k+1}}. 
\end{align*}
To conclude, by induction, 
\begin{align*}
    \Pr[A_a^{k+1}] & =\Pr[A_a^k] \Pr[A_a^{k+1}|A_a^k] \\
    & \ge \Pr[A_a^k]-\exp\left(-\frac{1}{2}|\Gamma_a^{k+1}|\delta_{T_a^k}^2\right)\\
    &\ge 1 - \exp\left(-\frac{T_b}{24NV}\right) - 2\exp\left(-\frac{T_b}{1152N^2V^2} \right) - \sum_{s=0}^{k} \exp\left(-\frac{1}{2}|\Gamma_a^{s+1}|\delta_{T_a^s}^2\right). 
\end{align*}

%\subsection{Proof of Corollary \ref{cor:second-stage-end}}
As $\delta_t=(\frac{1}{t})^{\frac{1}{3}}$ and $|\Gamma_a^{s}|=c^{s+d(v^1-3)-1}(c-1)T_0$, $T_a^s=c^{s+d(v^1-3)}T_0$ (let $v^1-3=0$ if $v^1 < 3$), we have
\begin{align*}
     & \sum_{s=0}^{k} \exp\left(-\frac{1}{2}|\Gamma_a^{s+1}|\delta_{T_a^s}^2\right) \\
     &=\sum_{s=0}^{k} \exp\left(-\frac{1}{2}c^{\frac{1}{3}(s+d(v^1-3))}(c-1)(T_0)^{\frac{1}{3}}\right)\\
     &=\exp\left(-\frac{1}{2}c^{\frac{1}{3}d(v^1-3)}(c-1)(T_0)^{\frac{1}{3}}\right)\left(1+\sum_{s=1}^{k} \exp\left(-\frac{1}{2}c^{\frac{1}{3}d(v^1-3)}(c-1)(T_0)^{\frac{1}{3}}(c^{\frac{s}{3}}-1)\right)\right)\\
     &\le \exp\left(-\frac{1}{2}c^{\frac{1}{3}d(v^1-3)}(c-1)(T_0)^{\frac{1}{3}}\right)\left(1+\sum_{s=1}^{k} \exp\left(-\frac{1}{2}c^{\frac{1}{3}d(v^1-3)}(c-1)(T_0)^{\frac{1}{3}}s(c^{\frac{1}{3}}-1)\right)\right)\\
     &\le \exp\left(-\frac{1}{2}c^{\frac{1}{3}d(v^1-3)}(c-1)(T_0)^{\frac{1}{3}}\right)\left(1+\sum_{s=1}^{k} (\frac{1}{2})^s\right)\\
     &\le 2 \exp\left(-\frac{1}{2}c^{\frac{1}{3}d(v^1-3)}(c-1)(T_0)^{\frac{1}{3}}\right), 
\end{align*}
where in the last but one inequality we suppose that $T_0$ is large enough so that $\exp\big(-\frac{1}{2} c^{\frac{1}{3}d(v^1-3)}(c-1)(T_0)^{\frac{1}{3}}s(c^{\frac{1}{3}}-1)\big) \le \frac{1}{2}$.  Substituting $T_0 = 12NV T_b = \frac{1}{c-1} T_b$, $c=1+\frac{1}{12NV}$, and $c^d=8NV$ gives
\begin{align*}
     \sum_{s=0}^{k} \exp\left(-\frac{1}{2}|\Gamma_a^{s+1}|\delta_{T_a^s}^2\right) & \le 2\exp\left(-\left(\frac{(8NV)^{(v^1-3)} T_b}{1152N^2V^2}\right)^{\frac{1}{3}}\right)\\
     &\le 2\exp\left(-\left(\frac{ T_b}{1152N^2V^2}\right)^{\frac{1}{3}}\right),
\end{align*}
concluding the proof.

\end{document}